\documentclass[letterpaper,10pt]{amsart}
\usepackage{indentfirst} 
\usepackage{amssymb}
\usepackage{mathrsfs}
\usepackage{amsmath}
\usepackage{amsthm}
\usepackage{thmtools}
\usepackage{enumitem}
\usepackage{bbm}             
\usepackage[usenames,dvipsnames]{xcolor} 
\usepackage{tikz}
\usepackage[left=2.5cm, right=2.5cm, bottom=3.4cm]{geometry}
\usepackage[colorlinks = true, linkcolor = black, citecolor = blue ]{hyperref}

\theoremstyle{plain}
\newtheorem{theorem}{Theorem}[section]
\newtheorem{proposition}{Proposition}[subsection]
\newtheorem{lemma}[proposition]{Lemma}
\newtheorem{corollary}[proposition]{Corollary}
\newtheorem{definition}[proposition]{Definition}

\theoremstyle{remark}
\newtheorem{remark}{Remark}[subsection]

\DeclareMathOperator\supp{supp}

\def\R{{\mathbb R}}
\def\e{{\epsilon}}

\def\P{{\mathcal P}}

\def\S{{\mathbb S}}
\def\T{{\mathrm T}}
\def\M{\mathcal{M}}
\def\D{\mathcal{D}}
\def\B{\mathcal{B}}

\def\d{\mathrm{d}}
\def\sgn{\mathrm{sgn}}

\begin{document}

\title[Decay properties for massive Vlasov fields on Schwarzschild spacetime]{Decay properties for massive Vlasov fields\\ on Schwarzschild spacetime}

\author[Renato Velozo Ruiz]{Renato Velozo Ruiz} 
\address{Department of Mathematics, University of Toronto, 40 St. George Street, Toronto, ON, Canada.}
\email{renato.velozo.ruiz@utoronto.ca}

\begin{abstract}
In this paper, we obtain pointwise decay estimates in time for massive Vlasov fields on the exterior of Schwarzschild spacetime. We consider massive Vlasov fields supported on the closure of the largest domain of the mass-shell where timelike geodesics either cross $\mathcal{H}^+$, or escape to infinity. For this class of Vlasov fields, we prove that the components of the energy-momentum tensor decay like $v^{-\frac{1}{3}}$ in the bounded region $\{r\leq R\}$, and like $u^{-\frac{1}{3}}r^{-2}$ in the far-away region $\{r\geq R\}$, where $R>2M$ is sufficiently large. Here, $(u,v)$ denotes the standard Eddington--Finkelstein double null coordinate pair.
\end{abstract}


\maketitle

\setcounter{tocdepth}{1}
\tableofcontents

\section{Introduction} 

The Schwarzschild family of black holes is a one-parameter family of spherically symmetric spacetimes, parametrised by mass $M>0$, that satisfy the \emph{Einstein vacuum equations} 
\begin{equation}\label{Eins_vac_eqn_intro}
\mathrm{Ric}_{\mu\nu}[g]=0.
\end{equation}
In this article, we investigate the decay properties of massive collisionless systems on the exterior $(\mathcal{E},g_M)$ of a Schwarzschild black hole background. We consider collisionless systems described statistically by a distribution function satisfying a transport equation along the timelike geodesic flow. Specifically, we study the solutions $f(x,p)$ of the \emph{massive Vlasov equation on Schwarzschild spacetime}  
\begin{equation}\label{vlasov_eqn_massive_intro}
\mathbb{X}_{g_M}f=0,
\end{equation}
in terms of the generator of the timelike geodesic flow. The distribution function $f\colon \P\to [0,\infty)$ is a real-valued function defined on the \emph{mass-shell}   
\begin{equation}\label{def_first_intro_mass_shell}
\P:=\Big\{(x,p)\in T\mathcal{E}: g_x(p,p)=-1, \text{ where $p$ is future directed}\Big\}.
\end{equation}
A distribution function $f$ is a \emph{massive Vlasov field on Schwarzschild} if it satisfies \eqref{vlasov_eqn_massive_intro}. We assumed in the definition of $\P$, that the rest mass of the particles we consider is normalised to be one. The massive Vlasov equation \eqref{vlasov_eqn_massive_intro} describes the evolution in time of systems composed by free falling particles. The study of massive Vlasov fields on Schwarzschild spacetime is motivated by the research of self-gravitating massive collisionless systems in general relativity. See \cite{GLW,A11, AGS} for more information about relativistic kinetic theory.

In the framework of general relativity, self-gravitating massive collisionless systems are described by the solutions $(\M,g,f)$ of the \emph{Einstein--massive Vlasov system}
\begin{align}
\begin{aligned}\label{EV_intro}
\mathrm{Ric}_{\mu\nu}[g]-\frac{1}{2}\mathrm{R}[g]\cdot g_{\mu\nu}&=8\pi \mathrm{T}_{\mu\nu}[f],\\
\mathbb{X}_{g}f&=0,
\end{aligned}
\end{align}
in terms of the generator of the timelike geodesic flow. Here, $\mathrm{T}_{\mu\nu}[f]$ denotes the components of the energy-momentum tensor of the Vlasov field $f$. The non-linear PDE system \eqref{EV_intro} reduces to the Einstein vacuum equations, when the distribution function vanishes. In particular, the members of the Schwarzschild family satisfy \eqref{EV_intro}. A prominent problem in relativistic kinetic theory consists in describing the behaviour of self-gravitating massive collisionless systems near the Schwarzschild geometry. In this paper, we investigate the decay properties of massive Vlasov fields on the exterior of a Schwarzschild background, in view of the intimate relation of this problem with the stability properties of \eqref{EV_intro} near Schwarzschild. See \cite{A11,AR08, CB09, R} for more information about the Einstein--massive Vlasov system.

The massive Vlasov equation on a Lorentzian manifold $(\M,g)$ is a transport equation along the timelike geodesic flow. For this reason, the linear dynamics of massive Vlasov fields on $(\M,g)$ depend strongly on the particular form of the geodesic flow in this background. In this article, we will study \emph{dispersive Vlasov fields} $f$, in the sense that decay estimates in time hold for the associated energy-momentum tensor $\T_{\mu\nu}[f]$. The energy-momentum tensor of a distribution function $f$ is defined by
\begin{equation}\label{energy_momentum_local_geometry_intro}
\T_{\mu\nu}[f]:=\int_{\P_x} f(x,p)p_{\mu}p_{\nu}\mathrm{d\mu}_{\mathcal{P}_x},
\end{equation}
in terms of the induced volume form on the fibers $\P_x$ of the mass-shell. The covectors in \eqref{energy_momentum_local_geometry_intro} are defined by $p_{\mu}=g_{\mu\nu}p^{\nu}$. We often write the components of the energy-momentum tensor by $\T_{\mu\nu}$, without making reference to the corresponding Vlasov field. 

In a Schwarzschild black hole, massive Vlasov fields are in general \emph{not dispersive}. The massive Vlasov equation \eqref{vlasov_eqn_massive_intro} admits a large class of non-trivial stationary solutions, which is an obstruction to decay in time. This is clear from the existence of bounded orbits that do not cross $\mathcal{H}^+$. These orbits are called \emph{bound orbits}. Localised stationary solutions of the massive Vlasov equation are contained in the closure $\mathcal{B}$ of the subset of the mass-shell where orbits are bound. We circumvent the obstruction to decay posed by these non-trivial stationary states, by simply considering massive Vlasov fields supported in the closure $\mathcal{D}$ of the complementary region of the mass-shell. The set $\mathcal{D}$ can be defined as the closure of the largest domain of the mass-shell, where timelike geodesics cross $\mathcal{H}^+$, or escape to infinity towards the future. The set $\mathcal{D}$ is the largest region of phase space where the problem of studying decay estimates for the massive Vlasov equation \eqref{vlasov_eqn_massive_intro} makes sense. However, even for compactly supported initial data, the support of a massive Vlasov field may go all the way to the boundary $\partial \mathcal{D}$, where non-trivial forms of trapping occur for massive particles. We call $\mathcal{D}$ the \emph{dispersive region} of the mass-shell.

In this article, we establish decay estimates in time for massive Vlasov fields on Schwarzschild supported on the dispersive region $\D$ of the mass-shell. For this, we first provide an explicit characterisation of the dispersive region $\mathcal{D}$, which follows from the complete integrability of the geodesic flow in Schwarzschild. Then, for this class of Vlasov fields we show that the components $\mathrm{T}_{\mu\nu}[f]$ of the energy-momentum tensor decay like $u^{-\frac{1}{3}}r^{-2}$ in the far-away region $\{r\geq R\}$, and like $v^{-\frac{1}{3}}$ in the bounded region $\{r\leq R\}$ (under a suitable normalisation), where $R>2M$. Here, $(u,v)$ denotes the standard Eddington--Finkelstein double null coordinate pair. This result is obtained by proving time decay of the volume of the momentum support of the distribution function. The proof of this decay property is based on a careful study of the intricate structure of trapping for the geodesic flow in $\mathcal{D}$. In the mass-shell of Schwarzschild, there are three different forms of trapping for the geodesic flow in $\mathcal{D}$: \emph{unstable trapping}, \emph{degenerate trapping at the sphere of innermost stable circular orbits}, and \emph{parabolic trapping at infinity}\footnote{By an abuse of terminology, we speak about parabolic trapping at infinity despite that this form of trapping does not occur in the bounded region of spacetime.}. We observe that these three forms of trapping occur at the boundary $\partial\D$, whereas only the first and the third ones occur in the interior $\mathrm{int}\, \mathcal{D}$. Our analysis requires quantitative estimates for the geodesic flow in a neighbourhood of the trapped set on $\mathcal{D}$. In particular, part of the analysis is carried out in a neighbourhood of the boundary $\partial\D$ of the dispersive region. We also exploit the red-shift effect near the future event horizon $\mathcal{H}^+$.

Unstable trapping, degenerate trapping at the sphere of innermost stable circular orbits (ISCO), and parabolic trapping at infinity, are classical trapping effects for the timelike geodesic flow on black hole exteriors. See \cite{Ch, ON95} for more information in the case of a Kerr black hole background. We hope that the methods developed in this paper for the analysis of massive Vlasov fields will be helpful when considering more complicated massive fields on black hole exteriors.

\subsection{The main result}

Let us present the main decay estimates we obtain for dispersive Vlasov fields on Schwarzschild. We recall that the well-posedness theory for the massive Vlasov equation on Schwarzschild spacetime is standard. See Subsection \ref{subsec_initial_value_probl} for more information on the Cauchy problem for this PDE.

From now on, we use the notation $A\lesssim B$ to specify that there exists a universal constant $C > 0$ such that $A \leq CB$, where $C$ depends on the black hole mass, or other fixed constants. We will often write $C>0$ to denote a general constant depending on the black hole mass, or other fixed constants.  

\subsubsection{Statement of the main theorem}

Fix a parameter $M>0$. Let $(\mathcal{E},g_M)$ be the exterior region of a Schwarzschild black hole background, including the event horizons $\mathcal{H}^+$ and $\mathcal{H}^-$. We define the dispersive region $\mathcal{D}$ of the mass-shell by $$\D:=\mathrm{clos~} \Big\{(x,p)\in \mathcal{P}: \text{$\gamma_{x,p} $ \,crosses \,$\mathcal{H}^{+}$ \,or  \, $r(\gamma_{x,p}(s))\to+\infty$ \,as\, $s\to \infty$} \Big\},$$ where $\pi\colon \mathcal{P}\to \mathcal{E}$ is the canonical projection, and $\gamma_{x,p}$ is the unique geodesic with initial data $(x,p)$ on the black hole exterior. The set $\mathcal{D}$ is the largest subset of the mass-shell where decay estimates for massive Vlasov fields on Schwarzschild hold. See Subsection \ref{subsubsection_dispersive_prop_dispers_dom} for more information.

Let $\underline{C}_{\mathrm{in}}\cup C_{\mathrm{out}}$ be a (bifurcate) initial null hypersurface, such that $\underline{C}_{\mathrm{in}}$ includes its terminal sphere on $\mathcal{H}^+$, and $C_{\mathrm{out}}$ goes out to $\mathcal{I}^+$. We define the subset $\Sigma\subset\D$ over the initial hypersurface $\underline{C}_{\mathrm{in}}\cup C_{\mathrm{out}}$ given by 
\begin{equation}
\Sigma:=\pi^{-1}(\underline{C}_{\mathrm{in}}\,\cup\, C_{\mathrm{out}})\, \cap\, \D. \nonumber
\end{equation}
For almost every $(x,p)\in\Sigma$, the geodesic $\gamma_{x,p}$ either crosses the future event horizon $\mathcal{H}^+$ or escapes to infinity towards the future, see Subsection \ref{subsub_charact_trapp}. We will consider Vlasov fields for which the initial distribution functions are supported on $\Sigma$. Let us set the nonnegative function $\Omega^2(r):=1-\frac{2M}{r}$, and the characteristic function $\chi_{\mathcal{D}}\colon \mathcal{P}\to \R$ of the set $\mathcal{D}$. We can now state our main result.

\begin{theorem}
\label{thm_main_intro}
Let $f$ be the solution of \eqref{vlasov_eqn_massive_intro} on the exterior of Schwarzschild spacetime arising from continuous compactly supported initial data $f_0$. Let $(u,v)$ be the Eddington--Finkelstein double null pair. For $R>2M$ sufficiently large, the components of the energy-momentum tensor $\T_{\mu\nu}[f\chi_{\mathcal{D}}]$ of the Vlasov field $f\chi_{\mathcal{D}}$ satisfy 
\begin{equation}\label{main_thm_intro_unbounded_estm}
\T_{uu}\lesssim \dfrac{\|f_0\|_{L^{\infty}_{x,p}}}{u^{\frac{1}{3}}r^2},\qquad \quad \T_{uv}\lesssim \dfrac{\|f_0\|_{L^{\infty}_{x,p}}}{u^{\frac{1}{3}}r^2}, \qquad\quad  \T_{vv}\lesssim \dfrac{\|f_0\|_{L^{\infty}_{x,p}}}{u^{\frac{1}{3}}r^2}, 
\end{equation}
for all $x\in  \{r\geq R\}$, and 
\begin{equation}\label{main_thm_intro_bounded_estm}
\frac{\T_{uu}}{\Omega^4}\lesssim \dfrac{\|f_0\|_{L^{\infty}_{x,p}}}{v^{\frac{1}{3}}},\qquad \quad 
 \frac{\T_{uv}}{\Omega^2}\lesssim \dfrac{\|f_0\|_{L^{\infty}_{x,p}}}{v^{\frac{1}{3}}}, \qquad\quad  \T_{vv}\lesssim \dfrac{\|f_0\|_{L^{\infty}_{x,p}}}{v^{\frac{1}{3}}}, 
 \end{equation}
for all $x\in \{r\leq R\}$. Similar estimates hold for the other components of the energy-momentum tensor.
\end{theorem}

\begin{remark}
\begin{enumerate}[label = (\alph*)]
\item The decay rates are related to the behaviour of timelike geodesics in a neighbourhood of $\{p_t=-1\}$, where parabolic trapping at infinity holds. We will show that the outgoing orbits in $\{p_t=-1\}$ satisfy $r(t)\sim t^{\frac{2}{3}}$ and $p^r(t)\sim t^{-\frac{1}{3}}$. In a neighbourhood of $\{p_t=-1\}$, we find bounded orbits that spend arbitrarily long periods of time in the far-away region before crossing $\mathcal{H}^+$. This effect makes the slow dispersion towards infinity on $\{p_t=-1\}$ become relevant even in the bounded region. The particle energy value $p_t=-1$ corresponds to the rest mass of the particles in the system.
\item The components $\T_{\mu\nu}$ of the energy-momentum tensor have different $\Omega^2$ normalisations in the estimates on the bounded region. This discrepancy is due to the different $\Omega^2$ normalisations for the covectors $p_u$ and $p_v$ near $\mathcal{H}^+$. We recall that the estimates \eqref{main_thm_intro_unbounded_estm}--\eqref{main_thm_intro_bounded_estm} are stated using Eddington--Finkelstein double null coordinates on the black hole exterior. The normalised expressions in the LHS of the estimates in \eqref{main_thm_intro_bounded_estm} correspond to regular quantities along $\mathcal{H}^+$. 
\end{enumerate}
\end{remark}

\begin{remark}
The distribution function $f\chi_{\mathcal{D}}$ satisfies the massive Vlasov equation, because $\chi_{\mathcal{D}}(x,p)$ is conserved along the geodesic flow. This property follows from the invariance of $\mathcal{D}$ under the geodesic flow.
\end{remark}

As additional results of the article, we will also show \emph{improved decay estimates} specialised to the cases of compactly supported data on $\{-p_t> 1\}$, and data supported up to the boundary of $\{-p_t\geq 1\}$. The improved decay estimates in these domains will show interesting dispersive features of massive Vlasov fields on Schwarzschild. Being compactly supported in the set $\{-p_t> 1\}$ will make the Vlasov fields behave similarly to \emph{massless Vlasov fields} on a fixed Schwarzschild background. On the other hand, we will consider compactly supported Vlasov fields in the region $\{-p_t\geq  1\}$, where parabolic trapping at infinity can also occur. These two results will be stated and further discussed in Section \ref{section_statements}.

\subsection{Key ingredients of the proof}

Let us discuss the key ingredients to establish Theorem \ref{thm_main_intro}.

We obtain quantitative decay estimates by proving the decay in time of the volume of the momentum support of the distribution function. Even though one can expect decay for massive Vlasov fields supported on the dispersive region $\mathcal{D}$ to hold, obtaining quantitative estimates is a difficult task because of the complexity of the trapped set on $\mathcal{D}$. We will show that the decay of the momentum support of the Vlasov fields holds, because of the concentration of the momentum support in suitable distributions of phase space. These distributions correspond to the tangent space of the stable manifolds associated to the trapped set on $\mathcal{D}$.

\subsubsection{Intricacies of the trapped set on $\mathcal{D}$} For the timelike geodesic flow in Schwarzschild spacetime, we find three different forms of trapping on $\mathcal{D}$: unstable trapping, degenerate trapping at the sphere of innermost stable circular orbits (ISCO), and parabolic trapping at infinity. Let us discuss some key features of these forms of trapping:

\begin{itemize}
\item \emph{Unstable trapping}. This occurs at all spheres of the form $\{r=r_-\}$ with $r_-\in (3M,6M)$. We remark the existence of homoclinic orbits for the spheres of trapped orbits $\{r=r_-\}$ with $r_-\in (4M,6M)$. These homoclinic orbits are also trapped, and can spend arbitrarily long periods of time in the far-away region before approaching back a sphere of trapped orbits. The trapped orbits over the spheres $\{r=r_-\}$ with $r_-\in [4M,6M)$, and the associated homoclinic orbits, are at the boundary $\partial\D$. On the contrary, the trapped orbits over the spheres $\{r=r_-\}$ with $r_-\in (3M,4M)$, are in $\mathrm{int} \,\D$.
\item \emph{Degenerate trapping at ISCO}. The degenerate trapping effect at ISCO takes place at the sphere $\{r=6M\}$, where the unstable trapping effect \emph{degenerates}. As a result, the local dispersion estimates are only inverse polynomial near $\{r=6M\}$. We note that the trapped orbits over the sphere $\{r=6M\}$, are at the boundary $\partial\D$.
\item \emph{Parabolic trapping at infinity}. This occurs at spheres of trapped orbits at infinity, where $p_t=-1$. For orbits in $\{p_t> -1\}$ with angular momentum less than $4M$, we find bounded orbits that spend arbitrarily long periods of time in the far-away region before crossing $\mathcal{H}^+$. This property lies at the heart of the estimates stated in Theorem \ref{thm_main_intro}. The orbits in the set $\{p_t=-1\}$ correspond to parabolic orbits in analogy with the classification of orbits for the two-body problem in classical mechanics. We note that the trapped orbits over the spheres at infinity with angular momentum greater or equal to $4M$, are at the boundary $\partial\D$. On the contrary, the trapped orbits over the spheres at infinity with angular momentum smaller to $4M$, are in the interior $\mathrm{int}\,\D$.
\end{itemize}

The proof of the decay estimates in Theorem \ref{thm_main_intro} boils down to show suitable concentration estimates on the stable manifolds associated to the three previous forms of trapping. Let us mention the key mechanisms that allow us to obtain these estimates.

\subsubsection{Mechanisms behind the concentration estimates} First, we show local dispersion near the spheres of unstable trapped orbits, by using the normal hyperbolicity of the geodesic flow. This analysis is particularly delicate near the homoclinic orbits, where we also need to exploit the degenerate dispersion near $\{r=6M\}$ associated to the degenerate trapping at ISCO, and the dispersion at infinity near $\{p_t=-1\}$ associated to the parabolic trapping at infinity. Secondly, we show local dispersion near the sphere of trapped orbits $\{r=6M\}$, by using the parabolic behaviour of the timelike geodesic flow in this regime. Finally, we obtain decay in time in a neighbourhood of $\{p_t=-1\}$, by using the parabolic trapping at infinity. For this, suitable McGehee type coordinates are useful to study the radial flow. We also exploit the red-shift effect to show decay in a neighbourhood of $\mathcal{H}^+$.

Putting together the concentration estimates on the stable manifolds leads to the proof of Theorem \ref{thm_main_intro}.

\subsection{Related works on confined Vlasov fields}

We overview some recent linear and non-linear results concerning confined Vlasov fields.

\subsubsection{Linear stability results for non-trivial stationary states on black hole exteriors}

Phase mixing is a mechanism that leads to the weak convergence of Vlasov fields towards non-trivial stationary states. Phase mixing results without a rate for massive Vlasov fields on the closure of the complement $\mathcal{P}\setminus\D$ have been shown by Rioseco--Sarbach \cite{RS20}. See also \cite{RS24} for phase mixing results for massive fields on the exterior of Kerr black holes. 

We also mention the work by Günther--Rein--Straub \cite{GRS22} which has shown the existence of linearly stable small matter shell solutions of the spherically symmetric Einstein--massive Vlasov system. This result is based on the Birman--Schwinger principle.

\subsubsection{Stationary bifurcations of Schwarzschild for the Einstein--massive Vlasov system}

The Einstein--massive Vlasov system admits one-parameter families of stationary spherically symmetric solutions bifurcating from Schwarzschild \cite{R94, Ja21}. Also, one-parameter families of stationary axisymmetric solutions of the Einstein--massive Vlasov system bifurcating from Kerr spacetimes have been constructed by Jabiri \cite{Ja22}.

\subsubsection{Quantitative phase mixing for inhomogeneous equilibria}

Quantitative phase mixing estimates in the closure of the complement $\mathcal{P}\setminus \mathcal{D}$ have not been shown yet. However, quantitative estimates have been shown for related Vlasov equations. Recently, linear and non-linear phase mixing has been obtained for the Vlasov--Poisson system with an external Kepler potential by Chaturvedi--Luk \cite{CL24}. Quantitative linear phase mixing estimates are shown outside symmetry, and also long-time nonlinear phase mixing for spherically symmetric data. Also, quantitative estimates have been derived for the solutions to linear transport equations driven by a general family of Hamiltonians in \cite{HRSS24}. See also \cite{CL21} and \cite{MRV22} for other quantitative estimates for the solutions of Vlasov equations with anharmonic potentials. 

\subsection{Related works on dispersive Vlasov fields}

We discuss works on related dispersive Vlasov fields.

\subsubsection{Asymptotic stability of Minkowski for the Einstein--massive Vlasov system}

Minkowski spacetime is the simplest solution of the Einstein vacuum equations. This spacetime corresponds to a member of the Schwarzschild family for which the mass parameter vanishes. The problem of describing the non-linear dynamics of self-gravitating collisionless systems near Minkowski has been resolved.

The asymptotic stability of Minkowski for the Einstein--massive Vlasov system was obtained independently in the seminal works by Lindblad--Taylor \cite {LT} and Fajman--Joudioux--Smulevici \cite{FJS}. The work \cite{FJS} estimates the Vlasov field by using a weighted commuting vector field method based on Sobolev inequalities. On the other hand, the work \cite {LT} estimates the Vlasov field by using a vector field method and quantitative estimates for the timelike geodesic flow. See also the stability result by Wang \cite{Wa22b} which estimates the Vlasov field using Fourier techniques. Previously, Rein--Rendall \cite{RR92} proved the non-linear stability of Minkowski for the Einstein--massive Vlasov system under spherically symmetric perturbations.

\subsubsection{Decay for massless Vlasov fields on black hole exteriors}

Massless collisionless systems are another class of systems of interest in relativistic kinetic theory. In these, the rest mass of their particles vanishes. The problem of establishing decay for massless Vlasov fields on black hole exteriors has been previously studied.

First, Andersson--Blue--Joudioux \cite{ABJ} studied massless fields on very slowly rotating Kerr black holes. Boundedness of a weighted energy norm, and a degenerated integrated energy decay estimate were derived. This work does not provide pointwise decay for momentum averages. Later, Bigorgne \cite{L23} adapted the $r^p$-weighted energy method of Dafermos--Rodnianski, to show polynomial decay of momentum averages for massless fields on Schwarzschild. For this, polynomial decay for a non-degenerate energy flux was shown.

More recently, \cite{VR23} shows a non-linear stability result of Schwarzschild for the Einstein--massless Vlasov system under spherically symmetric perturbations. This result uses the normal hyperbolicity of the trapped set for the null geodesic flow to show decay for the energy-momentum tensor. Here, the components of the energy-momentum tensor are proved to decay exponentially in the bounded region of spacetime. On the other hand, Weissenbacher \cite{W23} proved exponential decay of momentum averages for massless Vlasov fields on subextremal and extremal Reissner--Nordström spacetimes. In the extremal case, when the Vlasov field and its first time derivative are initially supported on a neighbourhood of $\mathcal{H}^+$, the transversal derivative of a suitable component of the energy-momentum tensor does not decay along $\mathcal{H}^+$. The works \cite{VR23,W23} are based on phase space volume estimates.

Finally, in joint work with Bigorgne \cite{BV24}, a commutation vector field approach is developed to study the decay for massless Vlasov fields on Schwarzschild spacetime using a weighted energy method. Here, the hyperbolicity of the trapped set is used to construct a suitable $W_{x,p}^{1,1}$ norm for which solutions to the massless Vlasov equation verify an \emph{integrated local energy decay estimate without relative degeneration}. As a result, time decay is established for a first order energy norm using the $r^p$-method. 

A central difficulty in the analysis of massless Vlasov fields on black hole exteriors is the unstable trapping effect of the null geodesic flow. This phenomenon also occurs for the timelike geodesic flow on Schwarzschild. 

\subsubsection{Decay for Vlasov equations on backgrounds with hyperbolic flows}

Motivated by the unstable trapping phenomenon, decay estimates have been studied for Vlasov fields on backgrounds with hyperbolic flows. The small data solutions for the Vlasov--Poisson system with the repulsive potential $\frac{-|x|^2}{2}$, for which unstable trapping occurs for the associated Hamiltonian flow, were studied in \cite{VV24, BVV23}. Global existence in dimension $n\geq 2$ was obtained in the former work, and the scattering properties in dimension two were addressed in the latter one. Also, \cite{VVel23} studied decay estimates for Vlasov fields on non-trapping asymptotically hyperbolic manifolds. These works are based on commuting vector fields approaches.

\subsection{Related work on extremal black hole formation as a critical phenomenon} On recent work Kehle--Unger \cite{KU24} proved that extremal black holes arise on the threshold of gravitational collapse. This article constructs one-parameter families of smooth spherically symmetric solutions of the Einstein--Maxwell--Vlasov system that interpolate between dispersion and collapse, and for which the critical solution is an extremal black hole. This work makes use of fine-tuned beams of self-gravitating collisionless charged particles. Suitable decay estimates for the energy-momentum tensor of the Vlasov matter are part of the analysis.

\subsection{Outline of the paper}
The rest of the article is structured as follows.

\begin{itemize}
    \item \textbf{Section \ref{secvlasovfields}.} We set up the geometric framework to study massive Vlasov fields on Schwarzschild spacetime. We review the formulation of the initial value problem for the massive Vlasov equation.
    \item \textbf{Section \ref{section_statements}.} We state detailed versions of the main results of the article. We also provide a detailed summary of the proof of Theorem \ref{thm_main_intro}.
    \item \textbf{Section \ref{section_timelike_geodesics}.} We review the complete integrability of the geodesic flow in Schwarzschild. We show an explicit characterisation of $\mathcal{D}$. We describe the trapped set for the timelike geodesic flow on $\mathcal{D}$. We show some dispersive properties of the geodesic flow in the near-horizon and the far-away regions. We show a priori estimates for the momentum coordinates along timelike geodesics.
    \item \textbf{Section \ref{section_struct_trapp}.} We study in detail the properties of the different forms of trapping for the timelike geodesic flow in $\mathcal{D}$. We identify suitable defining functions for the corresponding stable manifolds. We also obtain suitable expansion and contraction properties along the geodesic flow.
    \item \textbf{Section \ref{section_radial_flow}.} We establish concentration estimates on the stable manifolds for the different forms of trapping of the geodesic flow. The estimates here are satisfied in a neighbourhood of the trapped set.     
    \item \textbf{Section \ref{section_proofs_main_results_massive}.} We prove the main results of the article. We first obtain a priori estimates in the near-horizon region. Later, we use the concentration estimates to show decay in time for $\T_{\mu\nu}[f]$.   
\end{itemize}

\subsection{Acknowledgements}
I would like to express my gratitude to Mihalis Dafermos and Cl\'ement Mouhot for their continued guidance and encouragements. I also would like to thank Jonathan Luk and Jacques Smulevici, for several insightful discussions and corrections. Finally, I thank L\'eo Bigorgne and Yakov Shlapentokh-Rothman for many helpful discussions. I have received partial funding from the European Union’s Horizon 2020 research and innovation programme under the Marie Skłodowska-Curie grant 101034255.

\section{Preliminaries}\label{secvlasovfields}

In this section, we recall the basic properties of massive Vlasov fields on the exterior of Schwarzschild.

\subsection{The exterior of Schwarzschild spacetime}

The Schwarzschild family of black holes is a one-parameter family of stationary four dimensional Lorentzian manifolds. This family is parametrised by the black hole mass $M\in\R$. From now on, we fix $M>0$. 

\subsubsection{In Kruskal double null coordinates}

Let us define the differential structure and the metric of the exterior region of Schwarzschild in terms of Kruskal coordinates. We denote by $\mathcal{E}$ the exterior of the maximally extended Schwarzschild spacetime. We define the manifold with corners $\mathcal{E}$ in \emph{Kruskal double null coordinates} $(U,V,\theta,\phi)$ as $$\mathcal{E}:=\big\{(U,V,\theta,\phi)\in (-\infty,0]\times [0,\infty)\times \S^2\big\}.$$ This coordinate system is global up to the usual degeneration of the spherical variables. The boundary of $\mathcal{E}$ consists of the two hypersurfaces $$\mathcal{H}^+:=\{0\}\times (0,\infty)\times \S^2,\qquad \qquad \mathcal{H}^-:=(-\infty,0)\times \{0\}\times \S^2,$$ and the two-sphere $\S^2_{\mathrm{bif}}: =\{U=V=0\}$. We refer to $\mathcal{H}^+$ and $\mathcal{H}^-$, as the future and the past event horizon, respectively.

Fix $M>0$. Let $r(U,V)$ be the function implicitly defined by $$-UV=e^{\frac{r}{2M}}(\frac{r}{2M}-1).$$ We also define $\Omega^2_K(U,V)=\frac{8M^3}{r(U,V)}e^{-\frac{r(U,V)}{2M}}$. In the Kruskal coordinate system, the Schwarzschild metric $g_M$ is defined as $$g_M=-2\Omega^2_K(U,V)(\d U\otimes \d V+\d V\otimes \d U)+r^2(U,V)g_{\S^2},$$ where $g_{\S^2}=\d \theta\otimes\d\theta+\sin^2\theta \d\phi\otimes\d\phi$ is the round metric on the unit sphere $\S^2$. The hypersurfaces $\mathcal{H}^+$ and $\mathcal{H}^-$ are null with respect to $g_M$. We set the time-orientation in $(\mathcal{E},g_M)$ so that $\partial_U+\partial_V$ is future directed.

See \cite{Sy50, K60} for more information about the full maximal extension of Schwarzschild spacetime.

\subsubsection{In Eddington--Finkelstein double null coordinates}

We cover the interior of $\mathcal{E}$ by using \emph{Eddington--Finkelstein double null coordinates} $(u,v,\theta,\phi)$ given by $$U=-\exp\Big(-\frac{u}{2M}\Big),\qquad\quad V=\exp\Big(\frac{v}{2M}\Big).$$ In the coordinate system $(u,v,\theta,\phi)$, the Schwarzschild metric takes the form $$ g_M=-2\Omega^2(u,v)(\d u\otimes \d v+\d v\otimes \d u)+r^2(u,v)g_{\S^2}, \quad\qquad \Omega^2(u,v):=1-\frac{2M}{r},$$ where the function $r(u,v)$ is implicitly defined by $$e^{\frac{v-u}{2M}}=e^{\frac{r}{2M}}\Big(\frac{r}{2M}-1\Big).$$ Setting $t=u+v$, the metric $g$ in Schwarzschild coordinates $(t,r,\theta,\phi)$ takes the usual form $$g_M=-\Omega^2(r)\mathrm{d} t\otimes \mathrm{d} t+\dfrac{1}{\Omega^2(r)}\mathrm{d} r\otimes \mathrm{d} r+r^2 g_{\S^2}, \qquad\quad \Omega^2(r):=1-\frac{2M}{r}.$$

In Eddington--Finkelstein double null coordinates, the constant $v$ and $u$ level sets define two families of null hypersurfaces: the incoming family $\{\underline{C}_v\}$, and the outgoing family $\{C_u\}$. The coordinates $(u,v,\theta,\phi)$ do not cover the future event horizon $\mathcal{H}^{+}$, neither the future null infinity $\mathcal{I}^{+}$. However, we can formally parametrise $\mathcal{H}^+$ by $(\infty, v,\theta,\phi)$, and $\mathcal{I}^+$ by $(u,\infty,\theta,\phi)$. The same happens for the past event horizon $\mathcal{H}^{-}$, and the past null infinity $\mathcal{I}^{-}$. We can formally parametrise $\mathcal{H}^-$ by $(u, -\infty,\theta,\phi)$, and $\mathcal{I}^-$ by $(-\infty,v,\theta,\phi)$. See Figure \ref{Penrose_diagram_black_hole_exterior0} for the Penrose diagram of the exterior region.

\begin{figure}[h!]
\begin{center}
\begin{tikzpicture}
\node [below right] at (1.2,1.2) {$\mathcal{I}^+$};
\node [below left] at (-1,1.2) {$\mathcal{H}^+$};
\node [above] at (0,2) {$i^+$};
\node [below] at (0,-2) {$i^-$};
\draw [thin](-2,0) -- (0,-2);
\draw [thin](-1.5,0.5) -- (0.5,-1.5);
\node [below right] at (-1.1,0.5) {$\underline{C}_v$};
\draw [dashed, thin](2,0) -- (0,-2);
\draw [thin](1.5,0.5) -- (-0.5,-1.5);
\node [below left] at (1.1,0.5) {$C_u$};
\draw [thin] (-2,0) -- (0,2);
\draw [dashed, thin] (0, 2) -- (2,0);
\node [below left] at (-0.9,-0.9) {$\mathcal{H}^-$};
\node [below right] at (0.9,-0.9) {$\mathcal{I}^-$};
\draw [fill=white] (0,2) circle [radius=0.08];
\draw [fill=white] (2,0) circle [radius=0.08];
\node [right] at (2,0) {$i^0$};
\end{tikzpicture}
\end{center}
\caption{Penrose diagram of the exterior region $\mathcal{E}$.}
\label{Penrose_diagram_black_hole_exterior0}
\end{figure}
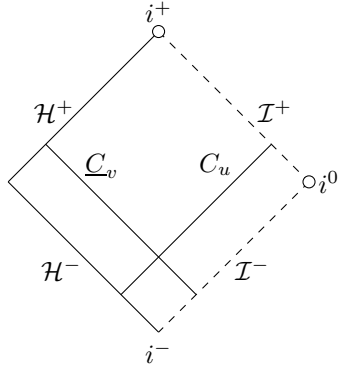

In the coordinate system $(u,v,\theta,\phi)$, the non-vanishing Christoffel symbols $\Gamma_{\beta\gamma}^{\alpha}$ are given by
\begin{align*}
    \Gamma_{uu}^u&=-\frac{2M}{r^2}, \qquad \qquad\quad\quad\Gamma_{vv}^v=\frac{2M}{r^2},\qquad \qquad\quad\,\,\,\Gamma_{AB}^u=\frac{1}{2} r(g_{\S^2})_{AB},\,\qquad\quad \Gamma_{AB}^v=-\frac{1}{2} r(g_{\S^2})_{AB},\\
    \Gamma_{Bu}^A&=-\frac{r-2M}{r^2} \delta_B^A, \qquad \quad\Gamma_{Bv}^A=\frac{r-2M}{r^2} \delta_B^A,\qquad \quad\Gamma_{\phi\phi}^{\theta}=-\sin\theta\cos\theta,\qquad \quad\Gamma_{\theta\phi}^{\phi}=\cot\theta.
\end{align*}
where the latin indices $A$ and $B$ correspond to the spherical coordinates of spacetime. 

The reader can consult \cite{HE73, W84} for more information about the geometry of Schwarzschild spacetime.

\subsubsection{Killing fields of Schwarzschild spacetime}

The interior of $(\mathcal{E},g_M)$ is stationary since $\partial_t$ is a Killing vector field that is timelike in $\{r>2M\}$. In Eddington--Finkelstein double null coordinates, this timelike Killing field is given by $$\partial_t=\frac{1}{2}(\partial_u+\partial_v).$$ The vector field $\partial_t$ extends to a smooth Killing vector field on $\mathcal{H}^+$. Moreover, it is null and normal to $\mathcal{H}^+$.

On the other hand, Schwarzschild is spherically symmetric since 
\begin{equation}\label{killing_sph_symm}
\mathbf{\Omega}_1=\partial_{\phi},\qquad  \mathbf{\Omega}_2=\cos \phi \,\partial_{\theta}-\sin\phi \cot \theta\, \partial_{\phi},\qquad \mathbf{\Omega}_3=-\sin \phi\,\partial_{\theta}-\cos \phi\cot\theta\,\partial_{\phi},
\end{equation}
are spacelike Killing fields generating a smooth action by isometries of $SO(3)$.

The Lie algebra of Killing fields of Schwarzschild is exactly the one generated by $\partial_t$, $\mathbf{\Omega}_1$, $\mathbf{\Omega}_2$, and $\mathbf{\Omega}_3$.

\subsection{The massive Vlasov equation on Schwarzschild}

Let us review the basic properties of the timelike geodesic flow and the massive Vlasov equation on Schwarzschild.

\subsubsection{The timelike geodesic flow in Schwarzschild}

The motion of free falling particles on Schwarzschild spacetime is described by its \emph{geodesic flow}. Let $(x^{\alpha})$ be a local coordinate system on Schwarzschild with dual momentum coordinates $(p^{\alpha})$ on the fibers of $T\mathcal{E}$. We recall that $(x^{\alpha})$ is \emph{dual} to $(p^{\alpha})$ on $T\mathcal{E}$, if $p^{\alpha}=\d x^{\alpha}(p)$ for any $p\in T_x\mathcal{E}$. A coordinate system $(x^{\alpha}, p^{\alpha})$ on $T\mathcal{E}$ with $(x^{\alpha})$ dual to $(p^{\alpha})$ is called \emph{canonical}. In a canonical coordinate system on $T\mathcal{E}$, the geodesic flow is determined by the geodesic equations
\begin{equation}\label{geo_eqn}
\frac{\d x^{\alpha}}{\d s}=p^{\alpha}, \quad\qquad  \frac{\d p^{\alpha}}{\d s}=-\Gamma^{\alpha}_{\beta \gamma}p^{\beta} p^{\gamma},
\end{equation}
where $\Gamma_{\beta \gamma}^{\alpha}$ are the Christoffel symbols of $g_M$ with respect to $(x^{\alpha})$, and $s$ parametrises the geodesic flow. 

Given $m>0$, we define the subset $\mathcal{P}_m\subset T\mathcal{E}$ of the tangent bundle as $$\mathcal{P}_m:=\Big\{(x,p)\in T\mathcal{E}: g_x(p,p)=-m^2, \text{ where $p$ is future directed}\Big\}.$$ The set $\mathcal{P}_m\subset T\mathcal{E}$ is a smooth seven dimensional submanifold, since it is the level set of a smooth function with non-vanishing gradient. We observe that $\mathcal{P}_m$ is invariant under the geodesic flow. By the condition $g_x(p,p)=-m^2$, the momentum vectors in the fibers of $\mathcal{P}_m$ are timelike. For a timelike geodesic $\gamma$ in $\P_m$, the value $m$ denotes the \emph{rest mass} of a free falling particle following $\gamma$. 

In relativistic kinetic theory is standard to consider systems composed by particles with a fixed rest mass $m$. From now on, we consider systems composed by particles whose rest mass is normalised to be one. Under this normalisation, we denote $\mathcal{P}:=\mathcal{P}_1$. We refer to $\mathcal{P}$ as the \emph{mass-shell}, and to the identity $g_x(p,p)=-1$ as the \emph{mass-shell relation}. We denote the canonical projection by $\pi\colon \mathcal{P}\to \mathcal{E}$.

\subsubsection{The massive Vlasov equation}

In this paper, we study collisionless systems on Schwarzschild spacetime. These systems are statistically described by a non-negative measure $f(x,p)\mathrm{d\mu}_{\mathcal{P}}$ that is absolutely continuous with respect to the induced volume form on $\mathcal{P}$. For collisionless systems $f\mathrm{d\mu}_{\mathcal{P}}$, the distribution functions $f$ are called \emph{Vlasov fields}. 

\begin{definition}\label{def_vla_field}
A \emph{Vlasov field on Schwarzschild} $f$ is a non-negative distribution function on $T\mathcal{E}$ that is constant along the geodesics in spacetime. A Vlasov field $f$ supported on $\mathcal{P}$ is called a massive Vlasov field on Schwarzschild. 
\end{definition}

By Definition \ref{def_vla_field}, in any canonical coordinate system $(x^{\alpha},p^{\alpha})$ on $T\mathcal{E}$, a sufficiently regular Vlasov field $f$ satisfies the \emph{Vlasov equation on Schwarzschild}
\begin{equation}\label{vlasov_eqn_defn_intro}
\mathbb{X}_{g_M}f:= p^{\alpha}\partial_{x^{\alpha}}f-p^{\alpha}p^{\beta}\Gamma^{\gamma}_{\alpha\beta}\partial_{p^{\gamma}}f=0,
\end{equation}
where $\Gamma_{\beta \gamma}^{\alpha}$ are the Christoffel symbols of $g_M$ in the given coordinate chart. The vector field $\mathbb{X}_{g_M}$ is independent of the chosen coordinate system. Note that $\mathbb{X}_{g_M}$ is tangent to the mass-shell, since $\mathcal{P}$ is invariant under the geodesic flow. The vector field $\mathbb{X}_{g_M}\in \Gamma(TT\mathcal{E})$ 
is the \emph{generator of the geodesic flow}.

The \emph{energy-momentum tensor} $\T_{\mu\nu}[f]$ of a Vlasov field $f$ is the symmetric $(0,2)$ tensor field
\begin{equation}\label{enmomtensor}
\T_{\mu\nu}[f]:=\int_{\mathcal{P}_x} f(x,p)p_{\mu}p_{\nu}\mathrm{d\mu}_{\mathcal{P}_x},
\end{equation}
where $\mathcal{P}_x$ are the fibers of the mass-shell, and $p_{\mu}=g_{\mu\nu}p^{\nu}$. The components of $\T_{\mu\nu}$ are the second moments of the distribution function with respect to the momentum variable. For a Vlasov field $f$, the energy-momentum tensor $\T_{\mu\nu}$ satisfies the conservation law $$\nabla^{\mu}\T_{\mu\nu}=0.$$ Both, the dominant and the strong energy conditions for Vlasov matter follow from the definition of the energy-momentum tensor.

Similarly, the \emph{particle current} $\mathrm{N}_{\mu}$ of a Vlasov field $f$ is the one-form $$\mathrm{N}_{\mu}:=\int_{\mathcal{P}_x} f(x,p)p_{\mu}\mathrm{d\mu}_{\mathcal{P}_x}.$$ The components of $\mathrm{N}_{\mu}$ are first moments of the distribution function. For a Vlasov field $f$, the particle current $\mathrm{N}_{\mu}$ satisfies the conservation law $$\nabla^{\mu}\mathrm{N}_{\mu}=0.$$ By considering higher order moments of the distribution function, one can define higher order tensor fields that also satisfy conservation laws for a Vlasov field $f$. See \cite[Chapter 10]{CB09} for more information.

\begin{remark}
We will consider smooth Vlasov fields $f$ which are compactly supported in the momentum variables for any $x\in \mathcal{E}$. Under this hypothesis, moments of the distribution function are smooth maps. In particular, $\mathrm{T}_{\mu\nu}(x)$ and $\mathrm{N}_{\mu}(x)$ are smooth tensor fields on Schwarzschild.
\end{remark}

In the canonical coordinate system $(u,v,\theta,\phi,p^u,p^v,p^{\theta}, p^{\phi})$ on $T\mathcal{E}$, the Vlasov equation on Schwarzschild spacetime is \begin{align*}
p^{u}\partial_uf+p^{v}\partial_vf&+p^{\theta}\partial_{\theta}f+p^{\phi}\partial_{\phi}f+\Big(\dfrac{2M}{r^2}(p^u)^2-\frac{\ell^2}{2r^3}\Big)\partial_{p^u}f+\Big(\dfrac{2M}{r^2}(p^v)^2+\frac{\ell^2}{2r^3}\Big)\partial_{p^v}f\\
&\qquad+\Big(-\dfrac{2p^rp^{\theta}}{r}+\sin\theta \cos \theta (p^{\phi})^2\Big)\partial_{p^{\theta}}f+\Big(-\dfrac{2p^rp^{\phi}}{r}-2\cot \theta p^{\theta}p^{\phi}\Big)\partial_{p^{\phi}}f=0,
\end{align*}
where $\ell^2:=r^4(g_{\S^2})_{AB}p^Ap^B$. We will consider smooth massive Vlasov fields, that is, Vlasov fields supported on $\mathcal{P}$. In this coordinate system, the mass-shell $\mathcal{P}$ over Schwarzschild is $$\P=\Big\{(x,p)\in T\mathcal{E}: 4\Omega^2p^up^v=1+r^2(g_{\S^2})_{AB}p^Ap^B, \text{ where $p$ is future directed}\Big\}.$$ Moreover, the mass-shell relation can be written as $$4\Omega^2p^up^v=1+\frac{\ell^2}{r^2}.$$ The canonical coordinate system $(u,v,\theta,\phi,p^u,p^v,p^{\theta}, p^{\phi})$ on $T\mathcal{E}$ induces a coordinate system $(u,v,\theta,\phi,p^v,$ $p^{\theta}, p^{\phi})$ on the mass-shell $\mathcal{P}$, where $p^u$ is defined by the mass-shell relation as $$p^u=\frac{1}{4\Omega^2 p^v}\Big(1+\frac{\ell^2}{r^2}\Big).$$ 

The components $\T_{\mu\nu}$ of the energy-momentum tensor of a smooth Vlasov field $f$ can be expressed in this coordinate system as
\begin{equation}\label{energy_momentum_local_coordinates}
\T_{\mu\nu}=\int_{\{4\Omega^2p^up^v=1+\frac{\ell^2}{r^2}\}} f(x,p)p_{\mu} p_{\nu}\dfrac{r^2\sqrt{\det \gamma}}{p^v}\d p^v\d p^A\d p^B,
\end{equation}
in terms of the volume form on the fibers of the mass-shell. One can write similarly the components $\mathrm{N}_{\mu}$ of the particle current in local coordinates.

\subsection{The initial value problem for the Vlasov equation}\label{subsec_initial_value_probl}

Let $u_0$, $v_0\in\R$. Let $\underline{C}_{\mathrm{in}}\cup \,C_{\mathrm{out}}$ be a bifurcate null hypersurface, where $\underline{C}_{\mathrm{in}}:=[u_0,\infty]\times\{v=v_0\}\times \S^2$ is an incoming hypersurface penetrating the future event horizon $\mathcal{H}^+$, and $C_{\mathrm{out}}:=\{u=u_0\}\times [v_0,\infty)\times \S^2$ is a complete outgoing hypersurface going out to $\mathcal{I}^+$. The light cones $\underline{C}_{\mathrm{in}}$ and $C_{\mathrm{out}}$ intersect transversely at the spacelike two-sphere $\S^2_{u_0,v_0}$. We will consider initial data $f_0$ for the massive Vlasov equation on the subset of the mass-shell over $\underline{C}_{\mathrm{in}}\cup C_{\mathrm{out}}$. 

The well-posedness of the massive Vlasov equation \eqref{vlasov_eqn_defn_intro} follows by using the regularity of the flow map, and the standard representation formula of Vlasov fields in terms of the geodesic flow. For this, we consider the regular canonical coordinate system $(U,V,\theta,\phi, p^{U},p^{V},p^{\theta},p^{\phi})$ on the mass-shell $\mathcal{P}$, which is induced by the coordinate system $(U,V,\theta,\phi)$ in the exterior region $\mathcal{E}.$

\begin{proposition}
Let $f_0\colon \pi^{-1}(\underline{C}_{\mathrm{in}}\cup\,  C_{\mathrm{out}})\to\R$ be a continuous function with respect to the regular coordinate system in $\mathcal{P}$. Then, there exists a unique solution of $\mathbb{X}_{g_M}f=0$ in $\pi^{-1}(\{x\in \mathcal{E}: u\geq u_0,\,v\geq v_0\})$ such that $f|_{\pi^{-1}(\underline{C}_{\mathrm{in}}\,\cup \,C_{\mathrm{out}})}=f_0$. Moreover, the function $f\colon \pi^{-1}(\{u\geq u_0,\,v\geq v_0\})\to \R$ is continuous with respect to the regular coordinate system in $\mathcal{P}$.
\end{proposition}

Let $S$ be a suitable Cauchy hypersurface for the exterior region, including its terminal sphere on $\mathcal{H}^+$, and going out to $\mathcal{I}^+$. The massive Vlasov equation is also well-posed on the black hole exterior $\mathcal{E}$ with continuous initial data $f_0$ on $\pi^{-1}(S)$. The main results of this article apply to both, Vlasov fields with initial data supported on $\pi^{-1}(\underline{C}_{\mathrm{in}}\, \cup \,C_{\mathrm{out}})$, and Vlasov fields with initial data supported on $\pi^{-1}(S)$. In the rest of the paper, we will only consider initial data $f_0$ supported on $\pi^{-1}(\underline{C}_{\mathrm{in}} \cup C_{\mathrm{out}})$ for the sake of clarity.

\section{The main results}\label{section_statements}

We state a detailed version of our main result. We also state two other theorems concerning the improved decay estimates specialised to the cases of compactly supported data on $\{-p_t> 1\}$, and data supported up to the boundary of $\{-p_t\geq 1\}$. We conclude this section with a summary of the proof of the main result.

\subsection{Conserved quantities along the geodesic flow}

The Vlasov equation on Schwarzschild is a transport equation along the geodesic flow in this background. This Hamiltonian flow can be studied in great detail since it is \emph{completely integrable}. By the stationarity and the spherical symmetry of this background, there exist three non-trivial integrals of motion for the geodesic flow: the \emph{particle energy} $E(x,p)$, the \emph{total angular momentum} $\ell(x,p)$, and the \emph{azimuthal angular momentum} $\ell_{\phi}(x,p)$. These conserved quantities are $$E(x,p)=\Omega^2(r) p^t, \qquad \ell(x,p)=r^2\sqrt{(p^{\theta})^2+\sin^2\theta (p^{\phi})^2},\qquad \ell_{\phi}(x,p)=r^2\sin^2\theta p^{\phi},$$ in Schwarzschild coordinates. The complete integrability of the geodesic flow in Schwarzschild follows by considering these integrals of motion. See Section \ref{section_timelike_geodesics} for more information.

\subsection{Decay for Vlasov fields compactly supported on $\mathcal{D}_0$}

Let us consider Vlasov fields supported on the subset $\D_0$ of the mass-shell $$\D_0:=\Big\{(x,p)\in \mathcal{P}:  E(x,p)>1 \Big\}.$$ The domain $\D_0$ is invariant under the geodesic flow, since $E(x,p)$ is an integral of motion. In this region of phase space, there is a normally hyperbolic trapped set contained in $\D_0\, \cap\, \{\ell>4M\}$. See Subsection \ref{subsec_radial_geo_flow} for more information.

Let $\underline{C}_{\mathrm{in}}\cup C_{\mathrm{out}}$ be a (bifurcate) initial null hypersurface, such that $\underline{C}_{\mathrm{in}}$ terminates at $\mathcal{H}^+$, and $C_{\mathrm{out}}$ goes out to $\mathcal{I}^+$. We define the subset $\Sigma_0\subset\D_0$ over the initial hypersurface $\underline{C}_{\mathrm{in}}\cup C_{\mathrm{out}}$ given by $$\Sigma_{0}:=\pi^{-1}(\underline{C}_{\mathrm{in}}\,\cup \,C_{\mathrm{out}})\, \cap\, \D_0.$$ For almost every $(x,p)\in\Sigma_0$, the timelike geodesic $\gamma_{x,p}$ either crosses the future event horizon $\mathcal{H}^+$ or escapes to infinity towards the future, see Subsection \ref{subs_decom_phsp}. We will consider Vlasov fields for which the initial distribution functions are supported on $\Sigma_{0}$. We can now state our first result.

\begin{theorem}\label{theorem_decay_fast}
Let $f$ be the solution of \eqref{vlasov_eqn_defn_intro} on the exterior of Schwarzschild spacetime arising from continuous initial data $f_0$ compactly supported on $\Sigma_0$. Let $(u,v)$ be the Eddington--Finkelstein double null pair. For $R>2M$ sufficiently large, the components of the energy-momentum tensor $\T_{\mu\nu}[f]$ satisfy $$ \T_{uu}\lesssim \dfrac{\|f_0\|_{L^{\infty}_{x,p}}}{u^3},\qquad\qquad \T_{uv}\lesssim \dfrac{\|f_0\|_{L^{\infty}_{x,p}}}{u^3}, \qquad\qquad \T_{vv}\lesssim \dfrac{\|f_0\|_{L^{\infty}_{x,p}}}{u^3},$$
for all $x \in  \{r\geq R\}$, and 
\begin{equation}\label{est_thm31_bounded}
\frac{\T_{uu}}{\Omega^4}\lesssim \dfrac{\|f_0\|_{L^{\infty}_{x,p}}}{\exp(\frac{1}{4\sqrt{2}M}v)}, \qquad\quad \frac{\T_{uv}}{\Omega^2}\lesssim \dfrac{\|f_0\|_{L^{\infty}_{x,p}}}{\exp(\frac{1}{4\sqrt{2}M}v)}, \qquad\quad \T_{vv}\lesssim \dfrac{\|f_0\|_{L^{\infty}_{x,p}}}{\exp(\frac{1}{4\sqrt{2}M}v)},
\end{equation}
for all $x\in \{r\leq R\}$. Similar estimates hold for the other components of the energy-momentum tensor.
\end{theorem}

\begin{remark}
\begin{enumerate}[label = (\alph*)]
    \item We obtain exponential decay for the components of the energy-momentum tensor in the bounded region $\{r\leq R\}$, because of the normal hyperbolicity of the trapped set for the geodesic flow in $\D_0$. In the bounded region, we obtain exponential decay with the rate $(4\sqrt{2}M)^{-1}$ that corresponds to the Lyapunov exponent of the hyperbolic fixed point $(r=4M,\, p^r=0)$ of the radial flow. The inverse polynomial decay in the far-away region $\{r\geq R\}$ coincides with the decay rate for the components of the energy-momentum tensor for massive Vlasov fields on Minkowski. 
    \item \label{remark_thm1_three} The components $\mathrm{T}_{\mu\nu}$ of the energy-momentum tensor have different $\Omega^2$ normalisations in the estimates on the bounded region. This discrepancy is due to the different $\Omega^2$ normalisations for the covectors $p_u=-2\Omega^2 p^v$ and $p_v=-2\Omega^2 p^u$ near $\mathcal{H}^+$, when working in Eddington--Finkelstein coordinates. We recall that $\Omega^2p^u$ and $p^v$ are regular quantities on $\mathcal{H}^+$, since the vector fields $\Omega^{-2}\partial_u$ and $\partial_v$ extend regularly to non-vanishing vector fields on $\mathcal{H}^+$. The normalised expressions in the LHS of the estimates in \eqref{est_thm31_bounded} correspond to regular quantities in the regular coordinate system $(U,V,\theta,\phi)$.
\end{enumerate}
\end{remark}

The exponential decay in Theorem \ref{theorem_decay_fast} can be compared with the exponential decay for \emph{massless Vlasov fields} on Schwarzschild provided in \cite{W23, VR23, BV24}. 

\subsection{Decay for Vlasov fields supported up to the boundary of $\mathcal{D}_1$}

Let us consider Vlasov fields supported on the subset $\D_1$ of the mass-shell $$\D_1:=\Big\{(x,p)\in \mathcal{P}: E(x,p)\geq 1 \Big\}.$$ The domain $\D_1$ is invariant under the geodesic flow since $E(x,p)$ is an integral of motion. In this region, apart from the normally hyperbolic trapped set contained in $\D_1\cap\{\ell\geq 4M\}$, we find the parabolic-type trajectories on the set $\{E=1\}$. As a result, the decay rates of the energy-momentum tensor in the far-away region change drastically.

Let $\underline{C}_{\mathrm{in}}\cup C_{\mathrm{out}}$ be a (bifurcate) initial null hypersurface, such that $\underline{C}_{\mathrm{in}}$ terminates at $\mathcal{H}^+$, and $C_{\mathrm{out}}$ goes out to $\mathcal{I}^+$. We define the subset $\Sigma_1\subset \D_1$ over the initial hypersurface $\underline{C}_{\mathrm{in}}\cup C_{\mathrm{out}}$ given by $$\Sigma_{1}:=\pi^{-1}(\underline{C}_{\mathrm{in}}\,\cup \,C_{\mathrm{out}})\,\cap \, \D_1.$$ For almost every $(x,p)\in\Sigma_1$, the timelike geodesic $\gamma_{x,p}$ either crosses the future event horizon $\mathcal{H}^+$ or escapes to infinity towards the future, see Subsection \ref{subs_decom_phsp}. We will consider Vlasov fields for which the initial distribution functions are supported on $\Sigma_1$. Let $\chi_{\mathcal{D}_1}\colon \mathcal{P}\to \R$ be the characteristic function of the set $\mathcal{D}_1$.

\begin{theorem}\label{theorem_decay_bded}
Let $f$ be the solution of \eqref{vlasov_eqn_defn_intro} on the exterior of Schwarzschild spacetime arising from continuous compactly supported initial data $f_0$. Let $(u,v)$ be the Eddington--Finkelstein double null pair. For $R>2M$ sufficiently large, the components of the energy-momentum tensor $\T_{\mu\nu}[f\chi_{\mathcal{D}_1}]$ of the Vlasov field $f\chi_{\mathcal{D}_1}$ satisfy  $$ \T_{uu}\lesssim \dfrac{\|f_0\|_{L^{\infty}_{x,p}}}{u^{\frac{5}{3}}},\qquad\qquad \T_{uv}\lesssim \dfrac{\|f_0\|_{L^{\infty}_{x,p}}}{u^{\frac{5}{3}}}, \qquad\qquad \T_{vv}\lesssim \dfrac{\|f_0\|_{L^{\infty}_{x,p}}}{u^{\frac{5}{3}}},$$ 
for all $x \in  \{r\geq R\}$, and $$\frac{\T_{uu}}{\Omega^4}\lesssim \dfrac{\|f_0\|_{L^{\infty}_{x,p}}}{\exp(\frac{1}{4\sqrt{2}M}v)}, \qquad\quad \frac{\T_{uv}}{\Omega^2}\lesssim \dfrac{\|f_0\|_{L^{\infty}_{x,p}}}{\exp(\frac{1}{4\sqrt{2}M}v)}, \qquad\quad \T_{vv}\lesssim \dfrac{\|f_0\|_{L^{\infty}_{x,p}}}{\exp(\frac{1}{4\sqrt{2}M}v)},$$ for all $x\in \{r\leq R\}$. Similar estimates hold for the other components of the energy-momentum tensor. 
\end{theorem}

\begin{remark}
The decay in time for the energy-momentum tensor in the far-away region is slower than the one stated in Theorem \ref{theorem_decay_fast}. By considering massive Vlasov fields supported all the way up to $\{E=1\}$, the decay rate in the far-away region changes drastically. The decay rate in the far-away region is closely related with the property that for outgoing orbits on $\{E=1\}$, we have $r(t)\sim t^{\frac{2}{3}}$ and $p^r(t)\sim t^{-\frac{1}{3}}$. This decay rate is slower than the one for massive Vlasov fields on Minkowski. 
\end{remark}

\subsection{Decay for Vlasov fields supported up to the boundary of $\mathcal{D}$}

The main result of this article shows quantitative decay for the components of the energy-momentum tensor of massive Vlasov fields supported on $\D$. We recall that the dispersive region $\mathcal{D}$ is defined by $$\D:=\mathrm{clos~} \Big\{(x,p)\in \mathcal{P}: \text{$\gamma_{x,p} $ \,crosses \,$\mathcal{H}^{+}\cup \mathcal{H}^{-}$ \,or  \, $r(\gamma_{x,p}(s))\to+\infty$ \,as\, $s\to \pm\infty$} \Big\},$$ where $\pi\colon \mathcal{P}\to \mathcal{E}$ is the canonical projection, and $\gamma_{x,p}$ is the unique geodesic with initial data $(x,p)$.

The dispersive region $\D$ can be characterised in terms of the radii $r_{\pm}(\ell)$ of the unique spheres $\{r=r_{\pm}(\ell)\}$, where circular orbits with angular momentum $\ell$ are contained. For all $\ell\in [2\sqrt{3}M,\infty)$, there exist geodesics with angular momentum $\ell$ contained in $\{r=r_{\pm}(\ell)\}$, where $r_{\pm}(\ell)$ are defined as the roots of 
\begin{equation}\label{defn_radius_trapped_sph_statem}
r^2-\frac{\ell^2}{M}r+3\ell^2=0,
\end{equation}
with $r_-(\ell)\leq r_+(\ell)$. We recall that the geodesics in $\{r=r_{\pm}(\ell)\}$ are called circular orbits. We denote the particle energy of the circular orbits in $\{r=r_{\pm}(\ell)\}$ by $E_{\pm}(\ell)$. In terms of these quantities, the region $\D$ can be characterised as  
\begin{align}
    \D&=\Big\{(x,p)\in \P: \text{ $\ell(x,p)\geq 4M$ such that if $E(x,p)<1$ then $r<r_-(\ell)$} \Big\}\label{ide_explicit_charact_main_result_1}\\
    &\qquad \cup\, \Big\{(x,p)\in \P: \text{ $\ell(x,p)< 4M$ such that if $E(x,p)\leq E_-(\ell)$ then $r \leq r_-(\ell)$} \Big\}.\nonumber
\end{align}
See Subsection \ref{subs_decom_phsp} for a proof of this. The set $\D$ is invariant under the geodesic flow since $E(x,p)$ and $l(x,p)$ are integrals of motion. In the region $\{E<1,~ \ell<4M\}$, geodesics can spend arbitrarily long periods of time in the far-away region before crossing $\mathcal{H}^+$. This property of the geodesic flow in Schwarzschild is at the heart of the leading order contribution in the decay estimates stated below.

Let $\underline{C}_{\mathrm{in}}\cup C_{\mathrm{out}}$ be a (bifurcate) initial null hypersurface, such that $\underline{C}_{\mathrm{in}}$ terminates at $\mathcal{H}^+$, and $C_{\mathrm{out}}$ goes out to $\mathcal{I}^+$. We define the subset $\Sigma\subset \D$ over the initial hypersurface $\underline{C}_{\mathrm{in}}\cup C_{\mathrm{out}}$ given by $$\Sigma:=\pi^{-1}(\underline{C}_{\mathrm{in}}\,\cup\,  C_{\mathrm{out}})\, \cap\, \D. $$ For almost every $(x,p)\in\Sigma$, the geodesic $\gamma_{x,p}$ either crosses the future event horizon $\mathcal{H}^+$ or escapes to infinity towards the future, see Subsection \ref{subs_decom_phsp}. We will consider Vlasov fields for which the initial distribution functions are supported on $\Sigma$. Let $\chi_{\mathcal{D}}\colon \mathcal{P}\to \R$ be the characteristic function of the set $\mathcal{D}$. For completeness, we write again the main result of the article.

\begin{theorem}\label{theorem_decay_slow}
Let $f$ be the solution of \eqref{vlasov_eqn_defn_intro} on the exterior of Schwarzschild spacetime arising from continuous compactly supported initial data $f_0$. Let $(u,v)$ be the Eddington--Finkelstein double null pair. For $R>2M$ sufficiently large, the components of the energy-momentum tensor $\T_{\mu\nu}[f\chi_{\mathcal{D}}]$ of the Vlasov field $f\chi_{\mathcal{D}}$ satisfy 
\begin{equation}\label{main_thm_main_unbounded_estm}
\T_{uu}\lesssim \dfrac{\|f_0\|_{L^{\infty}_{x,p}}}{u^{\frac{1}{3}}r^2},\qquad\quad \T_{uv}\lesssim \dfrac{\|f_0\|_{L^{\infty}_{x,p}}}{u^{\frac{1}{3}}r^2}, \qquad\quad \T_{vv}\lesssim \dfrac{\|f_0\|_{L^{\infty}_{x,p}}}{u^{\frac{1}{3}}r^2},
\end{equation}
for all $x\in  \{r\geq R\}$, and
\begin{equation}\label{main_thm_main_bounded_estm}
\frac{\T_{uu}}{\Omega^4}\lesssim \dfrac{\|f_0\|_{L^{\infty}_{x,p}}}{v^{\frac{1}{3}}},\qquad \quad
 \frac{\T_{uv}}{\Omega^2}\lesssim \dfrac{\|f_0\|_{L^{\infty}_{x,p}}}{v^{\frac{1}{3}}}, \qquad\quad \T_{vv}\lesssim \dfrac{\|f_0\|_{L^{\infty}_{x,p}}}{v^{\frac{1}{3}}},
\end{equation}
for all $x\in \{r\leq R\}$. Similar estimates hold for the other components of the energy-momentum tensor.
\end{theorem}

\begin{remark}
The decay rates for the components of $\T_{\mu\nu}[f\chi_{\mathcal{D}}]$ arise from the behaviour of the geodesic flow in a neighbourhood of the set $\{ E=1,\, \ell<4M\}$. In this neighbourhood, geodesics can spend arbitrarily long periods of time in the far-away region before crossing $\mathcal{H}^+$. We find that the components of $\T_{\mu\nu}[f\chi_{\mathcal{D}}]$ decay with a rate dictated by the behaviour of the momentum coordinate $p^r(t)\sim t^{-\frac{1}{3}}$ for outgoing particles on $\{E=1\}$. By the turning points of geodesics in $\{ E<1,\, \ell<4M\}$, this behaviour creates the leading order contribution in the decay estimates for the energy-momentum tensor in the whole spacetime. 
\end{remark}

\begin{remark}
The techniques used to show decay estimates for the Vlasov fields studied in this paper are \underline{not suitable} to address massive Vlasov fields on the closure of the complement $\mathcal{P}\setminus \mathcal{D}$ of phase space. 
\end{remark}

\subsection{Summary of the proof of Theorem \ref{theorem_decay_slow}}\label{subsection_summary_pf_sect_main}

Let us summarise the strategy of the proof of Theorem~\ref{theorem_decay_slow}. Along the way, we will also explain the key points in the proof of Theorem \ref{theorem_decay_fast} and Theorem \ref{theorem_decay_bded}.

\subsubsection{Step 1: Characterisation of the dispersive region $\mathcal{D}$}

As a first step, we show the \emph{explicit characterisation} \eqref{ide_explicit_charact_main_result_1} of the dispersive region $\mathcal{D}$. By definition of $\D$, we need to identify the geodesics in Schwarzschild that either cross $\mathcal{H}^+$, or escape to infinity towards the future. By the complete integrability of the geodesic flow, we can identify this class of orbits by studying the radial geodesic flow $(r(s),p^r(s))$ with fixed angular momentum. Remarkably, the radial geodesic flow $(r(s),p^r(s))$ with fixed angular momentum is determined by an autonomous ode system that \emph{decouples} from the rest of the geodesic equations $$\frac{\d r}{\d s}=p^r,\qquad \quad \frac{\d p^r}{\d s}=\frac{\ell^2}{r^3}\Big(1-\frac{3M}{r}\Big). $$

The radial geodesic flow $(r(s),p^r(s))$ can be integrated by using the mass-shell relation in the form $$E^2=(p^r)^2+V_{\ell}(r),\qquad \quad V_{\ell}(r):=\Big(1+\frac{\ell^2}{r^2}\Big)\Big(1-\frac{2M}{r}\Big).$$ We show the explicit form \eqref{ide_explicit_charact_main_result_1} of $\D$, by studying the shape of the radial potential $V_{\ell}(r)$ for any fixed $\ell \geq 0$.

\subsubsection{Step 2: Structure of the trapped set in $\mathcal{D}$}

The decay estimates for the components $\mathrm{T}_{\mu\nu}[f]$ of the energy-momentum tensor follow by proving decay in time of the momentum support of $f$. We will show the decay of the momentum support of the Vlasov field, by exploiting the concentration of the momentum support in suitable distributions of phase space associated to the trapped set. For this, we will need a careful study of the structure of the trapped set and the associated stable manifolds.

\emph{Unstable trapping and degenerate trapping at ISCO}. These forms of trapping occur on the spheres of trapped orbits $\{r=r_-(\ell)\}$, where $r_-(\ell)$ is defined by \eqref{defn_radius_trapped_sph_statem}. Unstable trapping holds for any value of the angular momentum $\ell$ in $(2\sqrt{3}M,\infty)$. For $\ell> 2\sqrt{3}M$, unstable trapping holds at the energy level $\{E=E_-(\ell)\}$, where
\begin{equation}\label{iden_good_cons_quant_unst_trapp}
H_{\ell}(E):=\frac{2M}{1-E^2}-\frac{2M}{1-E^2_-(\ell)}=\frac{r^3}{r^2-\frac{\ell^2}{2M}r+\ell^2}\Big(\frac{(p^r)^2}{1-E^2}+\Big(1+\frac{a(\ell)}{r}\Big)\Big(1-\frac{r_-(\ell)}{r}\Big)^2\Big)
\end{equation}
is equal to zero. The identity \eqref{iden_good_cons_quant_unst_trapp} follows by the mass-shell relation. The spheres of trapped orbits $\{r=r_-\}$ with $r_-\in [4M,6M)$, and the associated homoclinic orbits, are at the boundary $\partial\D$. The homoclinic orbits are tight to the existence of the closure $\mathcal{B}$  of the complementary region of phase space, where orbits are bound. On the contrary, the trapped orbits over the spheres $\{r=r_-\}$ with $r_-\in (3M,4M)$, are in $\mathrm{int}\, \mathcal{D}$.

On the other hand, degenerate trapping at ISCO occurs at $\{r=6M\}$ when $\ell=2\sqrt{3}M$. Degenerate trapping at ISCO holds at the energy level $\{E=\frac{2\sqrt{2}}{3}\}$, where
\begin{equation}\label{iden_good_cons_quant_deg_trapp}
H_{2\sqrt{3}M}(E):=\frac{2M}{1-E^2}-18M=\frac{r^3}{r^2-6Mr+12M^2}\Big(\frac{(p^r)^2}{1-E^2}-\Big(\frac{6M}{r}-1\Big)^3\Big)
\end{equation}
is equal to zero. The identity \eqref{iden_good_cons_quant_deg_trapp} follows by the mass-shell relation. The degenerate trapping arises as the bifurcation point when $\ell=2\sqrt{3}M$: recall that homoclinic orbits exist for $\ell\in (2\sqrt{3}M,4M)$, and there is no trapping for $\ell<2\sqrt{3}M$. The orbits in the sphere $\{r=6M\}$ lie at the boundary $\partial\D$.

\emph{Parabolic trapping at infinity}. To derive decay estimates for the components of $\mathrm{T}_{\mu\nu}[f]$, it is also relevant to consider the phenomenon of parabolic trapping at infinity. Parabolic trapping at infinity occurs for any value of the angular momentum $\ell$ at the energy level $\{E=1\}$, where
\begin{equation}\label{iden_good_cons_quant_parab_trapp}
E^2-1=(p^r)^2-\dfrac{2M}{r^3}\Big(r^2-\dfrac{\ell^2}{2M}r+\ell^2\Big)
\end{equation}
is equal to zero. The identity \eqref{iden_good_cons_quant_parab_trapp} follows by the mass-shell relation. The spheres of trapped orbits at infinity with $\ell\geq 4M$, are at the boundary $\partial\D$. On the contrary, the trapped orbits over the spheres at infinity with angular momentum $\ell<4M$ are in the interior $\mathrm{int}\, \mathcal{D}$.

\emph{Role of the stable manifolds}. For every form of trapping, we are interested in the corresponding stable manifolds. We will show that the momentum support of the distribution function concentrates on the tangent space of these submanifolds of $\mathcal{P}$. Suitable concentration estimates will be the main ingredient to show our main result. Key difficulties in the analysis of the geodesic flow in a neighbourhood of the stable manifolds comes from two bifurcations of the radial dynamics. These bifurcations occur when $\ell=2\sqrt{3}M$ and $\ell=4M$.

\subsubsection{Step 3: Concentration estimates on the stable manifolds}

Let us give an overview of the main difficulties to show the different concentration estimates. We discuss these estimates for the three forms of trapping.

\emph{Unstable trapping}. The radial dynamics in a neighbourhood of the spheres of trapped orbits $\{r=r_-(\ell)\}$ bifurcate when $\ell$ goes through the value $4M$. Let us consider separately the regions $\{\ell\geq 4M\}$ and $\{\ell\leq 4M\}$.

\begin{itemize}
\item \emph{The region $\{\ell\geq 4M\}$}. Unstable trapping occurs at the sphere of trapped orbits $\{r=r_-(\ell)\}$. We use suitable defining functions for the stable manifolds to obtain exponential decay in a neighbourhood of $\{r=r_-(\ell)\}$. Later, we propagate these estimates in the bounded region. In the far-away region, we use the Minkowskian asymptotics of the geodesic flow. 
\item \emph{The region $\{\ell\leq 4M\}$}. Unstable trapping holds at the sphere of trapped orbits $\{r=r_-(\ell)\}$. Here, \emph{trapping is not constrained to a fixed sphere}, since there are homoclinic orbits contained in $\{r\geq r_-(\ell)\}$. In this region, we integrate the radial flow using \eqref{iden_good_cons_quant_unst_trapp}. We estimate the flow in a neighbourhood of the whole homoclinic orbits.
\end{itemize}

\emph{Degenerate trapping at ISCO}. The radial dynamics in a neighbourhood of the sphere of trapped orbits $\{r=r_-(\ell)\}$ bifurcate when $\ell$ goes through $2\sqrt{3}M$. Let us consider separately the regions $\{\ell\geq 2\sqrt{3}M\}$ and $\{\ell\leq 2\sqrt{3}M\}$.

\begin{itemize}
\item \emph{The region $\{\ell\geq 2\sqrt{3}M\}$}. \emph{Unstable trapping degenerates} when $\ell\to 2\sqrt{3}M$. The radial flow with $\ell=2\sqrt{3}M$ has a degenerate form of trapping at $\{r=6M\}$. Furthermore, trapping for $\ell>2\sqrt{3}M$ is not constrained to a fixed sphere, since there are homoclinic orbits associated to the spheres of trapped orbits $\{r=r_-(\ell)\}$ with $\ell \in (2\sqrt{3}M,4M)$. In this region, we integrate the radial flow using \eqref{iden_good_cons_quant_unst_trapp}. We obtain estimates for the geodesic flow that \emph{do not degenerate} as $\ell\to 2\sqrt{3}M$.
\item \emph{The region $\{\ell\leq 2\sqrt{3}M\}$}. Trapping only occurs when $\ell=2\sqrt{3}M$. There is no trapping when $\ell<2\sqrt{3}M$. However, \emph{trapping almost occurs as} $\ell\to 2\sqrt{3}M$. In this regime, we consider \emph{spheres of almost trapped orbits} of the form $\{r=r_-(\ell)\}$ for $\ell<2\sqrt{3}M$. We estimate the flow in a neighbourhood of the spheres of almost trapped orbits. Crucially, the estimates \emph{do not degenerate} as $\ell\to 2\sqrt{3}M$.
\end{itemize}

\emph{Parabolic trapping at infinity}. The radial dynamics in a neighbourhood of the parabolic energy level $\{E=1\}$ bifurcate when $\ell$ goes through $4M$. Let us consider separately the regions $\{\ell\geq 4M\}$ and $\{\ell\leq 4M\}$.

\begin{itemize}
\item \emph{The region $\{\ell\geq 4M\}$}. Parabolic trapping only occurs in the boundary of the dispersive region $\D$. In a neighbourhood of the energy level $\{E=1\}$ in the far-away region, we only find outgoing geodesics escaping to infinity. We show that the asymptotics of the orbits in the energy level $\{E=1\}$ are \emph{slower} than the corresponding Minkowskian asymptotics, in the sense that $r(t)\sim t^{\frac{2}{3}}$ and $p^r(t)\sim t^{-\frac{1}{3}}$.  
\item \emph{The region $\{\ell\leq 4M\}$}. In a neighbourhood of $\{E=1\}$, we find bounded orbits that spend arbitrarily long periods of time in the far-away region before crossing $\mathcal{H}^+$. This effect makes the \emph{slow dispersion towards infinity for $E\sim 1$ become relevant even in the bounded region}. The concentration estimates in this region create the leading order contribution in the decay estimates for the components of $\mathrm{T}_{\mu\nu}[f]$. In this region, we integrate the radial flow using \eqref{iden_good_cons_quant_parab_trapp}. We obtain estimates for the geodesic flow that \emph{do not degenerate} as $\ell\to 4M$.
\end{itemize}

We will show Theorem \ref{theorem_decay_slow} by putting together the concentration estimates in the cases considered above.

\section{The timelike geodesic flow}\label{section_timelike_geodesics}

In this section, we begin the study of the timelike geodesic flow in the subset of the mass-shell $\mathcal{P}$ over the exterior of Schwarzschild. The \emph{timelike geodesic flow in Schwarzschild spacetime} is determined by the geodesic equations
\begin{align}
\begin{aligned}\label{geodesic_eqns_schwarzschild}
 \frac{\d u}{\d s}=p^u ,\qquad  \quad \frac{\d p^u}{\d s}&=\dfrac{2M}{r^2}(p^u)^2-\frac{\ell^2}{2r^3},\\
  \frac{\d v}{\d s}=p^v ,\qquad \quad \frac{\d p^v}{\d s}&=-\dfrac{2M}{r^2}(p^v)^2+\frac{\ell^2}{2r^3},\\
  \frac{\d\theta}{\d s}=p^{\theta} ,\qquad  \quad\frac{\d p^{\theta}}{\d s}&=-\dfrac{2p^rp^{\theta}}{r}+\sin\theta \cos \theta (p^{\phi})^2,\\
  \frac{\d\phi}{\d s}=p^{\phi} ,\qquad \quad \frac{\d p^{\phi}}{\d s}&=-\dfrac{2p^rp^{\phi}}{r}-2\cot \theta p^{\theta}p^{\phi},
\end{aligned}
\end{align}
in terms of the Christoffel symbols $\Gamma_{\beta \gamma}^{\alpha}$ of Schwarzschild in Eddington--Finkelstein double null coordinates. Here, we have parametrised the geodesic flow by the parameter $s$. 

\subsection{Complete integrability of the geodesic flow}

A fundamental property of the geodesic flow in Schwarzschild is its \emph{complete integrability}. This feature of the geodesic flow holds because of the spherical symmetry and the stationarity of spacetime.

\subsubsection{Hamiltonian structure of the geodesic flow}

The geodesic flow in Schwarzschild can be viewed as a Hamiltonian flow in the cotangent bundle $T^*\mathcal{E}$ with respect to its standard symplectic structure. We recall the \emph{canonical symplectic form} $\omega$ on $T^*\mathcal{E}$, given by $$\omega:=\d x^{\alpha}\wedge \d p_{\alpha},$$ in a canonical coordinate system. The two-form $\omega$ is known as the \emph{Poincar\'e two-form}. In this setup, the Hamiltonian $\mathcal{H}\colon T^{*}\mathcal{E}\to [0,\infty)$ of the geodesic flow is given by
$$\mathcal{H}(x,p):=\dfrac{1}{2}g^{\alpha\beta}p_{\alpha}p_{\beta}.$$ We note that $\mathcal{H}(x,p)$ is an integral of motion that is independent of the choice of coordinates. The Hamiltonian $\mathcal{H}(x,p)$ induces a Hamiltonian flow in the cotangent bundle that is known as the \emph{cogeodesic flow}. By using the canonical isomorphism between $T^*\mathcal{E}$ and $T\mathcal{E}$, the cogeodesic flow corresponds to the usual geodesic flow on $T\mathcal{E}$.

\subsubsection{The conserved quantities}

By the stationarity and the spherical symmetry of Schwarzschild, there exist three non-trivial quantities that are conserved along the geodesics in spacetime. We recall that the Hamiltonian $\mathcal{H}(x,p)$ is also conserved along the geodesic flow.

\emph{Associated to stationarity}. We define the \emph{particle energy function} $E\colon T^*\mathcal{E}\to [0,\infty)$ by $$E(x,p):=-g_M(p,\partial_t)=-p_t.$$ Note that $E(x,p)$ is independent under change of coordinates. Moreover, the particle energy function satisfies $E(x,p)> 0$ for every $(x,p)\in \mathcal{P}$. The particle energy function $E(x,p)$ is conserved along the geodesic flow, since $\partial_t$ is a Killing vector field. In particular, the \emph{particle energy $E$ of a geodesic} $\gamma$ is well-defined.

\emph{Associated to spherical symmetry}. We define the \emph{azimuthal angular momentum function} $\ell_{\phi}\colon T^*\mathcal{E}\to \R$ by $$\ell_{\phi}(x,p):=g_M(p,\partial_\phi)=p_{\phi}.$$ The azimuthal angular momentum function $\ell_{\phi}(x,p)$ is conserved along the geodesic flow, since $\partial_{\phi}$ is a Killing vector field. Thus, the \emph{azimuthal angular momentum $\ell_{\phi}$ of a geodesic} $\gamma$ is well-defined. 

Let $Q$ be the symmetric $(0,2)$ Killing tensor field $$Q=\partial_{\theta}\otimes \partial_{\theta}+(\sin^{-2}\theta)\partial_{\phi}\otimes \partial_{\phi}=\sum_{i=1}^3 \mathbf{\Omega}_i\otimes \mathbf{\Omega}_i, $$ where $\mathbf{\Omega}_1$, $\mathbf{\Omega}_2$, $\mathbf{\Omega}_3$, are the usual spherically symmetric Killing vector fields \eqref{killing_sph_symm} of Schwarzschild. We define the \emph{total angular momentum function} $\ell \colon T^*\mathcal{E}\to [0,\infty)$ by $$\ell(x,p):=Q(p,p)=\sqrt{p_{\theta}^2+\dfrac{1}{\sin^2\theta}p_{\phi}^2}.$$ Note that $\ell(x,p)$ is independent of the chosen angular coordinate system. By the Killing tensor field $Q$, the total angular momentum function $\ell(x,p)$ is conserved along the geodesic flow. Hence, the \emph{total angular momentum $\ell$ of a geodesic} $\gamma$ is well-defined. From now on, we refer to $\ell(x,p)$ simply as the angular momentum.

\subsubsection{Integrability of the geodesic flow in Schwarzschild}

 The geodesic flow in Schwarzschild belongs to a special class of Hamiltonian systems known as \emph{integrable systems}. This class is characterised by the property of possessing enough integrals of motion, so that the Hamiltonian flow can be \emph{integrated}. 

A smooth function $F\colon T\mathcal{E}\to \R$ is an \emph{integral of motion} for the geodesic flow, if $\d F(x,p)X(x,p)=0$ and $\d F(x,p)\neq 0$ for all $(x,p)\in T\mathcal{E}$. By a direct application of the chain rule, an integral of motion $F$ is constant along every geodesic in Schwarzschild. The functions $\mathcal{H}(x,p),$ $E(x,p)$, $\ell(x,p)$, $\ell_{\phi}(x,p)$, considered in the previous subsection, are examples of integrals of motion in Schwarzschild. The geodesic flow in a four-dimensional Lorentzian manifold is called \emph{completely integrable}, if it possesses four smooth integrals of motion $F_j$ having the property that $\{\d F_j\}$ are linearly independent, and $\{F_j,F_k\}=0$ for all $j$ and $k$.

\begin{proposition}\label{prop_com_int}
The geodesic flow in the cotangent bundle of Schwarzschild $(T^{*}\mathcal{E},\omega, \mathcal{H})$ is a completely integrable Hamiltonian flow.
\end{proposition}

Proposition \ref{prop_com_int} can be proved by considering the four independent conserved quantities $\mathcal{H}$, $E$, $\ell$, and $\ell_{\phi}$. See \cite{RS24} for a detailed proof of this proposition.

\subsection{The trapped set}\label{subsec_radial_geo_flow}

The phenomenon of \emph{trapping} refers to the existence of geodesics on the black hole exterior that do not cross $\mathcal{H}^{+}\cup \mathcal{H}^{-}$, and do not escape to infinity neither towards the past nor the future. The subset of the mass-shell where these orbits lie is the so-called \emph{trapped set}. Bounded geodesics that do not cross $\mathcal{H}^{+}\cup \mathcal{H}^{-}$, are called \emph{trapped orbits}. Specifically, we define the \emph{trapped set} $\Gamma\subset \D$ as
\begin{equation}\label{eq_trapped_defn}
\Gamma:=\Big\{ (x,p)\in \mathcal{D}: \text{$\gamma_{x,p} $ is bounded and does not cross $\mathcal{H}^{+}\cup \mathcal{H}^{-}$}\Big\}.
\end{equation}
We also say that a geodesic $\gamma_{x,p}$ is \emph{future-trapped} if the curve $\gamma_{x,p}([0,\infty))$ is bounded and does not cross $\mathcal{H}^{+}$. On the other hand, we say that a geodesic $\gamma_{x,p}$ is \emph{past-trapped} if the curve $\gamma_{x,p}((-\infty,0])$ is bounded and does not cross $\mathcal{H}^{-}$. According to the definition \eqref{eq_trapped_defn}, the trapped set $\Gamma$ is a subset of the dispersive region $\D$. In other words, we do \underline{not include} the region $\B$ in the trapped set. We make this convention to differentiate the phenomena of \emph{trapping} and \emph{confinement}. We note that confinement holds for the geodesic flow in $\mathcal{B}$. 

In the following, we will study the class of trapped orbits by considering the geodesic flow in the radial variable. Remarkably, the geodesic flow in the radial variable decouples from the rest of the geodesic equations. By the mass-shell relation $g_x(p,p)=-1$, the particle energy function $E$ can be written as 
\begin{equation}\label{identity_particle_energy_angular_momentum}
        E^2=(p^r)^2+\Big(1-\dfrac{2M}{r}\Big)\Big(1+\dfrac{\ell^2}{r^2}\Big),\qquad \quad V_{\ell}(r):=\Big(1-\dfrac{2M}{r}\Big)\Big(1+\dfrac{\ell^2}{r^2}\Big),
\end{equation} 
where $V_{\ell}(r)$ is the \emph{potential of the radial flow}. Differentiating \eqref{identity_particle_energy_angular_momentum} along the geodesic flow, we find the \emph{radial geodesic flow with angular momentum $\ell$} determined by 
\begin{equation}\label{r_second_derivative_equation}
    \dfrac{\d r}{\d s}=p^r,\qquad \quad \dfrac{\d p^r}{\d s}=-\dfrac{M}{r^4}\Big(r^2-\dfrac{\ell^2}{M}r+3\ell^2\Big).
\end{equation} 
We often refer to the flow \eqref{r_second_derivative_equation} simply as the \emph{radial geodesic flow}, whenever is clear the value of the angular momentum. We refer to the second equation in \eqref{r_second_derivative_equation} as the \emph{radial geodesic equation}. According to \eqref{r_second_derivative_equation}, there exist circular orbits contained in hypersurfaces of fixed radii. The geodesics in these hypersurfaces are examples of trapped orbits. We refer to these hypersurfaces as the \emph{spheres of trapped orbits}.

\subsubsection{The spheres of trapped orbits}

Let us study the orbits contained in spheres of fixed radii. 

\begin{proposition}\label{propspherestrap}
For all $\ell\geq 2\sqrt{3}M$, the unique spheres $\mathcal{S}^{\pm}(\ell)$ containing geodesics with angular momentum $\ell$ are given by $$\mathcal{S}^{\pm}(\ell):=\Big\{(x,p)\in \P: \ell(x,p)=\ell,\quad r=r_{\pm}(\ell):=\dfrac{\ell^2}{2M}\Big(1\pm \sqrt{1-\dfrac{12M^2}{\ell^2}}\Big),\quad p^r=0\Big\}.$$ Moreover, the particle energies of the orbits in $\mathcal{S}^{\pm}(\ell)$ are 
\begin{equation}\label{partenertrappsph}
E_{\pm}(\ell):=\dfrac{\ell}{\sqrt{Mr_{\pm}(\ell)}}\Big(1- \dfrac{2M}{r_{\pm}(\ell)}\Big).
\end{equation}  
\end{proposition}

\begin{proof}
The stationary points of the radial flow \eqref{r_second_derivative_equation} satisfy $p^r=0$ and $r^2-\frac{\ell^2}{M}r+3\ell^2=0$. For $\ell\geq 2\sqrt{3}M$, the zeroes of this quadratic polynomial are given by $r_{+}(\ell)\geq r_{-}(\ell)>2M$. We note that $r^2-\frac{\ell^2}{M}r+3\ell^2$ is strictly positive for $\ell<2\sqrt{3}M$. The zeroes of $r^2-\frac{\ell^2}{M}r+3\ell^2$ correspond to orbits with angular momentum $\ell$ contained in $\{r=r_{\pm}(\ell)\}$. In particular, we have \eqref{partenertrappsph} by using the mass-shell relation \eqref{identity_particle_energy_angular_momentum}. 
The identity \eqref{partenertrappsph} accounts for solving the $t$-geodesic equation. The geodesic equations for the spherical variables can be solved independently by considering the geodesic flow on the sphere of radius $r=r_{\pm}(\ell)$. 
\end{proof}

In the next subsection, we will show that the spheres $\mathcal{S}^{-}(\ell)$ are contained in the trapped set $\Gamma$. For this reason, we call $\mathcal{S}^{-}(\ell)$ the \emph{spheres of trapped orbits}.

The radii $r_{\pm}(\ell)$ correspond to the critical points of the radial potential $V_{\ell}(r)$. For all $\ell>2\sqrt{3}M$, there exists a local maximum of $V_{\ell}(r)$ at $r_{-}(\ell)$, so $(r_{-}(\ell),0)$ is a \emph{hyperbolic fixed point} of the radial flow. The geodesics in $\{r=r_-(\ell)\}$ with $\ell>2\sqrt{3}M$ are called \emph{unstable circular orbits}. On the other hand, for all $\ell>2\sqrt{3}M$ there exists a local minimum of $V_{\ell}(r)$ at $r_{+}(\ell)$, so $(r_{+}(\ell),0)$ is an \emph{elliptic fixed point} of the radial flow. The geodesics in $\{r=r_+(\ell)\}$ with $\ell>2\sqrt{3}M$ are called \emph{stable circular orbits}. Finally, for $\ell=2\sqrt{3}M$ there exists an inflection point of $V_{\ell}(r)$ at $r_{+}(2\sqrt{3}M)=r_{-}(2\sqrt{3}M)=6M$, so $(6M,0)$ is a \emph{degenerated fixed point} of the radial flow. The geodesics in $\{r=6M\}$ with $\ell=2\sqrt{3}M$ are called the \emph{innermost stable circular orbits}.

\subsection{Decomposition of phase space}\label{subs_decom_phsp}

In this subsection, we show the main properties of the dispersive region $\D$ of the mass-shell. In particular, we prove the explicit characterisation \eqref{ide_explicit_charact_main_result_1} of the dispersive region $\D$. We also show a similar characterisation for the complementary region $\B$.

\subsubsection{Characterisation of the dispersive region $\D$}
 
Let us decompose the subset $\mathcal{P}$ of the mass-shell over the exterior of Schwarzschild spacetime as 
\begin{equation*}
\mathcal{P}=\D\,\cup \,\B,
\end{equation*}
in terms of the subsets 
\begin{align*}
    \D&:=\mathrm{clos~} \Big\{(x,p)\in \mathcal{P}: \text{$\gamma_{x,p} $ \,crosses \,$\mathcal{H}^{+}\cup \mathcal{H}^{-}$ \,or  \,
$r(\gamma_{x,p}(s))\to+\infty$ \,as \,$s\to \infty$} \Big\},\\
    \B&:=\mathrm{clos~}  \Big\{(x,p)\in \mathcal{P}: \text{$\gamma_{x,p} $ is bounded and does not cross $\mathcal{H}^{+}$ nor $\mathcal{H}^{-}$} \Big\},
\end{align*}
where $\pi\colon \mathcal{P}\to \mathcal{E}$ is the canonical projection, and $\gamma_{x,p}$ is the unique geodesic with initial data $(x,p)$ on the black hole exterior. The decomposition $\mathcal{P}=\D\,\cup\, \B$ holds by the definition of the sets $\D$ and $\B$.  

In the following, we identify \emph{explicitly} the sets $\D$ and $\B$, in terms of the particle energy function $E(x,p)$, and the total angular momentum function $\ell(x,p)$. For this purpose, we will study the specific form of the radial potential $V_{\ell}(r)$ for all $\ell\geq 0$.

\begin{proposition}\label{prop_decomp_mass_shell}
The sets $\D$ and $\B$ are invariant under the geodesic flow. Moreover, the sets $\D$ and $\B$ are characterised as 
\begin{align}
    \D&=\Big\{(x,p)\in \mathcal{P}: \text{ $\ell(x,p)\geq 4M$ such that if $E(x,p)<1$ then $r<r_-(\ell)$} \Big\}\label{charact_prop_disp}\\
    &\qquad \cup \,\Big\{(x,p)\in \mathcal{P}: \text{ $\ell(x,p)< 4M$ such that if $E(x,p) \leq E_-(\ell)$ then $r \leq r_-(\ell)$} \Big\},\nonumber
\end{align}
and
\begin{align}
    \B&=\Big\{(x,p)\in \mathcal{P}: \text{ $\ell(x,p)\geq 4M$ such that if $E(x,p)<1$ then $r>r_-(\ell)$} \Big\}\cup_{\ell\in (4M,\infty)}\mathcal{S}^-(\ell) \label{charact_prop_bound}\\
    &\qquad \cup\, \Big\{(x,p)\in \mathcal{P}: \text{ $\ell(x,p)< 4M$ such that if $E(x,p)\leq E_-(\ell)$ then $r \geq r_-(\ell)$} \Big\}.\nonumber
\end{align}
Finally, the intersection $\D\,\cap\,\B$ is given by 
\begin{align}
\D\,\cap\,\B&=\Big\{(x,p)\in\mathcal{P} : \ell(x,p)\in[2\sqrt{3}M,4M],\quad E(x,p)=E_-(\ell), \quad  r\geq r_-(\ell)\Big\}\label{charact_intersect_degenera}\\
&\qquad \cup_{\ell\in (4M,\infty)}\mathcal{S}^-(\ell)\,\cup \,\Big\{(x,p)\in\mathcal{P} : \ell(x,p)\in[4M,\infty),\quad E(x,p)=1, \quad  r\geq r_-(\ell)\Big\}.\nonumber
\end{align}
\end{proposition}

\begin{proof}
In the following, we will consider three different regions of the mass-shell. We consider the sets of points $(x,p)\in\mathcal{P}$ for which the geodesics $\gamma_{x,p}$ have angular momentum $\ell$ in $[0,2\sqrt{3}M)$, in $(2\sqrt{3}M,4M]$, or in $[4M,\infty)$. In each of these cases, we first show that if $(x,p)$ belongs to the RHS of \eqref{charact_prop_disp}--\eqref{charact_prop_bound}, then $(x,p)$ belongs to $\D$ or $\B$, respectively. We recall that by definition of the radial potential $V_{\ell}(r)$, we have 
\begin{equation}\label{derivative_radial_pot_first}
\frac{\d}{\d r}V_{\ell}(r)=\frac{2M}{r^4}\Big(r^2-\frac{\ell^2}{M}r+3\ell^2\Big).
\end{equation}
Moreover, we note that $V_{\ell}(2M)=0$, and $\lim_{r\to \infty} V_{\ell}(r)=1$.

\emph{Case 1: $\ell(x,p)\leq 2\sqrt{3}M$}. 
In this case the potential $V_{\ell}(r)$ is nondecreasing. In particular, there are no critical points. We proceed to study two subcases: orbits $\gamma_{x,p}$ with particle energy $E<1$ or $E\geq 1$.
\begin{itemize}
\item \emph{Subcase 1a: $E(x,p)< 1$}. A geodesic $\gamma_{x,p}$ with $E<1$ and $\ell< 2\sqrt{3}M$ crosses $\mathcal{H}^+$. Thus, if $E(x,p)\leq 1$ and $\ell(x,p)\leq  2\sqrt{3}M$, we have $(x,p)\in \D$. 
\item \emph{Subcase 1b: $E(x,p)\geq 1$}. On the other hand, geodesics with $E\geq 1$ and $\ell\leq 2\sqrt{3}M$, either cross the future event horizon $\mathcal{H}^{+}$, or escape to infinity towards the future. So, if $E(x,p)\geq 1$ and $\ell(x,p)\leq 2\sqrt{3}M$, we also have $(x,p)\in \D$. 
\end{itemize}
We note that $(x,p)\in \D \,\cap \, \B$ for trapped orbits $\gamma_{x,p}$ at $r=6M$ with $\ell(x,p)=2\sqrt{3}M$. Indeed, trapped orbits at $r=6M$ can be obtained as limit points of circular orbits on other spheres of trapped orbits.

\emph{Case 2: $\ell(x,p)\in (2\sqrt{3}M,4M]$}. In this case, the radial potential $V_{\ell}(r)$ has critical points at $r_{\pm}(\ell)$ with $E_-(\ell)=V_{\ell}(r_-(l))\leq 1$. We proceed to study two subcases: orbits $\gamma_{x,p}$ with $E<E_-(\ell)$ or $E\geq E_-(\ell)$.

\begin{itemize}
\item \emph{Subcase 2a: $E(x,p)<E_-(\ell)$}. The geodesics $\gamma_{x,p}$ with $E(x,p)\leq E_-(\ell)$ in $\{r\geq r_-(\ell)\}$ are bounded, so $(x,p)\in \B$. In contrast, a geodesic $\gamma_{x,p}$ with $E(x,p)< E_-(\ell)$ contained in $\{r < r_-(\ell)\}$ crosses $\mathcal{H}^+$. So, $(x,p)\in \D$ if $E(x,p)\leq E_-(\ell)$ in $\{r \leq r_-(\ell)\}$.  
\item \emph{Subcase 2b: $E(x,p)\geq E_-(\ell)$}. The geodesics $\gamma_{x,p}$ with particle energy $E(x,p)\in (E_-(\ell),1)$ cross $\mathcal{H}^+$. Thus, if $E(x,p)\in [E_-(\ell),1]$, then $(x,p)\in \D$. On the other hand, geodesics $\gamma_{x,p}$ with $E(x,p)\geq 1$ and $\ell\in (2\sqrt{3}M,4M)$, either cross the future event horizon $\mathcal{H}^{+}$, or escape to infinity towards the future. So, in this last case $(x,p)\in \D$. 
\end{itemize}
We note that for orbits $\gamma_{x,p}$ with particle energy $E(x,p)=E_-(\ell)$ contained in $\{r \geq r_-(\ell)\}$, we have $(x,p)\in \D\, \cap \, \B$, since $(x,p)$ can be obtained as limit points of $(x,p)$ with $E<E_-(\ell)$, and limit points of $(x,p)$ with $E>E_-(\ell)$.

\emph{Case 3: $\ell(x,p)\in [4M,\infty)$}. In this case, the radial potential $V_{\ell}(r)$ has also critical points at $r_{\pm}(\ell)$, but $E_-(\ell)=V_{\ell}(r_-(l))\geq 1$ instead. We proceed to study two subcases: orbits $\gamma_{x,p}$ with $E<1$ or $E\geq 1$.
\begin{itemize}
\item \emph{Subcase 3a: $E(x,p) <1$}. 
Geodesics $\gamma_{x,p}$ with $E(x,p)<1$ contained in $\{r>r_-(\ell)\}$ are bounded. So $(x,p)\in\B$ when $E(x,p)\leq 1$ in $r>r_-(\ell)$. In contrast, the orbits $\gamma_{x,p}$ with $E(x,p)<1$ contained in $\{r<r_-(\ell)\}$ cross $\mathcal{H}^+$. So, $(x,p)\in\D$ when $E(x,p)\leq 1$ in $\{r\leq r_-(\ell)\}$. 
\item \emph{Subcase 3b: $E(x,p)\geq 1$}. On the other hand, geodesics $\gamma_{x,p}$ with $E(x,p)\in (1,E_-(\ell))\,\cup \,(E_-(\ell),\infty)$ and $\ell\in (4M,\infty)$, either cross the future event horizon $\mathcal{H}^{+}$, or escape to infinity towards the future. So $(x,p)\in \D$ when $E(x,p)\in [1,\infty)$ and $\ell\in [4M,\infty)$. Moreover, the geodesics in the union $\cup_{\ell\in (4M,\infty)}\mathcal{S}^-(\ell)$ of trapped spheres is also contained in $\B$ by definition. 
\end{itemize}
We note that for orbits $\gamma_{x,p}$ with particle energy $E(x,p)=1$ contained in $\{r \geq r_-(\ell)\}$, we have $(x,p)\in \D\, \cap \, \B$, since $(x,p)$ can be obtained as limit points of $(x,p)$ with $E<1$, and limit points of $(x,p)$ with $E>1$. We also note that the union $\cup_{\ell\in (4M,\infty)}\mathcal{S}^-(\ell)$ of trapped spheres is also contained in the intersection $\D\, \cap \, \B$.

Putting together the previous three cases, we obtain that if $(x,p)$ belongs to the RHS of \eqref{charact_prop_disp}--\eqref{charact_prop_bound}, then $(x,p)$ belongs to either $\D$ or $\B$, respectively. A posteriori, it is straightforward to show the opposite, namely, that if $(x,p)$ belongs to either $\D$ or $\B$, then $(x,p)$ belongs to the RHS of \eqref{charact_prop_disp} or \eqref{charact_prop_bound}, respectively. The invariance of the sets $\mathcal{D}$ and $\mathcal{B}$ is a direct consequence of the characterisations \eqref{charact_prop_disp}--\eqref{charact_prop_bound}. The identity \eqref{charact_intersect_degenera} also follows from the previous analysis.
\end{proof}

\begin{remark}
In the characterisation of $\B$, the union $\cup_{\ell\in (4M,\infty)}\mathcal{S}^-(\ell)$ does not play any role for the purpose of studying massive Vlasov fields on $\B$. We note that the closure of this set is given by $\cup_{\ell\in [4M,\infty)}\mathcal{S}^-(\ell)$, which has measure zero on the mass-shell. 
\end{remark}

See \cite[Section 19]{Ch} for more information on the qualitative description of the timelike geodesic flow in Schwarzschild.

\subsubsection{Dispersive properties of the domain $\mathcal{D}$}\label{subsubsection_dispersive_prop_dispers_dom}

In this subsection, we show some fundamental properties of the dispersive region $\D$. We first show that $\D$ is the largest invariant open subset of $\mathcal{P}$, where we can expect to prove decay estimates for the massive Vlasov equation on Schwarzschild.

\begin{proposition}\label{prop_larg_inva}
The set $\mathring{\D}$ is the largest invariant open subset of $\mathcal{P}$, where for almost every $(x,p)$ the geodesic $\gamma_{x,p}$ either crosses $\mathcal{H}^+$, or escapes to infinity towards the future.
\end{proposition}

\begin{proof}
By contradiction, we suppose that there exists a non-trivial invariant open subset $\D'$ containing $\D$. By Proposition \ref{prop_decomp_mass_shell}, the mass-shell is decomposed as $\mathring{\D}\sqcup (\D\cap\B)\sqcup \mathring{B}$. Let $(x,p)\in \mathcal{D}'\setminus \D$. 

If $(x,p)\in \mathring{B}$, then there is a neighbourhood $A\subset \D'\cap \mathring{B}$ of $(x,p)$, since $\D'$ is open. We can assume that $A$ is invariant, since $\D'$ and $\mathring{B}$ are also invariant. We have obtained then an open set $A$ such that $A\subset\mathring{B}\subset \D'$. In particular, the orbits in $A$ are bounded and do not cross $\mathcal{H}^+$. Finally, we obtain the desired contradiction since $A$ is open. 

On the other hand, if $(x,p)\in \D\cap\B$, then there is a neighbourhood $A\subset \D'$ of $(x,p)$, since $\D'$ is open. Furthermore, there exists $(x',p')\in A\,\cap \,\mathring{B}$ since $\D\,\cap\, \B$ has codimension one. Then, there exists a neighbourhood $B\subset \D'\cap \mathring{\B}$ of $(x',p')$. We can assume that $B$ is invariant by the invariance of $\D'\cap \mathring{\B}$. The rest of the proof follows as in the case where $(x,p)\in \mathring{B}$, but using $(x',p')$ instead.
\end{proof}

We show now that massive Vlasov fields do not decay in general when supported on any invariant open set $\D'$ larger than $\D.$

\begin{proposition}\label{prop_stat_sol}
For every non-trivial invariant open set $\D'\subset \mathcal{P}$ containing the domain $\mathring{\D}$, there exists non-trivial compactly supported stationary solutions $f$ of the massive Vlasov equation such that $\supp(f)\subset \D'$.
\end{proposition}

\begin{proof}
By the proof of Proposition \ref{prop_larg_inva}, there exists an invariant open subset $A\subset \D'\cap \mathring{B}$. One can then easily show that there exists a stationary Vlasov field of the form $f(x,p)=\Phi(E(x,p),\ell(x,p))\chi(x,p)$ with $\supp(f)\subset \D'$. Here, $\Phi$ is a regular non-negative function, and $\chi$ a suitably chosen cut-off function. 
\end{proof}

By straightforward arguments Proposition \ref{prop_larg_inva} and Proposition \ref{prop_stat_sol} can be suitably modified to make similar statements for the set $\mathring{\Sigma}\subset \pi^{-1}(\underline{C}_{\mathrm{in}}\,\cup \,C_{\mathrm{out}})$ over the initial bifurcate hypersurface, and the class of initial data for the massive Vlasov equation on $\pi^{-1}(\underline{C}_{\mathrm{in}}\,\cup\, C_{\mathrm{out}})$. We will not pursue this here.

\subsubsection{Characterisation of the trapped set $\Gamma$}\label{subsub_charact_trapp}

Let us characterise the trapped set $\Gamma\subset \D$ of Schwarzschild, in terms of the particle energy $E(x,p)$, and the total angular momentum $\ell(x,p)$.

\begin{proposition}
The trapped set $\Gamma$ is characterised as 
\begin{align}
\Gamma&=\Big\{(x,p)\in\mathcal{D} : \ell(x,p)\in[2\sqrt{3}M,4M),\quad E(x,p)=E_-(\ell), \quad  r\geq r_-(\ell)\Big\}\\
&\qquad \cup\, \Big\{(x,p)\in\mathcal{D} : \ell(x,p)\in[4M,\infty),\quad E(x,p)=E_-(\ell), \quad  r= r_-(\ell)\Big\}.\nonumber
\end{align}
Moreover, the trapped set $\Gamma$ has measure zero on the mass-shell.
\end{proposition}

\begin{proof}
By the proof of Proposition \ref{prop_decomp_mass_shell}, there exists a non-trivial trapped set for every $\ell\geq 2\sqrt{3}M$. For $\ell\in [2\sqrt{3}M, 4M)$, the trapped set is composed by the points $(x,p)$ in the homoclinic orbits associated to the sphere $\mathcal{S}^-(\ell)$. In other words, we consider the points $(x,p)$ with $E(x,p)=E_-(\ell)$ in $\{r\geq r_-(\ell)\}$. This set includes the corresponding spheres of trapped orbits $\mathcal{S}^-(\ell)$. For $\ell\in [4M,\infty)$, the trapped set is composed only by the points $(x,p)$ in the spheres $\mathcal{S}^-(\ell)$. Finally, we observe that the subsets of $\Gamma$ with $\ell\in[2\sqrt{3}M, 4M)$ and $\ell\in[4M,\infty)$, have both codimension one. In particular, $\Gamma$ has measure zero on the mass-shell.
\end{proof}

Let us define the subset $\Sigma_{\mathrm{ft}}\subset \Sigma$ given by $$\Sigma_{\mathrm{ft}}:=\Big\{(x,p)\in \Sigma : \text{$\gamma_{x,p}$ is future-trapped}  \Big\}.$$ Let us show now that $\Sigma_{\mathrm{ft}}$ has measure zero with respect to the induced volume form on $\Sigma$. 

\begin{proposition} 
The set $\Sigma_{\mathrm{ft}}$ has measure zero on $\pi^{-1}(\underline{C}_{\mathrm{in}}\cup C_{\mathrm{out}})$. In particular, for almost every $(x,p)\in\Sigma$, the future of $(x,p)$ along $\gamma_{x,p}$ either crosses $\mathcal{H}^+$ or is unbounded. 
\end{proposition}

\begin{proof}
Let $\Sigma_{\mathrm{pt}}$ be the set of points $(x,p)$ in $\Sigma$ such that $\gamma_{x,p}$ is past-trapped. Note that the union $\Sigma_{\mathrm{pt}}\cup \Sigma_{\mathrm{ft}}$ is given by $$\Sigma_{\mathrm{pt}}\cup \Sigma_{\mathrm{ft}}=\pi^{-1}(\underline{C}_{\mathrm{in}}\cup C_{\mathrm{out}})\cap \Big\{(x,p)\in\P: \ell(x,p)\in[2\sqrt{3}M,\infty), \quad E(x,p)=E_-(\ell)\Big\},$$ where $E_-(\ell)$ defines the energy level of the orbits in the sphere $\mathcal{S}^-(\ell)$. Finally, we observe that $\Sigma_{\mathrm{pt}}\cup \Sigma_{\mathrm{ft}}$ has codimension one, so it has measure zero. In particular, the set $\Sigma_{\mathrm{pt}}$ has also measure zero.
\end{proof}

In a similar fashion, one can easily show that for almost every $(x,p)\in\Sigma_1$ (or in $\Sigma_0$), the future of $(x,p)$ along $\gamma_{x,p}$ either crosses $\mathcal{H}^+$ or escapes to infinity.

\subsection{Dispersive properties in the near-horizon and the far-away regions}

In this subsection, we study dispersive properties of the geodesic flow in the near-horizon and the far-away regions. These properties are obtained by studying the geodesic equations for the null coordinates. By the mass-shell relation, the null momentum coordinates satisfy  
\begin{equation}\label{relation_mass-shell_null}
\dfrac{4r^2\Omega^2}{\ell^2+r^2}p^up^v=1.
\end{equation}
Motivated by this identity, we will show suitable expansion and contraction properties of the geodesic flow towards the submanifolds of the mass-shell $\{ \Omega^2p^u=0\}$ and $\{p^v=0\}$. We note that the submanifolds $\{ \Omega^2p^u=0\}$ and $\{p^v=0\}$ intersect at the tip of the light cones $\mathcal{P}_x$ for every point in spacetime.

\subsubsection{In the far-away region}

Let us consider the normalised momentum coordinates $\frac{r^2p^u}{\ell^2+r^2}$ and $\Omega^2p^v$. The product of these normalised coordinates is constant by \eqref{relation_mass-shell_null}. In the following, it will be convenient to parametrise timelike geodesics by the retarded time coordinate $u$. We define the derivative along the geodesic flow with respect to $u$ as $\frac{\d}{\d u}:=\frac{1}{p^u}\frac{\d}{\d s}$.

\begin{proposition}\label{prop_rp_expcontrac}
Along any timelike geodesic on the exterior of Schwarzschild with angular momentum $\ell$, we have
\begin{align}
\dfrac{\d}{\d u}\Big(\dfrac{r^2p^u}{\ell^2+r^2}\Big)&=-\frac{2M}{(r^2+\ell^2)r^2}\Big(r^2-\frac{\ell^2}{M}r+3\ell^2\Big)\Big(\dfrac{r^2p^u}{\ell^2+r^2}\Big)\label{geo_eqn_normalis_farpu},\\
\dfrac{\d}{\d u}(\Omega^2p^v)&=\frac{2M}{(r^2+\ell^2)r^2}\Big(r^2-\frac{\ell^2}{M}r+3\ell^2\Big)(\Omega^2p^v)\label{geo_eqn_normalis_farpv}.
\end{align}
\end{proposition}

By integrating the geodesic equations \eqref{geo_eqn_normalis_farpu}--\eqref{geo_eqn_normalis_farpv} for the normalised momentum coordinates, we obtain the following constants of motion along the geodesic flow:
\begin{align*} 
\dfrac{r^2p^u}{\ell^2+r^2}(0)&=\dfrac{r^2p^u}{\ell^2+r^2}(s)\exp\Big(\int_{u(0)}^{u(s)}\frac{2M}{(r^2+\ell^2)r^2}\Big(r^2-\frac{\ell^2}{M}r+3\ell^2\Big)\d u'\Big),\\ \Omega^2p^v(0)&=\Omega^2p^v(s)\exp\Big(-\int_{u(0)}^{u(s)}\frac{2M}{(r^2+\ell^2)r^2}\Big(r^2-\frac{\ell^2}{M}r+3\ell^2\Big)\d u'\Big). 
\end{align*}
We remark that the term $\frac{2M}{(r^2+\ell^2)r^2}(r^2-\frac{\ell^2}{M}r+3\ell^2)$ is positive when $r$ is sufficiently large. As a result, we obtain expansion and contraction for the geodesic flow towards the submanifolds of the mass-shell $\{ \Omega^2p^u=0\}$ and $\{p^v=0\}$, respectively.

\begin{remark}
The normalised momentum coordinate $\frac{r^2p^u}{\ell^2+r^2}$ satisfies the relation $\frac{r^2p^u}{\ell^2+r^2}=g(-\frac{r^2}{2\Omega^2(\ell^2+r^2)}\partial_v,p).$ The vector field $-\frac{r^2}{2\Omega^2(\ell^2+r^2)}\partial_v$ is a modification of the vector field $r^2\partial_v$ that is used in the well-known \emph{$r^p$-weighted energy method} introduced by Dafermos--Rodnianski for the study of wave equations on black hole spacetimes \cite{DR10}. On the other hand, the normalised coordinate $\Omega^2p^v$ satisfies the relation $\Omega^2p^v=g(-\frac{1}{2}\partial_u,p)$.
\end{remark}

\subsubsection{In the near-horizon region}

Let us consider the normalised momentum coordinates $\Omega^2p^u$ and $\frac{r^2p^v}{\ell^2+r^2}$. The product of these normalised coordinates is constant by \eqref{relation_mass-shell_null}. In the following, it will be convenient to parametrise timelike geodesics by the advanced time coordinate $v$. We define the derivative along the geodesic flow with respect to $v$ as $\frac{\d}{\d v}:=\frac{1}{p^v}\frac{\d}{\d s}$.

\begin{proposition}\label{lemma_redshift}
Along any timelike geodesic on the exterior of Schwarzschild with angular momentum $\ell$, we have
\begin{align}
\dfrac{\d}{\d v}(\Omega^2p^u)&=\frac{2M}{(r^2+\ell^2)r^2}\Big(r^2-\frac{\ell^2}{M}r+3\ell^2\Big)(\Omega^2p^u),\label{geo_eqn_normalis_horpu}\\
\dfrac{\d}{\d v}\Big(\dfrac{r^2p^v}{\ell^2+r^2}\Big)&=-\frac{2M}{(r^2+\ell^2)r^2}\Big(r^2-\frac{\ell^2}{M}r+3\ell^2\Big)\Big(\dfrac{r^2p^v}{\ell^2+r^2}\Big).\label{geo_eqn_normalis_horpv}
\end{align}
\end{proposition}

By integrating the geodesic equations \eqref{geo_eqn_normalis_horpu}--\eqref{geo_eqn_normalis_horpv} for the normalised momentum coordinates, we obtain the following constants of motion along the geodesic flow:
\begin{align*} 
\dfrac{r^2p^v}{\ell^2+r^2}(0)&=\dfrac{r^2p^v}{\ell^2+r^2}(s)\exp\Big(\int_{v(0)}^{v(s)}\frac{2M}{(r^2+\ell^2)r^2}\Big(r^2-\frac{\ell^2}{M}r+3\ell^2\Big)\d   v'\Big),\\ \Omega^2p^u(0)&=\Omega^2p^u(s)\exp\Big(-\int_{v(0)}^{v(s)}\frac{2M}{(r^2+\ell^2)r^2}\Big(r^2-\frac{\ell^2}{M}r+3\ell^2\Big)\d v'\Big). 
\end{align*}
We remark that the term $\frac{2M}{(r^2+\ell^2)r^2}(r^2-\frac{\ell^2}{M}r+3\ell^2)|_{r=2M}$ is equal to $\frac{1}{2M}$ for every $\ell\geq 0$. The value $\frac{1}{2M}$ corresponds to the surface gravity of $\mathcal{H}^+$. More generally, the term $\frac{2M}{(r^2+\ell^2)r^2}(r^2-\frac{\ell^2}{M}r+3\ell^2)$ is positive when $r\sim 2M$. As a result, we obtain expansion and contraction for the geodesic flow towards the submanifolds of the mass-shell $\{ p^v=0\}$ and $\{\Omega^2p^u=0\}$, respectively.

\begin{remark}
We note that the normalised momentum coordinate $\frac{r^2p^v}{\ell^2+r^2}$ satisfies the relation $\frac{r^2p^v}{\ell^2+r^2}=g(-\frac{r^2}{2\Omega^2(\ell^2+r^2)}\partial_u,p).$ The vector field $-\frac{r^2}{2\Omega^2(\ell^2+r^2)}\partial_u$ is a modification of the vector field $\Omega^{-2}\partial_v$ that is used to exploit the well-known \emph{red-shift effect} in the study of wave equations on black hole spacetimes. See the lecture notes \cite[Chapter 3]{DR13} for more information. On the other hand, the normalised coordinate $\Omega^2p^u$ satisfies the relation $\Omega^2p^u=g(-\frac{1}{2}\partial_v,p)$. 
\end{remark}

\subsection{Estimates for the momentum coordinates along timelike geodesics}

In this subsection, we prove some elementary a priori estimates for the momentum coordinates along geodesics in the exterior of Schwarzschild. We first address the case of geodesics that are not contained in the near-horizon region. 

In the rest of the paper, the following terminology will be useful. We say that a geodesic $\gamma$ is \emph{outgoing} or \emph{ingoing at a point} $x\in\mathcal{E}$, if $p^r>0$ or $p^r<0$, respectively. We also say that a geodesic $\gamma$ is \emph{outgoing} or \emph{ingoing in a region} $\mathcal{R}$, if $\gamma$ is outgoing or ingoing for every $x\in \mathcal{R}$, respectively. We sometimes say that a geodesic is simply outgoing or ingoing whenever the meaning is clear from the context.

\begin{lemma}\label{uniform_bound_pu_pv_bounded_region}
Let $r_0>2M$, $L_1>0$, and $E_1>0$. For every geodesic $\gamma$ in $\{r>r_0\}$ with angular momentum $\ell\leq L_1$ and particle energy $E\leq E_1$, the momentum coordinates along $\gamma$ satisfy $$|p^u(s)|\lesssim 1,\qquad |p^v(s)|\lesssim 1,\qquad |p^\theta(s)|\lesssim \frac{1}{r^2(s)},\qquad |\sin\theta p^\phi(s)|\lesssim \frac{1}{r^2(s)}.$$ 
\end{lemma}

\begin{proof}
By the definition of $E$, the null momentum coordinates are bounded by $$|p^u(s)+p^v(s)|=\dfrac{E}{\Omega^2(r(s))}\leq \dfrac{E_1}{\Omega^2(r_0)}.$$ We have used here the conservation of $E$ along the geodesic flow. On the other hand, the spherical momentum coordinates satisfy $$\sqrt{(p^{\theta})^2+\sin^2\theta(p^{\phi})^2}=\frac{\ell}{r^2(s)}\leq \frac{L_1}{r^2(s)}.$$ We have used here the conservation of $\ell$ along the geodesic flow.
\end{proof}

We obtain next explicit formulae for the null momentum coordinates in terms of the radial coordinate, and the conserved quantities of the corresponding timelike geodesic. This formulae follow from the mass-shell relation.

\begin{lemma}\label{lemma_normalised_angular_information}
Let $\gamma$ be a geodesic with angular momentum $\ell$ and particle energy $E$. Then, the null momentum coordinates $p^u$, $p^v$ along $\gamma$ are equal to
\begin{equation}\label{identity_values_null_momentum_coordinates}
   \dfrac{E}{2\Omega^2}\Big(1\pm\dfrac{1}{E}\sqrt{\dfrac{(E^2-1)r^3+2Mr^2-\ell^2r+2M\ell^2}{r^3}}\Big),
\end{equation}
where the sign is chosen depending on whether the geodesic is ingoing or outgoing at $\gamma(s)$.
\end{lemma}

\begin{proof}
By the mass-shell relation and the definition of $E$, we have $$\dfrac{E}{\Omega^2}=\dfrac{1}{4\Omega^2p^v}\Big(1+\dfrac{\ell^2}{r^2}\Big)+p^v.$$ We rewrite this identity as the quadratic equation $$(p^v)^2-\dfrac{E}{\Omega^2}p^v+\dfrac{1}{4\Omega^2}\Big(1+\dfrac{\ell^2}{r^2}\Big)=0.$$ We note that the same quadratic equation is satisfied by $p^u$. Thus, the null momentum coordinates are equal to $$\dfrac{E}{2\Omega^2}\Big(1\pm \sqrt{1-\frac{V_{\ell}(r)}{E^2}}\Big),$$ where the sign is chosen depending on whether the geodesic is ingoing or outgoing at $x\in \mathcal{E}$.
\end{proof}

\subsubsection{The momentum coordinates in the far-away region}

We will now perform estimates for the momentum coordinates along outgoing geodesics in the far-away region of spacetime. 

\begin{proposition}\label{proposition_asymptotic_tangent_at_infinity}
Let $E_1> E_0 > 1$ and $L_1>0$. There exists $R>2M$ such that every geodesic $\gamma$ in $\{r>R\}$ with angular momentum $\ell\leq L_1$ and particle energy $E\in [1,E_1]$ is either outgoing or ingoing. Moreover, for every outgoing geodesic $\gamma$ in $\{r>R\}$ with angular momentum $\ell\leq L_1$ and particle energy $E\in [1,E_1]$, the null momentum coordinates along $\gamma$ satisfy
\begin{equation}\label{estimate_momentum_coordinates_far_away_region}
    \Big|p^u(s)-\dfrac{E-\sqrt{E^2-1}}{2}\Big|\lesssim\dfrac{1}{r^{\frac{1}{2}}(s)}, \qquad \Big|p^v(s)-\dfrac{E+\sqrt{E^2-1}}{2}\Big|\lesssim\dfrac{1}{r^{\frac{1}{2}}(s)},
\end{equation}
Finally, for every outgoing geodesic $\gamma$ in $\{r>R\}$ with angular momentum $\ell\leq L_1$ and particle energy $E\in [E_0,E_1]$, the null momentum coordinates along $\gamma$ satisfy $$\Big|p^u(s)-\dfrac{E-\sqrt{E^2-1}}{2}\Big|\lesssim\dfrac{1}{r(s)}, \qquad \Big|p^v(s)-\dfrac{E+\sqrt{E^2-1}}{2}\Big|\lesssim\dfrac{1}{r(s)}.$$ 
\end{proposition}

\begin{proof}
By the mass-shell relation, the radial momentum coordinate $p^r$ along a geodesic $\gamma$ is bounded below by
\begin{equation}\label{estimate_radial_momentum_particle_enegy_larger_one}
(p^r)^2=E^2-V_{\ell}(r)=\dfrac{r^3(E^2-1)+2Mr^2-\ell^2r+2M\ell^2}{r^3}\geq \dfrac{2Mr^2-\ell^2r+2M\ell^2}{r^3}\gtrsim \dfrac{1}{r},
\end{equation}
by choosing $R>2M$ sufficiently large. We are using here that $\gamma$ is contained in $\{r\geq R\}$. In particular, every geodesic $\gamma$ in $\{r\geq R\}$ is either ingoing or outgoing. 

We now restrict our attention to outgoing geodesics. In this case, the null momentum coordinates of a geodesic $\gamma$ can be written as $$p^u=\dfrac{E}{2\Omega^2}\Big(1- \sqrt{1-\frac{V_{\ell}(r)}{E^2}}\Big),\qquad \quad  p^v=\dfrac{E}{2\Omega^2}\Big(1+ \sqrt{1-\frac{V_{\ell}(r)}{E^2}}\Big),$$ by Lemma \ref{lemma_normalised_angular_information}. We write the difference between the momentum coordinate $p^v$ and its limit value at infinity as
\begin{equation}\label{difference_momentum_coordinate_bounds_faraway_inside_lemma}
    p^v(s)-\dfrac{E+\sqrt{E^2-1}}{2}=\dfrac{EM}{\Omega^2r}+\dfrac{1}{2}\Big(\sqrt{\dfrac{E^2}{\Omega^4}-\dfrac{V_{\ell}(r)}{\Omega^4}}-\sqrt{E^2-1}\Big).
\end{equation}
We first address the case when $E\geq 1$. We bound the second term in \eqref{difference_momentum_coordinate_bounds_faraway_inside_lemma} by
\begin{align*}
   \Big|\sqrt{\dfrac{E^2}{\Omega^4}-\dfrac{V_{\ell}(r)}{\Omega^4}}-\sqrt{E^2-1}\Big|&\leq \Big|\dfrac{E^2}{\Omega^4}(1-\Omega^4)+\dfrac{1}{\Omega^4}(\Omega^4-V_{\ell}(r))\Big|^{\frac{1}{2}}\\ 
   &=\dfrac{1}{r^{\frac{1}{2}}\Omega^2}\Big|\dfrac{2M\ell^2}{r^2}-\dfrac{\ell^2+4M^2E^2}{r}+2M(2E^2-1)\Big|^{\frac{1}{2}}\lesssim \dfrac{1}{r^{\frac{1}{2}}},
\end{align*}
where we used the $\frac{1}{2}$-Hölder continuity of the square root in the first estimate. The estimate for $p^v$ then follows, since the first term in \eqref{difference_momentum_coordinate_bounds_faraway_inside_lemma} decays faster. A similar argument proves the corresponding estimate for $p^u$. 

On the other hand, if $E\geq E_0>1$, then the second term in \eqref{difference_momentum_coordinate_bounds_faraway_inside_lemma} can be bounded by
\begin{align*}
    \Big|\sqrt{\dfrac{E^2}{\Omega^4}-\dfrac{V_{\ell}(r)}{\Omega^4}}-\sqrt{E^2-1}\Big|&\lesssim \Big|\dfrac{E^2}{\Omega^4}(1-\Omega^4)+\dfrac{1}{\Omega^4}(\Omega^4-V_{\ell}(r))\Big|\\
    &= \dfrac{1}{r\Omega^4}\Big|\dfrac{2M\ell^2}{r^2}-\dfrac{\ell^2+4M^2E^2}{r}+2M(2E^2-1)\Big| \lesssim \dfrac{1}{r},
\end{align*}
where we used in the first estimate that the square root is Lipschitz strictly away of the origin. The estimate for $p^v$ then follows, since the first term in \eqref{difference_momentum_coordinate_bounds_faraway_inside_lemma} decays with the same rate. A similar argument proves the corresponding estimate for $p^u$.
\end{proof}

\begin{remark}
The decay rate $r^{-\frac{1}{2}}$ in the bounds \eqref{estimate_momentum_coordinates_far_away_region} for the momentum coordinates is optimal. For geodesics with $E=1$, we have $$p^v-\dfrac{1}{2}=\dfrac{M}{\Omega^2r}+\dfrac{1}{\Omega^2}\sqrt{1-V_{\ell}(r)}=\dfrac{M}{\Omega^2r}+\dfrac{\sqrt{2M}}{\Omega^2r^{\frac{3}{2}}}\Big(r^2-\dfrac{\ell^2}{2M}r+\ell^2\Big)^{\frac{1}{2}}.$$ A similar identity holds for the difference $p^u\,-\,\frac{1}{2}$ when $E=1$. This energy value (equal to one) corresponds to the rest mass of the particles in the system. 
\end{remark}

\subsubsection{The momentum coordinates in the near-horizon region}

We estimate the null momentum coordinates along ingoing geodesics in a neighbourhood of $\mathcal{H}^+$. In the following, we will estimate $Dp^u$ instead of $p^u$, because $\Omega^{-2}\partial_u$ extends regularly to a non-vanishing vector field on $\mathcal{H}^+$.

\begin{lemma}\label{lemma_asymptotic_tangent_at_horizon}
Let $E_0>0$ and $L>0$. There exists $r_0>2M$ such that every geodesic $\gamma$ in $\{r<r_0\}$ with angular momentum $\ell\leq L$ and particle energy $E\in [E_0,\infty)$, is either outgoing or ingoing. Moreover, for every ingoing geodesic $\gamma$ in $\{r<r_0\}$ with angular momentum $\ell\leq L$ and particle energy $E\in [E_0,\infty)$, the null momentum coordinates along $\gamma$ satisfy $$|E-\Omega^2p^u(s)|\lesssim \Omega(s), \qquad \Big|p^v(s)-\dfrac{1}{4E}\Big(\dfrac{\ell^2}{4M^2}+1\Big)\Big|\lesssim  \Omega(s).$$
\end{lemma}

\begin{proof}
By the mass-shell relation, the radial momentum coordinate $p^r$ along a geodesic $\gamma$ satisfies $$(p^r)^2=E^2-V_{\ell}(r)\geq E_0^2-V_{\ell}(r)>\frac{1}{2}E_0^2,$$ by choosing $r_0-2M$ sufficiently small. Thus, every such geodesic $\gamma$ in $\{r<r_0\}$ is either outgoing or ingoing.

We now restrict our attention to incoming geodesics. In this case, we note that $p^r(s)=\Omega^2p^v(s)-\Omega^2p^u(s)<0$. Thus, we can bound the term $\Omega^2p^v$ by 
\begin{equation}\label{estimate_good_dpv}
\Omega^2p^v\leq(\Omega^2p^u)^{\frac{1}{2}}(\Omega^2p^v)^{\frac{1}{2}}= \dfrac{1}{2}\Omega\Big(\dfrac{\ell^2}{r^2}+1\Big)^{\frac{1}{2}}\lesssim \Omega
\end{equation}
where we have used the mass-shell relation. As a result, we obtain the desired estimate for $\Omega^2p^u$ by 
\begin{equation}\label{estimate_used_mom_nearh}
|E-\Omega^2p^u(s)|=\Omega^2p^v(s)\leq C \Omega(s),
\end{equation}
where $C>0$ is a uniform constant. In particular, we obtain a uniform lower bound for $\Omega^2p^u$, since $$\Omega^2p^u\geq E-C\Omega(r)\geq E_0-C\Omega(r)\geq \frac{1}{2}E_0,$$ by choosing $r_0-2M$ sufficiently small. Finally, we obtain
\begin{align*}
    \Big|p^v-\dfrac{1}{4E}\Big(\dfrac{\ell^2}{4M^2}+1\Big)\Big|&=\Big|\dfrac{\ell^2}{4}\Big(\dfrac{1}{\Omega^2p^ur^2}-\dfrac{1}{4M^2E}\Big)+\dfrac{1}{4}\Big(\dfrac{1}{\Omega^2p^u}-\dfrac{1}{E}\Big)\Big|\\
    &=\Big|\dfrac{\ell^2}{4\Omega^2p^u}\Big(\dfrac{1}{r^2}-\dfrac{1}{4M^2}\Big)+\dfrac{\ell^2}{16M^2}\Big(\dfrac{1}{\Omega^2p^u}-\dfrac{1}{E}\Big)+\dfrac{E-\Omega^2p^u}{4E\Omega^2p^u}\Big|\\
    &=\Big|\dfrac{\ell^2(2M-r)(2M+r)}{16M^2\Omega^2r^2p^u}+\dfrac{\ell^2(E-\Omega^2p^u)}{16M^2E\Omega^2p^u}+\dfrac{E-\Omega^2p^u}{4E\Omega^2p^u}\Big|\lesssim \Omega,
\end{align*}
where we used in the last inequality the estimate \eqref{estimate_used_mom_nearh}, and the lower bounds for $\Omega^2p^u$ and $E$.
\end{proof}

\section{The trapped set and the stable manifolds}\label{section_struct_trapp}

In this section, we study the geodesic flow in a neighbourhood of the trapped set $\Gamma$ in the dispersive region $\D$ of phase space. In this analysis, an important role is played by the stable manifolds associated to the spheres of trapped orbits $\mathcal{S}^{-}(\ell)$ for $\ell\geq 2\sqrt{3}M$. We will show suitable contraction and expansion properties for the geodesic flow, by using suitable defining functions of the stable manifolds. We will also estimate the corresponding rates of contraction and expansion. We finish this section by studying similar properties for the parabolic trapping at infinity. 

\subsection{Unstable trapping}

Let $\ell>2\sqrt{3}M$. By Proposition \ref{propspherestrap}, $\mathcal{S}^{-}(\ell)$ is the unique sphere containing geodesics with angular momentum $\ell$. One can easily show the normal hyperbolicity of the sphere of trapped orbits $\mathcal{S}^{-}(\ell)$ by using variants of the arguments in \cite{WZ11}, because of the hyperbolicity of the fixed point $(r_-(\ell),0)$ of the radial flow.\footnote{Specifically, the sphere $\mathcal{S}^{-}(\ell)$ of trapped orbits is \emph{eventually absolutely $r$-normally hyperbolic} for every $r$ according to \cite[Chapter 1, Definition 4]{HPS77}.} The normal hyperbolicity of the sphere $\mathcal{S}^{-}(\ell)$ implies the existence of suitable stable and unstable manifolds $W^{\pm}(\ell)$ in phase space $\mathcal{P}$. This property follows by the powerful stable manifold theorem for normally hyperbolic sets \cite{HPS77}. We remark that the orbits in the stable and the unstable manifolds, are future-trapped and past-trapped, respectively. 

For the geodesic flow in Schwarzschild, we can write the submanifolds $W^{\pm}(\ell)$ explicitly, since they correspond to the energy level $\{E=E_-(\ell)\}$ of the radial flow. For this purpose, we parametrise the radial momentum coordinate $p^r$ of future-trapped and past-trapped orbits with $\ell>2\sqrt{3}M$ and $E=E_-(\ell)$, by 
\begin{equation}\label{iden_radial_mom_stable_mflds_large_ell}
p^{r,-}_{\ell}(r):=\sgn(r_-(\ell)-r)\sqrt{E^2_-(\ell)-V_{\ell}(r)} \qquad \text {and}\qquad p^{r,+}_{\ell}(r):=\sgn(r-r_-(\ell))\sqrt{E^2_-(\ell)-V_{\ell}(r)} ,
\end{equation}
respectively. We summarise this discussion with the following propositions.

\begin{proposition}\label{prop_stablemfldunsttrapp}
Let $\ell>4M$. The stable manifolds of the sphere of trapped orbits $\mathcal{S}^{-}(\ell)$ are analytic codimension two submanifolds of $\mathcal{P}$, given by
\begin{align*}
    W^+(\ell)&=\Big\{(x,p)\in\P: \ell(x,p)=\ell,\quad p^r=p^{r,+}_{\ell} \Big\},\\  
    W^-(\ell)&=\Big\{ (x,p)\in\P: \ell(x,p)=\ell,\quad p^r=p^{r,-}_{\ell} \Big\}.
\end{align*}
In particular, the intersection $W^+(\ell)\,\cap \, W^-(\ell)$ is equal to the sphere $\mathcal{S}^{-}(\ell)$.	
\end{proposition}

We note that the stable manifolds $W^{\pm}(\ell)$ are contained in the unbounded region $\{r>2M\}$. By the mass-shell relation, we have the constraint that $E^2_-(\ell)-V_{\ell}(r)\geq 0$ which is satisfied for any $r>2M$. 

\begin{proposition}\label{prop_stablemfldunsttrapp_medium_ell}
Let $\ell\in (2\sqrt{3}M, 4M)$. The stable manifolds of the sphere of trapped orbits $\mathcal{S}^{-}(\ell)$ are analytic codimension two submanifolds of $\mathcal{P}$, given by
\begin{align*}
    W^+(\ell)&=\Big\{(x,p)\in\P: \ell(x,p)=\ell,\quad r\leq r_-(\ell),\quad p^r=p^{r,+}_{\ell} \Big\}\\
    &\qquad \cup \Big\{(x,p)\in\P: \ell(x,p)=\ell,\quad r\geq r_-(\ell),\quad p^r=p^{r,\pm}_{\ell}  \Big\},\\
    W^-(\ell)&=\Big\{(x,p)\in\P: \ell(x,p)=\ell,\quad r\leq r_-(\ell),\quad p^r=p^{r,-}_{\ell} \Big\}\\
    &\qquad \cup \Big\{(x,p)\in\P: \ell(x,p)=\ell,\quad r\geq r_-(\ell),\quad p^r=p^{r,\pm}_{\ell}  \Big\}.
\end{align*}
In particular, the intersection $W^+(\ell)\,\cap \, W^-(\ell)$ is equal to the set $$W_{\mathrm{hom}}(\ell):=\Big\{(x,p)\in\P: \ell(x,p)=\ell,\quad r\geq r_-(\ell),\quad p^r=p^{r,\pm}_{\ell}  \Big\}.$$	
\end{proposition}

If $\ell\in (2\sqrt{3}M,4M)$, the stable manifolds $W^{\pm}(\ell)$ are only contained in the domain $\{2M<r\leq a(\ell)\}$, where $a(\ell)$ is the unique root of $E^2_-(\ell)-V_{\ell}(r)$ on $(r_-(\ell),\infty)$. Recall that by the mass-shell relation, we have the constraint that $E^2_-(\ell)-V_{\ell}(r)\geq 0$.

\begin{remark}
\begin{enumerate}[label = (\alph*)]
\item For $\ell\in (2\sqrt{3}M, 4M)$, the orbits in the stable manifolds with $r> r_-(\ell)$, are of \emph{homoclinic type}. In other words, these orbits satisfy $\lim_{s\to \pm \infty}(r(s),p^r(s))=(r_-(\ell),0)$. For this reason, the intersection of the stable manifolds is given by $W_{\mathrm{hom}}(\ell)$. We note that $W_{\mathrm{hom}}(\ell)$ is contained in the trapped set $\Gamma$.
\item We clarify here an abuse of terminology used in Proposition \ref{prop_stablemfldunsttrapp_medium_ell}. For $\ell\in (2\sqrt{3}M, 4M)$, the stable manifolds $W^{\pm}(\ell)$ are \underline{not submanifolds} of $\P$ due to global considerations. The stable manifolds $W^{\pm}(\ell)$ self-intersect at the sphere $\mathcal{S}^{-}(\ell)$. This property holds because the set $W_{\mathrm{hom}}(\ell)$, where we find homoclinic orbits with angular momentum $\ell$, is contained in the stable and the unstable manifolds. Despite this global misbehaviour, the sets $W^{\pm}(\ell)$ are indeed codimension two submanifolds of $\mathcal{P}$ in a \emph{neighbourhood of the sphere} $\mathcal{S}^{-}(\ell)$. 
\end{enumerate}
\end{remark}

\subsubsection{Lyapunov exponents of the fixed points $(r_-(\ell),0)$}

Let $\ell>2\sqrt{3}M$. Let us study the infinitesimal rates of contraction and expansion of the radial flow on the spheres of trapped orbits $\mathcal{S}^-(\ell)$.

It will be convenient to parametrise timelike geodesics by the time coordinate $t$. We define the derivative along the geodesic flow with respect to $t$ as $\frac{\d}{\d t}:=\frac{1}{p^t}\frac{\d}{\d s}$. In these terms, we write the radial geodesic equation as
\begin{align}
    \dfrac{\d^2r}{\d t^2}&=\dfrac{1}{(p^t)^3}\Big(\dfrac{\d p^r}{\d s}p^t-\dfrac{\d p^t}{\d s}p^r\Big)=-\dfrac{M\Omega^4}{E^2r^4}\Big(r^2-\dfrac{\ell^2}{M}r+3\ell^2\Big)+\dfrac{2M}{r^2\Omega^2}\Big(\dfrac{p^r}{p^t}\Big)^2\nonumber\\
    &=-\dfrac{M}{r^4V_{\ell}}\Big(\Omega^4-\Big(\dfrac{p^r}{p^t}\Big)^2\Big)\Big(r^2-\dfrac{\ell^2}{M}r+3\ell^2\Big)+\dfrac{2M}{r^2\Omega^2}\Big(\dfrac{p^r}{p^t}\Big)^2\nonumber\\
    &=-\dfrac{M\Omega^4}{r^4V_{\ell}}\Big(r^2-\dfrac{\ell^2}{M}r+3\ell^2\Big)+\dfrac{3M}{r^4V_{\ell}}\Big(r^2-\dfrac{\ell^2}{3M}r+\dfrac{5\ell^2}{3}\Big)\Big(\dfrac{p^r}{p^t}\Big)^2.\label{geoeqntimet}
\end{align}
The radial geodesic equation \eqref{geoeqntimet} defines the radial flow in the $(r,\frac{\d r}{\d t})$ coordinates by 
\begin{equation}\label{ode_autonomous_non-linear_massive_first}
 \dfrac{\d r}{\d t}=\dfrac{p^r}{p^t},\qquad\quad \dfrac{\d}{\d t}\Big(\dfrac{p^r}{p^t}\Big)=-\dfrac{M\Omega^4}{r^4V_{\ell}}\Big(r^2-\dfrac{\ell^2}{M}r+3\ell^2\Big)+\dfrac{3M}{r^4V_{\ell}}\Big(r^2-\dfrac{\ell^2}{3M}r+\dfrac{5\ell^2}{3}\Big)\Big(\dfrac{p^r}{p^t}\Big)^2.
\end{equation}

The sphere $\mathcal{S}^{-}(\ell)$ of trapped orbits corresponds to the fixed point $(r=r_-(\ell),\,p^r=0)$ for $\ell\geq 2\sqrt{3}M$. The \emph{linearisation of the radial flow \eqref{ode_autonomous_non-linear_massive_first} around the fixed point} $(r_-(\ell),0)$ is given by
\begin{equation}\label{ode_autonomous_linear_massive_fist}
        \dfrac{\d}{\d t}(r-r_-(\ell))=\dfrac{p^r}{p^t},\qquad \quad \dfrac{\d}{\d t}\Big(\dfrac{p^r}{p^t}\Big)=\Phi_{\ell}(r_-(\ell))(r-r_-(\ell)) ,
\end{equation} 
where the function $\Phi_{\ell}\colon [2M,\infty)\to\R$ is defined by
\begin{align*}
  \Phi_{\ell}(r) :=~&\dfrac{\d}{\d r}\Big(-\dfrac{M\Omega^4}{r^4V_{\ell}}\Big(r^2-\dfrac{\ell^2}{M}r+3\ell^2\Big)\Big).
\end{align*}
An elementary computation shows that 
\begin{equation}\label{dfn_lyap_exp}
\lambda(\ell):=\Phi_{\ell}(r_-(\ell))=\frac{M\Omega^2(r_-(\ell))}{r^2_-(\ell)(r_-^2(\ell)+\ell^2)}(r_+(\ell)-r_-(\ell)),
\end{equation}
by using that $r_{\pm}(\ell)$ satisfies $r^2_{\pm}-\frac{\ell^2}{M}r_{\pm}+3\ell^2=0$. As a result, we have $\Phi_{\ell}(r_-(\ell))>0$ when $\ell>2\sqrt{3}M$, and $\Phi_{2\sqrt{3}M}(6M)=0$ when $\ell=2\sqrt{3}M$. In particular, the fixed point $(r_-(\ell),0)$ of the radial flow is hyperbolic when $\ell>2\sqrt{3}M$, and the hyperbolicity \emph{degenerates} at $\ell=2\sqrt{3}M$.

\begin{definition}
Let $\lambda\colon (2\sqrt{3}M,\infty)\to [0,\infty)$ be defined by $$\lambda(\ell):=\Phi_{\ell}(r_-(\ell)).$$ The \emph{positive and negative Lyapunov exponents of the fixed point} $(r_-(\ell),0)$, are $\lambda^{\frac{1}{2}}(\ell)$ and $-\lambda^{\frac{1}{2}}(\ell)$, respectively. From now on, we refer to $\lambda^{\frac{1}{2}}(\ell)$ simply as the \emph{Lyapunov exponent of} $(r_-(\ell),0)$. 
\end{definition}

The Lyapunov exponent $\lambda^{\frac{1}{2}}(\ell)$ determines the infinitesimal rate of expansion and contraction of the radial flow on the unstable and stable manifolds, respectively. In other words, the Lyapunov exponent $\lambda^{\frac{1}{2}}(\ell)$ determines the rate of expansion and contraction of the differential of the radial flow on the stable and unstable subspaces (which are tangent to the stable and unstable manifolds), respectively. 

We will now show that the Lyapunov exponent $\lambda^{\frac{1}{2}}(\ell)$ is strictly increasing in $\ell$. We begin proving that $r_-(\ell)$ is a monotone function. 

\begin{lemma}\label{lemma_derivative_rlminusunstrtrap}
The derivative of the function $r_-\colon [2\sqrt{3}M,\infty)\to (3M,6M]$ is 
\begin{align*}
\frac{\d r_-}{\d \ell}=-\frac{2M}{\ell^3}\frac{1}{\sqrt{1-\frac{12M^2}{\ell^2}}}r^2_-(\ell).
\end{align*}
In particular, we have $r_-(\ell_2)\leq r_-(\ell_1)$ for all $l_2\geq l_1$.
\end{lemma}

\begin{proof}
The monotonicity of $r_-(\ell)$ follows directly by proving that $\frac{\d r_-}{\d \ell}\geq 0$. We compute the derivative
\begin{align*}
\frac{\d r_-}{\d \ell}&=\frac{\ell}{M}\Big(1-\sqrt{1-\frac{12M^2}{\ell^2}}\Big)-\frac{6M}{\ell}\frac{1}{\sqrt{1-\frac{12M^2}{\ell^2}}}\\
&=\frac{\ell}{M}\frac{1}{\sqrt{1-\frac{12M^2}{\ell^2}}}\Big(\sqrt{1-\frac{12M^2}{\ell^2}}-\Big(1-\frac{12M^2}{\ell^2}\Big)-\frac{6M}{\ell^2}\Big)\\
&=-\frac{\ell}{2M}\frac{1}{\sqrt{1-\frac{12M^2}{\ell^2}}}\Big(1-\sqrt{1-\frac{12M^2}{\ell^2}}\Big)^2=-\frac{2M}{\ell^3}\frac{1}{\sqrt{1-\frac{12M^2}{\ell^2}}}r^2_-(\ell),
\end{align*}
where we used the definition of $r_-(\ell)$ in the last line. As a result, we obtain the desired monotonicity property of $r_-(\ell)$.
\end{proof}
 
We prove next the desired monotonicity property for the Lyapunov exponent $\lambda^{\frac{1}{2}}(\ell)$.

\begin{proposition}\label{prop_monot_lyap_exp}
The derivative of the function $\lambda\colon (2\sqrt{3}M,\infty)\to [0,\infty)$ is 
\begin{align*}
\frac{\d\lambda}{\d \ell}=\frac{M^2}{\ell^3}\frac{1}{r_-(\ell)(r_-(\ell) -3M)} \frac{1}{\sqrt{1-\frac{12M^2}{\ell^2}}}\Big(1-\sqrt{1-\frac{12M^2}{\ell^2}}\Big)\Big(1+4\sqrt{1-\frac{12M^2}{\ell^2}}\Big).
\end{align*}
In particular, we have $\lambda(\ell_2)\geq \lambda(\ell_1)$ for all $l_2\geq l_1$.
\end{proposition}

\begin{proof}
The monotonicity of $\lambda(\ell)$ follows directly by proving that $\frac{\d\lambda}{\d \ell}\geq 0$. We first rewrite the function $\lambda(\ell)$ as $$\lambda(\ell):=\frac{M\Omega^2(r_-(\ell))}{r^2_-(\ell)(r_-^2(\ell)+\ell^2)}(r_+(\ell)-r_-(\ell))=\frac{M^2}{\ell^2}\frac{1}{r_-^2(\ell)-3Mr_-(\ell)}\sqrt{1-\frac{12M^2}{\ell^2}},$$ where we used the definition of $r_{\pm}(\ell)$. We compute the derivative 
\begin{align*}
\frac{\d\lambda}{\d \ell}&=-\frac{2M^2}{\ell^3}\frac{1}{r_-^2-3Mr_-}\sqrt{1-\frac{12M^2}{\ell^2}}-\frac{M^2}{\ell^2}\frac{2r_- -3M}{(r_-^2-3Mr_-)^2}\frac{\d r_-}{\d \ell}\sqrt{1-\frac{12M^2}{\ell^2}}\\
&\qquad +\frac{12M^4}{\ell^5}\frac{1}{r_-^2-3Mr_-}\frac{1}{\sqrt{1-\frac{12M^2}{\ell^2}}}\\
&=\frac{2M^2}{\ell^3}\frac{1}{r_-^2-3Mr_-}\frac{1}{\sqrt{1-\frac{12M^2}{\ell^2}}}\Big(\frac{18M^2}{\ell^2}-1\Big)+\frac{2M^3}{\ell^5}\frac{2r_- -3M}{(r_-^2-3Mr_-)^2}r_-^2(\ell)\\
&=\frac{2M^2}{\ell^3}\frac{1}{r_-(r_--3M)}\frac{1}{\sqrt{1-\frac{12M^2}{\ell^2}}}\Big(\frac{18M^2}{\ell^2}-1+\sqrt{1-\frac{12M^2}{\ell^2}}+\frac{Mr_-}{\ell^2}\sqrt{1-\frac{12M^2}{\ell^2}}\Big)\\
&=\frac{M^2}{\ell^3}\frac{1}{r_-(r_- -3M)} \frac{1}{\sqrt{1-\frac{12M^2}{\ell^2}}}\Big(1-\sqrt{1-\frac{12M^2}{\ell^2}}\Big)\Big(1+4\sqrt{1-\frac{12M^2}{\ell^2}}\Big), 
\end{align*}
where we used Lemma \ref{lemma_derivative_rlminusunstrtrap} in the second equality, and the definition of $r_-(\ell)$ in the last two lines. As a result, we obtain the desired monotonicity property of $\lambda(\ell).$
\end{proof}

\subsubsection{Expansion and contraction of the radial flow with $\ell\in(2\sqrt{3}M,4M)$}

Let us set the function $H_{\ell}\colon [0,\infty)\to\R$ given by 
\begin{equation*}
H_{\ell}(E):=\frac{2M}{1-E^2}-\frac{2M}{1-E^2_-(\ell)},
\end{equation*}
with $\ell\in(2\sqrt{3}M,4M)$. Note that $H_{\ell}(E)$ is conserved along the geodesic flow when the angular momentum is $\ell$. By Proposition \ref{prop_stablemfldunsttrapp}, the union $W^{+}(\ell)\, \cup \,W^{-}(\ell)$ of the stable manifolds associated to the sphere of trapped orbits $\mathcal{S}^-(\ell)$, is characterised as $\{H_{\ell}(E)=0\}$. 

It will be convenient to parametrise timelike geodesics by the time coordinate $\bar{t}=s(1-E^2)^{\frac{1}{2}}$ in the domain $\{E<1\}$. We define the derivative along the geodesic flow with respect to $\bar{t}$ as $\frac{\d}{\d \bar{t}}=(1-E^2)^{-\frac{1}{2}}\frac{\d}{\d s}$.

\begin{proposition}\label{prop_hamilt_exp_contr_main_hyp}
Let us consider the radial flow with $\ell\in (2\sqrt{3}M,4M)$. The quantity $H_{\ell}(E)$ satisfies
\begin{equation}\label{identity_impact_parameter_massive_fields_1}
H_{\ell}(E)=\frac{r^3}{r^2-\frac{\ell^2}{2M}r+\ell^2}\Big(\frac{(p^r)^2}{1-E^2}+\Big(1+\frac{a(\ell)}{r}\Big)\Big(1-\frac{r_-(\ell)}{r}\Big)^2\Big),
\end{equation}
where $$a(\ell)=\frac{2M^2r_-(\ell)^3}{\ell^2(4M-r_-(\ell))(r_-(\ell)-3M)}.$$ Set the functions $\varphi_{\pm}^{\ell} \colon \P\, \cap \, \{E<1\}\to\R$ given by $$\varphi_{\pm}^{\ell}(x,p):=\frac{r^{\frac{3}{2}}}{(r^2-\frac{\ell^2}{2M}r+\ell^2)^{\frac{1}{2}}}\Big(\frac{p^r}{(1-E^2)^{\frac{1}{2}}}\pm\Big(-\frac{a(\ell)}{r}-1\Big)^{\frac{1}{2}}\Big(1-\frac{r_-(\ell)}{r}\Big)\Big).$$ Then, the derivative of $\varphi_{\pm}^{\ell}$ along the geodesic flow is 
\begin{equation}\label{expcontrac_hyperbol}
\dfrac{\d\varphi_{\pm}^{\ell}}{\d \bar{t}}=\mp \frac{r^{\frac{1}{2}}(r-r_+(\ell))}{2(-a(\ell)-r)^{\frac{1}{2}}(r^2-\frac{\ell^2}{2M}r+\ell^2)}\varphi_{\pm}^{\ell}.
\end{equation}
\end{proposition}

\begin{proof}
Rearranging the mass-shell relation \eqref{identity_particle_energy_angular_momentum}, we have 
\begin{equation}\label{first_id_impact_paramet_general}
\frac{(p^r)^2}{E^2-1}-1=\frac{2M}{E^2-1}\cdot\frac{1}{r}-\frac{\ell^2}{E^2-1}\cdot\frac{r-2M}{r^3}=\frac{2M}{E^2-1}\cdot\frac{r^2-\frac{\ell^2}{2M}r+\ell^2}{r^3}.
\end{equation}
Evaluating this identity on $(r=r_-(\ell),\,p^r=0)$, we obtain 
\begin{equation}\label{first_id_impact_paramet_specific}
\frac{\ell^2}{E_-^2(\ell)-1}=-\frac{\ell^2}{2M}\cdot\frac{r_-^3}{r_-^2-\frac{\ell^2}{2M}r_-+\ell^2}=\frac{r_-^3}{4M-r_-},
\end{equation}
where we used $r_-^2-\frac{\ell^2}{M}r_-+3\ell^2=0$ in the last equality. Putting together the identities \eqref{first_id_impact_paramet_general} and \eqref{first_id_impact_paramet_specific}, we have 
\begin{align*}
\frac{2M}{E^2-1}-\frac{2M}{E_-^2(\ell)-1}&=\frac{r^3}{r^2-\frac{\ell^2}{2M}r+\ell^2}\cdot\frac{(p^r)^2}{E^2-1} -\frac{1}{r^2-\frac{\ell^2}{2M}r+\ell^2}\Big(r^3+\frac{2Mr_-^3}{\ell^2(4M-r_-)}\Big(r^2-\frac{\ell^2}{2M}r+\ell^2\Big)\Big).
\end{align*}

To obtain \eqref{identity_impact_parameter_massive_fields_1}, it is enough to show that $a(\ell)$ satisfies 
\begin{equation}\label{keyidentimpactparam}
r^3+\frac{2Mr_-^3}{\ell^2(4M-r_-)}r^2-\frac{r_-^3}{4M-r_-}r+\frac{2Mr_-^3}{4M-r_-}=(r+a(\ell))(r-r_-)^2.
\end{equation}
Using the equation $r_-^2-\frac{\ell^2}{M}r_-+3\ell^2=0$, we first write
\begin{align}
(r+a(\ell))(r-r_-)^2&=(r+a(\ell))\Big(r^2-2rr_-+\frac{\ell^2}{M}(r_--3M)\Big)\nonumber\\
&=r^3+(a(\ell)-2r_-)r^2+\Big(\frac{\ell^2}{M}(r_--3M)-2a(\ell) r_-\Big)r+\frac{\ell^2}{M}(r_--3M)a(\ell).\label{polyal}
\end{align}
The identity \eqref{keyidentimpactparam} follows if and only if the corresponding coefficients are equal, in other words, if 
\begin{equation}\label{identities_sat_al}
a(\ell)=2r_-+\frac{2Mr_-^3}{\ell^2(4M-r_-)},\qquad \quad a(\ell)=\frac{1}{2r_-}\Big(\frac{\ell^2}{M}(r_--3M)+\frac{r_-^3}{4M-r_-}\Big).
\end{equation} 
The identities \eqref{identities_sat_al} are proved by an elementary computation using $r_-^2-\frac{\ell^2}{M}r_-+3\ell^2=0$.

Before proceeding, we note that the functions $\varphi_{\pm}^{\ell}$ are well-defined. On the one hand, $r_-+a(\ell)>0$ when $\ell\in (2\sqrt{3}M,4M)$. Furthermore, we have $$r^2-\frac{\ell^2}{2M}r+\ell^2=(r-\frac{\ell^2}{4M})^2+\frac{\ell^2}{16M^2}(16M^2-\ell^2)\geq 0$$ when $\ell\in (2\sqrt{3}M,4M)$. Thus,  the functions $\varphi_{\pm}^{\ell}$ are well-defined on $\P \, \cap\, \{E<1\}$. 

We show next the relation \eqref{expcontrac_hyperbol}. We first compute the derivative 
\begin{align}
\frac{\d}{\d s}& \Big(\frac{r^{\frac{3}{2}}p^r}{(r^2-\frac{\ell^2}{2M}r+\ell^2)^{\frac{1}{2}}}\Big)=\frac{\d}{\d s} \Big(\frac{1}{(r^2-\frac{\ell^2}{2M}r+\ell^2)^{\frac{1}{2}}}\Big)r^{\frac{3}{2}}p^r+\frac{1}{(r^2-\frac{\ell^2}{2M}r+\ell^2)^{\frac{1}{2}}}\frac{\d}{\d s}(r^{\frac{3}{2}}p^r)\nonumber\\
&=\frac{(p^r)^2r^{\frac{3}{2}}(\frac{\ell^2}{4M}-r)}{(r^2-\frac{\ell^2}{2M}r+\ell^2)^{\frac{3}{2}}}+\frac{1}{(r^2-\frac{\ell^2}{2M}r+\ell^2)^{\frac{1}{2}}}\Big(\frac{3}{2}r^{\frac{1}{2}}(p^r)^2-\frac{M}{r^{\frac{5}{2}}}\Big(r^2-\frac{\ell^2}{M}r+3\ell^2\Big)\Big)\nonumber\\
&=\frac{(p^r)^2r^{\frac{1}{2}}}{2(r^2-\frac{\ell^2}{2M}r+\ell^2)^{\frac{3}{2}}}\Big(r^2-\frac{\ell^2}{M}r+3\ell^2\Big)-\frac{1}{(r^2-\frac{\ell^2}{2M}r+\ell^2)^{\frac{1}{2}}}\frac{M}{r^{\frac{5}{2}}}\Big(r^2-\frac{\ell^2}{M}r+3\ell^2\Big)\nonumber\\
&=\frac{(r^2-\frac{\ell^2}{M}r+3\ell^2)}{(r^2-\frac{\ell^2}{2M}r+\ell^2)^{\frac{3}{2}}}\frac{r^{\frac{1}{2}}}{2}(E^2-1),\label{identity_firsthalf_expcontract}
\end{align}
where we used the mass-shell relation in the last equality. Secondly, we compute the derivative 
\begin{align*}
\frac{\d}{\d s}& \Big(\frac{(-r-a(\ell))^{\frac{1}{2}}(r-r_-)}{(r^2-\frac{\ell^2}{2M}r+\ell^2)^{\frac{1}{2}}}\Big)=\frac{\d}{\d s} \Big(\frac{1}{(r^2-\frac{\ell^2}{2M}r+\ell^2)^{\frac{1}{2}}}\Big)(-r-a(\ell))^{\frac{1}{2}}(r-r_-)\\
&\qquad \qquad \qquad \qquad \qquad \qquad \qquad + \frac{1}{(r^2-\frac{\ell^2}{2M}r+\ell^2)^{\frac{1}{2}}}\frac{\d}{\d s}\Big((-r-a(\ell))^{\frac{1}{2}}(r-r_-)\Big)\\
&=-\frac{p^r (-a(\ell)-r)^{-\frac{1}{2}}}{2(r^2-\frac{\ell^2}{2M}r+\ell^2)^{\frac{3}{2}}}\Big(2\Big(\frac{\ell^2}{4M}-r\Big)(r+a(\ell))(r-r_-)+(3r+2a(\ell)-r_-)\Big(r^2-\frac{\ell^2}{2M}r+\ell^2\Big)\Big).
\end{align*}
An explicit computation using $r_-^2-\frac{\ell^2}{M}r_-+3\ell^2=0$ and the definition of $a(\ell)$, shows that $$2\Big(\frac{\ell^2}{4M}-r\Big)(r+a(\ell))(r-r_-)+(3r+2a(\ell)-r_-)\Big(r^2-\frac{\ell^2}{2M}r+\ell^2\Big)=r^3-r^2r_+.$$ Hence, we have
\begin{equation}\label{identity_secondhalf_expcontract}
\frac{\d}{\d s} \Big(\frac{(-a(\ell)-r)^{\frac{1}{2}}(r-r_-)}{(r^2-\frac{\ell^2}{2M}r+\ell^2)^{\frac{1}{2}}}\Big)=\frac{p^r r^2(r-r_+)}{2(-a(\ell)-r)^{\frac{1}{2}}(r^2-\frac{\ell^2}{2M}r+\ell^2)^{\frac{3}{2}}}.
\end{equation}
Putting together \eqref{identity_firsthalf_expcontract} and \eqref{identity_secondhalf_expcontract}, we obtain \eqref{expcontrac_hyperbol}.
\end{proof}

We observe that the term  in the RHS of \eqref{expcontrac_hyperbol} satisfies $$\frac{r^{\frac{1}{2}}(r-r_+(\ell))}{2(-a(\ell)-r)^{\frac{1}{2}}(r^2-\frac{\ell^2}{2M}r+\ell^2)}\Big|_{r=r_-(\ell)}<0.$$ This observation shows the local expansion and contraction properties of the geodesic flow near the unstable and stable manifolds, respectively. By the identity \eqref{expcontrac_hyperbol}, the Lyapunov exponent of $(r_-(\ell),0)$ can also be written as $$\lambda^{\frac{1}{2}}(\ell)=\frac{(1-E^2_-(\ell))^{\frac{1}{2}}}{2E_-(\ell)}\frac{(r_+(\ell)-r_-(\ell))(r_-(\ell)-2M)}{r^{\frac{1}{2}}_-(-a(\ell)-r_-(\ell))^{\frac{1}{2}}(r^2_-(\ell)-\frac{\ell^2}{2M}r_-(\ell)+\ell^2)}.$$ 

\begin{remark}\label{remark_stable_defin_homocl}
The zero sets of the functions $\varphi_{+}^{\ell}$ and $\varphi_{-}^{\ell}$ define \emph{only subsets} of the unstable and stable manifolds associated to the sphere $\mathcal{S}^{-}(\ell)$, respectively. In other words, for all $\ell \in (2\sqrt{3}M, 4M)$, we have 
\begin{equation}
W^{\pm}_{\mathrm{in/out}}(\ell):=\Big\{(x,p)\in \mathcal{P}: \ell(x,p)=\ell, \quad \varphi_{\mp}^{\ell}(x,p)=0\Big\}\subset W^{\pm}(\ell).
\end{equation}
We note that the geodesics in the subsets $W^{\pm}_{\mathrm{in/out}}(\ell)$ of the stable manifolds are always either ingoing or outgoing. In particular, the sets $W^{\pm}_{\mathrm{in/out}}(\ell)$ only contain the ingoing or the outgoing parts of the homoclinic orbits in the energy level $\{E=E_-(\ell)\}$.
\end{remark}

Next, we will show the exponential rate of contraction and expansion of the geodesic flow on the stable manifolds $W^{\pm}(\ell)$ in a neighbourhood of $(r_-(\ell),0)$. We show that the rate of contraction and expansion can be taken arbitrarily close to the Lyapunov exponent in a small neighbourhood of $(r_-(\ell),0)$.

\begin{proposition}\label{exp_concentr_stable_mflds_hyp}
Let $\ell\in (2\sqrt{3}M,4M)$ and $\delta\in (0,\lambda^{\frac{1}{2}}(\ell))$. There exists $\epsilon>0$ such that for every geodesic $\gamma_{x,p}\colon [0,a]\to \{|r-r_-(\ell)|\lesssim  \epsilon \}$ with $(x,p)\in W^-(\ell)$, we have
\begin{equation}
  |r(t(s))-r_-(\ell) |\lesssim \frac{1}{\exp((\lambda^{\frac{1}{2}}(\ell)-\delta)t(s))},\qquad |p^r(t(s))|\leq \frac{1}{\exp((\lambda^{\frac{1}{2}}(\ell)-\delta)t(s))}
\end{equation}
for all $s\in [0,a]$. Moreover, for every geodesic $\gamma_{x,p} \colon [-a,0]\to \{  |r-r_-(\ell)|\leq  \epsilon\}$ with $(x,p)\in W^+(\ell)$, we have
\begin{equation}
  |r(t(-s))-r_-(\ell)|\lesssim \frac{1}{\exp((\lambda^{\frac{1}{2}}(\ell)-\delta)|t(-s)|)},\qquad |p^r(t(-s))|\lesssim \frac{1}{\exp((\lambda^{\frac{1}{2}}(\ell)-\delta)|t(-s)|)},
\end{equation}
for all $s\in [-a,0]$. 
\end{proposition}

\begin{proof}
We suppose first that $(x,p)\in W^-(\ell)$. By the invariance of the stable manifold $W^-(\ell)$, we have that $(x(s),p(s))\in W^-(\ell)$ for all $s\geq 0$. Assuming that $\epsilon>0$ is sufficiently small (depending on $\ell$), we have that $\{|r-r_-(\ell)|\leq  \epsilon \}\cap W^-(\ell)$ is a subset of $\{\varphi_{+}^{\ell}=0\}$, so $$\frac{p^r}{(1-E^2_-(\ell))^{\frac{1}{2}}}=-\Big(-\frac{a(\ell)}{r}-1\Big)^{\frac{1}{2}}\Big(1-\frac{r_-(\ell)}{r}\Big).$$ As a result, the defining function $\varphi_{-}^{\ell}$ satisfies the relations
\begin{equation}\label{cancell_varphi_on_stablemfl}
\varphi_{-}^{\ell}(x,p)=2\frac{\omega_{\ell}(r) p^r}{(1-E^2_-(\ell))^{\frac{1}{2}}}=-2\omega_{\ell}(r)\Big(-\frac{a(\ell)}{r}-1\Big)^{\frac{1}{2}}\Big(1-\frac{r_-(\ell)}{r}\Big), 
\end{equation}
where $\omega_{\ell}(r):=r^{\frac{3}{2}}(r^2-\frac{\ell^2}{2M}r+\ell^2)^{-\frac{1}{2}}$ is a positive radial weight. 

Next, we integrate the derivative \eqref{expcontrac_hyperbol} along the geodesic flow by $$\varphi_{-}^{\ell}\Big(x(t(s)),p(t(s))\Big)=\varphi_{-}^{\ell}(x,p)\exp \Big(-\frac{(1-E^2_-(\ell))^{\frac{1}{2}}}{E_-(\ell)}\int_{0}^{t(s)} \frac{(r_+(\ell)-r)(r-2M)}{2r^{\frac{1}{2}}(-a(\ell)-r)^{\frac{1}{2}}(r^2-\frac{\ell^2}{2M}r+\ell^2)}\d t \Big),$$ where here $(x(0)=x,\, p(0)=p)$. In the following, we set $\epsilon>0$ small enough so that $$\Big|\lambda^{\frac{1}{2}}(\ell)-\frac{(1-E^2_-(\ell))^{\frac{1}{2}}}{E_-(\ell)}\frac{(r_+(\ell)-r)(r-2M)}{2r^{\frac{1}{2}}(-a(\ell)-r)^{\frac{1}{2}}(r^2-\frac{\ell^2}{2M}r+\ell^2)} \Big| \leq \delta,$$ for all $(x,p)\in \{|r-r_-(\ell)|\leq  \epsilon\}$. As a result, we obtain the upper bound 
\begin{align*}
\Big|\varphi_{-}^{\ell}\Big(x(t(s)),p(t(s))\Big)\Big|&\lesssim \exp \Big(-\frac{(1-E^2_-(\ell))^{\frac{1}{2}}}{E_-(\ell)}\int_{0}^{t(s)} \frac{(r_+(\ell)-r)(r-2M)}{2r^{\frac{1}{2}}(-a(\ell)-r)^{\frac{1}{2}}(r^2-\frac{\ell^2}{2M}r+\ell^2)}\d t \Big)\\
&\lesssim \exp \Big(-(\lambda^{\frac{1}{2}}(\ell)-\delta)t(s)\Big).
\end{align*}
Finally, we use the relations \eqref{cancell_varphi_on_stablemfl} to show that $$|r(t(s))-r_-(\ell)|+|p^r(t(s))|\lesssim \Big|\varphi_{-}^{\ell}\Big(x(t(s)),p(t(s))\Big)\Big|\lesssim \exp \Big(-(\lambda^{\frac{1}{2}}(\ell)-\delta)t(s)\Big),$$ where we have used suitable lower bounds for the factors in \eqref{cancell_varphi_on_stablemfl}. The case of a geodesic $\gamma_{x,p}$ with $(x,p)\in W^+(\ell)$ follows similarly.
\end{proof}

\subsubsection{Expansion and contraction of the radial flow with $\ell=4M$}

Let us consider the conserved quantity along the geodesic flow given by 
\begin{equation*}
E^2-1.
\end{equation*}
By Proposition \ref{prop_stablemfldunsttrapp}, the union $W^{+}(4M) \,\cup \,W^{-}(4M)$ of the stable manifolds associated to the sphere of trapped orbits $\mathcal{S}^-(4M)$, is characterised as $\{E^2-1=0\}$.

\begin{proposition}\label{prop_expcontrac_4M}
Let us consider the radial flow with $\ell=4M$. The quantity $E^2-1$ satisfies 
\begin{equation}\label{defin_funct_meetseparatr}
E^2-1=(p^r)^2-\dfrac{2M}{r}\Big(1-\frac{4M}{r}\Big)^2.
\end{equation} 
Set the functions $\psi_{\pm}^{4M}\colon \P\to\R$ given by $$\psi_{\pm}^{4M}(x,p):=p^r\pm\dfrac{\sqrt{2M}}{r^{\frac{1}{2}}}\Big(1-\frac{4M}{r}\Big).$$ Then, the derivative of $\psi_{\pm}^{4M}$ along the geodesic flow is 
\begin{equation}\label{expcontrmeetseparatr}
\dfrac{\d \psi_{\pm}^{4M}}{\d s}=\mp\frac{\sqrt{2M}}{2}\frac{1}{r^{\frac{3}{2}}} \Big(1-\frac{12M}{r}\Big) \psi_{\pm}^{4M}.
\end{equation}
\end{proposition}

\begin{proof}
The identity \eqref{defin_funct_meetseparatr} follows directly by using the mass-shell relation \eqref{identity_particle_energy_angular_momentum} with $\ell=4M$. We show next \eqref{expcontrmeetseparatr}. We first compute the derivative
\begin{equation}\label{eqnmeetsepat1}
\dfrac{\d }{\d s}\Big(\frac{r-4M}{r^{\frac{3}{2}}}\Big)=-\dfrac{p^r}{2r^{\frac{5}{2}}}(r-12M).
\end{equation}
Moreover, the radial geodesic equation is 
\begin{equation}\label{eqnmeetsepat2}
\dfrac{\d p^r}{\d s}=-\dfrac{M}{r^4}\Big(r^2-16Mr+48M^2\Big)=-\dfrac{M}{r^4}(r-4M)(r-12M),
\end{equation}
when $\ell=4M$. The identity \eqref{expcontrmeetseparatr} follows by putting together the identities \eqref{eqnmeetsepat1} and \eqref{eqnmeetsepat2}.
\end{proof}

We note that the term in the RHS of \eqref{expcontrmeetseparatr} satisfies $$\frac{\sqrt{2M}}{2}\frac{1}{r^{\frac{3}{2}}} \Big(1-\frac{12M}{r}\Big)\Big|_{r=4M}<0.$$ This observation shows the local expansion and contraction properties of the geodesic flow near the unstable and stable manifolds, respectively. By the identity \eqref{expcontrmeetseparatr}, the Lyapunov exponent of $(4M,0)$ can also be written as $$\lambda^{\frac{1}{2}}(4M)=\frac{1}{4\sqrt{2}M}=\frac{\sqrt{2M}}{2r^{\frac{3}{2}}_-(4M)} \Big(\frac{12M}{r_-(4M)}-1\Big) \Big(1-\frac{2M}{r_-(4M)}\Big).$$

\begin{remark}
The zero sets of the functions $\psi_{+}^{4M}$ and $\psi_{-}^{4M}$ define the unstable and stable manifolds associated to the sphere $\mathcal{S}^{-}(4M)$, respectively. In other words, the stable manifolds $W^{\pm}(4M)$ can be written as 
\begin{equation}\label{iden_charact_stable_4m}
W^{\pm}(4M)=\Big\{(x,p)\in \mathcal{P}: \ell(x,p)=4M,\quad\psi_{\mp}(4M)(x,p)=0\Big\}.
\end{equation}
Note the contrast between this property compared to the zero sets $\{\varphi_{\pm}^{\ell}=0\}$ in Remark \ref{remark_stable_defin_homocl}. 
\end{remark}

\begin{proposition}
Let $\delta\in (0,\frac{1}{8\sqrt{2}M})$. There exists $\epsilon>0$ such that for every geodesic $\gamma_{x,p}\colon [0,a]\to \{|r-4M|\leq  \epsilon \}$ with $(x,p)\in W^-(4M)$, we have
\begin{equation}
  |r(t(s))-4M|\lesssim \frac{1}{\exp((\frac{1}{8\sqrt{2}M}-\delta)t(s))},\qquad |p^r(t(s))|\lesssim \frac{1}{\exp((\frac{1}{8\sqrt{2}M}-\delta)t(s))},
\end{equation}
for all $s\in [0,a]$. Moreover, for every geodesic $\gamma_{x,p} \colon [-a,0]\to \{  |r-4M|\leq  \epsilon\}$ with $(x,p)\in W^+(4M)$, we have
\begin{equation}
  |r(t(-s))-4M|\lesssim \frac{1}{\exp((\frac{1}{8\sqrt{2}M}-\delta)|t(-s)|)},\qquad |p^r(t(-s))|\lesssim \frac{1}{\exp((\frac{1}{8\sqrt{2}M}-\delta)|t(-s)|)},
\end{equation}
for all $s\in [-a,0]$. 
\end{proposition}

\begin{proof}
We suppose first that $(x,p)\in W^-(4M)$. By the invariance of the stable manifold $W^-(4M)$, we have that $(x(s),p(s))\in W^-(4M)$ for all $s\geq 0$. By the characterisation \eqref{iden_charact_stable_4m} of the stable manifolds, we have that $\{\psi_{-}^{\ell}=0\}$, so $$p^r =\dfrac{\sqrt{2M}}{r^{\frac{1}{2}}}\Big(1-\frac{4M}{r}\Big).$$ As a result, the defining function $\psi_{+}^{4M}$ satisfies
\begin{equation}\label{cancell_varphi_on_stablemfl_22}
\psi_{+}^{4M}(x,p)=2p^r=2\dfrac{\sqrt{2M}}{r^{\frac{1}{2}}}\Big(1-\frac{4M}{r}\Big).
\end{equation}

Next, we integrate the derivative \eqref{expcontrmeetseparatr} along the geodesic flow by $$\psi_{+}^{4M}\Big(x(t(s)),p(t(s))\Big)=\psi_{+}^{4M}(x,p)\exp \Big(-\int_{0}^{t(s)}\frac{\sqrt{2M}}{2}\frac{1}{r^{\frac{3}{2}}} \Big(1-\frac{12M}{r}\Big)\Big(1-\frac{2M}{r}\Big)\d t \Big),$$ where here $(x(0)=x,\, p(0)=p)$. Let us set $\epsilon>0$ small enough so that $$\Big|\frac{1}{8\sqrt{2}M}-\frac{\sqrt{2M}}{2}\frac{1}{r^{\frac{3}{2}}} \Big(1-\frac{12M}{r}\Big)\Big(1-\frac{2M}{r}\Big)\Big| \leq \delta,$$ for all $(x,p)\in \{|r-4M|\leq  \epsilon\}$. As a result, we have the upper bound
\begin{align*}
\Big|\psi_{+}^{4M}\Big(x(t(s)),p(t(s))\Big)\Big|&\lesssim \exp \Big(-\int_{0}^{t(s)} \frac{\sqrt{2M}}{2}\frac{1}{r^{\frac{3}{2}}} \Big(1-\frac{12M}{r}\Big)\Big(1-\frac{2M}{r}\Big)\d t \Big)\\
&\lesssim \exp \Big(-\Big(\frac{1}{8\sqrt{2}M}-\delta\Big)t(s)\Big).
\end{align*}
Finally, we use the relations \eqref{cancell_varphi_on_stablemfl_22} to show that $$|r(t(s))-4M|+|p^r(t(s))|\lesssim \Big|\psi_{+}^{4M}\Big(x(t(s)),p(t(s))\Big)\Big|\lesssim \exp \Big(-\Big(\frac{1}{8\sqrt{2}M}-\delta\Big)t(s)\Big),$$ where we have used lower bounds for the factors in \eqref{cancell_varphi_on_stablemfl_22}. The case of a geodesic $\gamma_{x,p}$ with $(x,p)\in W^+(4M)$ follows similarly.
\end{proof}

\subsubsection{Expansion and contraction of the radial flow with $\ell\in (4M,\infty)$}  

Let us set the function $H_{\ell} \colon [0,\infty) \to\R$ given by 
\begin{equation*}
H_{\ell}(E):=\frac{2M}{E^2-1}-\frac{2M}{E^2_-(\ell)-1},
\end{equation*}
with $\ell\in(4M,\infty)$. Note that $H_{\ell}(E)$ is conserved along the geodesic flow when the angular momentum is $\ell\geq 0$. By Proposition \ref{prop_stablemfldunsttrapp}, the union $W^{+}(\ell) \,\cup \,W^{-}(\ell)$ of the stable manifolds associated to the sphere of trapped orbits $\mathcal{S}^-(\ell)$, is characterised as $\{H_{\ell}(E)=0\}$. 

It will be convenient to parametrise timelike geodesics by the time coordinate $\bar{t}=s(E^2-1)^{\frac{1}{2}}$ in the domain $\{E>1\}$. We define the derivative along the geodesic flow with respect to $\bar{t}$ as $\frac{\d}{\d \bar{t}}=(E^2-1)^{-\frac{1}{2}}\frac{\d}{\d s}$.

\begin{proposition}
Let us consider the radial flow with $l>4M$. The quantity $H_{\ell}(E)$ satisfies 
\begin{equation}\label{identity_impact_parameter_massive_fields_2}
H_{\ell}(E)=\frac{r^3}{r^2-\frac{\ell^2}{2M}r+\ell^2}\Big(\frac{(p^r)^2}{E^2-1}-\Big(1+\frac{a(\ell)}{r}\Big)\Big(1-\frac{r_-(\ell)}{r}\Big)^2\Big),
\end{equation}
where $$a(\ell)=\frac{2M^2r_-(\ell)^3}{\ell^2(4M-r_-(\ell))(r_-(\ell)-3M)}.$$ Set the functions $\varphi_{\pm}^{\ell}\colon \P\, \cap\, \{E>1\}\to\R$ given by $$\varphi_{\pm}^{\ell}(x,p):=\frac{r^{\frac{3}{2}}}{(r^2-\frac{\ell^2}{2M}r+\ell^2)^{\frac{1}{2}}}\Big(\frac{p^r}{(E^2-1)^{\frac{1}{2}}}\pm\Big(1+\frac{a(\ell)}{r}\Big)^{\frac{1}{2}}\Big(1-\frac{r_-(\ell)}{r}\Big)\Big).$$ Moreover, the derivative of $\varphi_{\pm}^{\ell}$ along the geodesic flow is 
\begin{equation}\label{expcontrac_hyperbol_2}
\dfrac{\d\varphi_{\pm}^{\ell}}{\d \bar{t}}=\pm \frac{r^{\frac{1}{2}}(r-r_+(\ell))}{2(r+a(\ell))^{\frac{1}{2}}(r^2-\frac{\ell^2}{2M}r+\ell^2)}\varphi_{\pm}^{\ell}.
\end{equation}
\end{proposition}

\begin{proof}
For $\ell\in (4M,\infty)$, we have $a(\ell)>0$ by definition. Let $\bar{r}_+(\ell)> \bar{r}_-(\ell)>2M$ be the roots of $r^2-\frac{\ell^2}{2M}r+\ell^2$. The functions $\varphi_{\pm}^{\ell}$ are well-defined on the set $\{\bar{r}_-<r<\bar{r}_+\}$, where the polynomial $r^2-\frac{\ell^2}{2M}r+\ell^2$ is positive. The identities \eqref{identity_impact_parameter_massive_fields_2} and \eqref{expcontrac_hyperbol_2} follow by the proof of Proposition \ref{prop_hamilt_exp_contr_main_hyp}.
\end{proof}

We note that the term in the RHS of \eqref{expcontrac_hyperbol_2} satisfies $$\frac{r^{\frac{1}{2}}(r-r_+(\ell))}{2(r+a(\ell))^{\frac{1}{2}}(r^2-\frac{\ell^2}{2M}r+\ell^2)}\Big|_{r=r_-(\ell)}<0.$$ This shows the local expansion and contraction properties of the geodesic flow near the unstable and stable manifolds, respectively. By the identity \eqref{expcontrac_hyperbol_2}, the Lyapunov exponent of $(r_-(\ell),0)$ can also be written as 
\begin{equation}\label{repres_form_lyap_exp_large_angu_mom}
\lambda^{\frac{1}{2}}(\ell)=\frac{(E^2_-(\ell)-1)^{\frac{1}{2}}}{2E_-(\ell)}\frac{(r_+(\ell)-r_-(\ell))(r_-(\ell)-2M)}{r^{\frac{1}{2}}_-(a(\ell)+r_-(\ell))^{\frac{1}{2}}(r^2_-(\ell)-\frac{\ell^2}{2M}r_-(\ell)+\ell^2)}.
\end{equation}

\begin{remark}
The zero sets of the functions $\varphi_{+}^{\ell}$ and $\varphi_{-}^{\ell}$ define the unstable and stable manifolds associated to the sphere of trapped orbits $\mathcal{S}^{-}(\ell)$, respectively. In other words, the stable manifolds $W^{\pm}(\ell)$ can be written as 
\begin{align*}
W^{\pm}(\ell)&=\Big\{(x,p)\in \mathcal{P}: \ell(x,p)=\ell,\quad \varphi_{\mp}^{\ell}(x,p)=0\Big\}.
\end{align*}
This characterisation of the stable manifolds $W^{\pm}(\ell)$ holds similarly as in the case when $\ell=4M$.
\end{remark}

\subsection{Degenerate trapping at ISCO}

The unique sphere in $\mathcal{P}$ containing geodesics with $\ell=2\sqrt{3}M$ is $\mathcal{S}^{-}(2\sqrt{3}M)$. This property contrasts with the case of geodesics with $\ell>2\sqrt{3}M$, which can be contained in two different spheres. On the other hand, there are no geodesics contained in spheres of fixed radii when $\ell<2\sqrt{3}M$. These properties hold because of a \emph{bifurcation of the radial dynamics at} $\ell=2\sqrt{3}M$.

We recall that the fixed point $(6M,0)$ of the radial flow with $\ell=2\sqrt{3}M$ is not hyperbolic, since $\Phi_{2\sqrt{3}M}(6M)=0$. In spite of this, there are still suitable stable and unstable manifolds $W^{\pm}(2\sqrt{3}M)$ associated to the sphere $\mathcal{S}^{-}(2\sqrt{3}M)$. The submanifolds $W^{\pm}(2\sqrt{3}M)$ are contained on the energy level $\{E=\frac{2\sqrt{2}}{3}\}$ of the radial flow when $\ell=2\sqrt{3}M$. We parametrise the radial momentum coordinate $p^r$ of future-trapped and past-trapped geodesics with $\ell=2\sqrt{3}M$ and $E=\frac{2\sqrt{2}}{3}$, by $$p^{r,-}_{2\sqrt{3}M}(r):=\dfrac{1}{3}\Big(\frac{6M}{r}-1\Big)^{\frac{3}{2}} \qquad \text {and}\qquad p^{r,+}_{2\sqrt{3}M}(r):=-\dfrac{1}{3}\Big(\frac{6M}{r}-1\Big)^{\frac{3}{2}},$$ respectively. We summarise this discussion with the following proposition.

\begin{proposition}\label{prop_stablemflddegentrapp}
The stable manifolds of the sphere of trapped orbits $\mathcal{S}^{-}(2\sqrt{3}M)$ are analytic codimension two submanifolds of $\mathcal{P}$, given by
\begin{align*}
    W^+(2\sqrt{3}M)=~&\Big\{(x,p)\in\P: \ell(x,p)=2\sqrt{3}M,\quad p^r=p^{r,+}_{2\sqrt{3}M} \Big\},\\  
    W^-(2\sqrt{3}M)=~&\Big\{ (x,p)\in\P: \ell(x,p)=2\sqrt{3}M,\quad p^r=p^{r,-}_{2\sqrt{3}M}\Big\}.
\end{align*}
In particular, the intersection $W^+(2\sqrt{3}M)\,\cap \,W^-(2\sqrt{3}M)$ is equal to the sphere $\mathcal{S}^{-}(2\sqrt{3}M)$.
\end{proposition}

By definition, the stable manifolds $W^{\pm}(2\sqrt{3}M)$ are defined in the region $\{r\leq 6M\}$.

\subsubsection{Expansion and contraction of the radial flow with $\ell=2\sqrt{3}M$}

Let us set the function $H_{2\sqrt{3}M}\colon [0,\infty)\to\R$ given by 
\begin{equation*}
H_{2\sqrt{3}M}(E):=\frac{2M}{1-E^2}-18M.
\end{equation*}
Note that $H_{2\sqrt{3}M}(E)$ is conserved along the geodesic flow when $\ell=2\sqrt{3}M$. The union $W^{+}(2\sqrt{3}M) \,\cup \, W^{-}(2\sqrt{3}M)$ of the stable manifolds associated to the sphere $\mathcal{S}^-(2\sqrt{3}M)$, is characterised as $\{H_{2\sqrt{3}M}(E)=0\}$.

\begin{proposition}\label{propexpcontracisoc}
Consider the radial flow with $\ell=2\sqrt{3}M$. The conserved quantity $H_{2\sqrt{3}M}(E)$ satisfies 
\begin{equation}\label{idHamiltisco}
H_{2\sqrt{3}M}(E)=\frac{r^3}{r^2-6Mr+12M^2}\Big(\frac{(p^r)^2}{1-E^2}-\Big(\frac{6M}{r}-1\Big)^3\Big).
\end{equation}
Set the functions $\varphi_{\pm}^{2\sqrt{3}M}\colon \P\,\cap\,\{E<1\}\to\R$ given by $$\varphi_{\pm}^{2\sqrt{3}M}(x,p):=\frac{r^{\frac{3}{2}}}{(r^2-6Mr+12M^2)^{\frac{1}{2}}}\Big(\frac{p^r}{(1-E^2)^{\frac{1}{2}}}\mp\Big(\frac{6M}{r}-1\Big)^{\frac{3}{2}}\Big).$$ Then, the derivative of $\varphi_{\pm}^{2\sqrt{3}M}$ along the geodesic flow is \begin{equation}\label{isco_expansion_contraction}
\dfrac{\d\varphi_{\pm}^{2\sqrt{3}M}}{\d \bar{t}}=\pm \frac{r^{\frac{1}{2}}(6M-r)^{\frac{1}{2}}}{2(r^2-6Mr+12M^2)}\varphi_{\pm}^{2\sqrt{3}M}.
\end{equation}  
\end{proposition}

\begin{proof}
We first note that the functions $\varphi_{\pm}^{2\sqrt{3}M}$ are well-defined on $\{r\leq 6M\}$ since $$r^2-6Mr+12M^2=(r-3M)^2+3M^2\geq 3M^2.$$ The identities \eqref{idHamiltisco} and \eqref{isco_expansion_contraction} follow by the proof of Proposition \ref{prop_hamilt_exp_contr_main_hyp}. 
\end{proof}

We observe that the term in the RHS of \eqref{isco_expansion_contraction} satisfies 
\begin{equation}\label{remark_expansion_contraction_degenerat}
\frac{r^{\frac{1}{2}}(6M-r)^{\frac{1}{2}}}{2(r^2-6Mr+12M^2)}>0,
\end{equation}
whenever $r< 6M$. This shows the local expansion and contraction properties of the geodesic flow near the unstable and stable manifolds, respectively. However, the expansion and contraction properties \emph{degenerate}, since the function in \eqref{remark_expansion_contraction_degenerat} vanishes at $r=6M$.

\begin{remark}
The zero sets of the functions $\varphi_{+}^{2\sqrt{3}M}$ and $\varphi_{-}^{2\sqrt{3}M}$ define the unstable and stable manifolds associated to the sphere $\mathcal{S}^{-}(2\sqrt{3}M)$, respectively. In other words, the stable manifolds $W^{\pm}(2\sqrt{3}M)$ can be written as 
\begin{equation}\label{ident_charact_stable_mld_isco}
W^{\pm}(2\sqrt{3}M)=\Big\{(x,p)\in \mathcal{P}: \ell(x,p)=2\sqrt{3}M,\quad \varphi_{\mp}^{2\sqrt{3}M}(x,p)=0\Big\}.
\end{equation}
\end{remark}

For the fixed point $(r=6M,\, p^r=0)$, the rate of contraction and expansion for the geodesic flow on the stable manifolds $W^{\pm}(6M)$ is of the form $t^3$. We proceed to show this rate of contraction and expansion in a neighbourhood of $\{r=6M,\, p^r=0\}$.

\begin{proposition}
For every geodesic $\gamma_{x,p} \colon [0,a]\to \{5M\leq r\leq 6M \}$ with $(x,p)\in W^-(2\sqrt{3}M)$, we have
\begin{equation}
  |r(t(s))-6M|^{\frac{3}{2}}\lesssim \frac{1}{t^3(s)},\qquad\quad |p^r(t(s))|\lesssim \frac{1}{t^3(s)},\qquad \forall s\in [0,a].
\end{equation} 
Moreover, for every geodesic $\gamma_{x,p} \colon [-a,0]\to \{ 5M\leq r \leq 6M\}$ with $(x,p)\in W^+(2\sqrt{3}M)$, we have
\begin{equation}
  |r(t(-s))-6M|^{\frac{3}{2}}\lesssim \frac{1}{|t^3(-s)|},\qquad\quad |p^r(t(-s))|\lesssim \frac{1}{|t^3(-s)|},\qquad \forall s\in [-a,0].
\end{equation}
\end{proposition}

\begin{proof}
We suppose first that $(x,p)\in W^-(2\sqrt{3}M)$. By the invariance of the stable manifold $W^-(4M)$, we have that $(x(s),p(s))\in W^-(2\sqrt{3}M)$ for all $s\geq 0$. By the characterisation \eqref{ident_charact_stable_mld_isco} of the stable manifolds, we have that $\{\varphi_{+}^{2\sqrt{3}M}=0\}$, so 
\begin{equation}\label{radial_zero_defn_funct_stable_isco}
3p^r=\Big(\frac{6M}{r}-1\Big)^{\frac{3}{2}},
\end{equation}
since $E=\frac{2\sqrt{2}}{3}.$ As a result, the defining function $\varphi_{-}^{2\sqrt{3}M}$ satisfies that $$\varphi_{-}^{2\sqrt{3}M}(x,p)=2\omega_{2\sqrt{3}M}(r)\frac{p^r}{(1-E^2)^{\frac{1}{2}}}=2\omega_{2\sqrt{3}M}(r)\Big(\frac{6M}{r}-1\Big)^{\frac{3}{2}}, $$ where $ \omega_{2\sqrt{3}M}(r)=r^{\frac{3}{2}}(r^2-6Mr+12M^2)^{-\frac{1}{2}}$ is a positive radial weight.

If $(x,p)\in W^-(2\sqrt{3}M)$, then $\gamma$ is outgoing. So, by integrating the relation \eqref{radial_zero_defn_funct_stable_isco}, we obtain
\begin{align}
\begin{aligned}\label{estimate_droping_rate_exponential_to_polynomial}
    \dfrac{s}{3}&=\int_{r(0)}^{r(s)}\frac{r^{\frac{3}{2}}}{(6M-r)^{\frac{3}{2}}}\d r\\
    &=-18M\arcsin{\sqrt{\dfrac{r(s)}{6M}}}+\dfrac{\sqrt{r(s)}(18M-r(s))}{\sqrt{6M-r(s)}}+18M\arcsin{\sqrt{\dfrac{r(0)}{6M}}}-\dfrac{\sqrt{r(0)}(18M-r(0))}{\sqrt{6M-r(0)}}\\
    &\lesssim 1+\dfrac{1}{\sqrt{6M-r(s)}},
\end{aligned}    
\end{align}
Moreover, we have the lower bound $$s=\int_{0}^{t(s)} \frac{\d t}{p^t}=\frac{3}{2\sqrt{2}}\int_{0}^{t(s)} \Big(1-\frac{2M}{r}\Big)\d t\gtrsim t(s),$$ since $E=\frac{2\sqrt{2}}{3}$. Therefore, we have 
\begin{equation}\label{estimate_rate_contraction_stable_manifold_isco}
    p^{r}(t(s))=\dfrac{(6M-r(t(s)))^{\frac{3}{2}}}{3r^{\frac{3}{2}}(t(s))}\lesssim t^{-3}(s),
\end{equation}
for all $s\geq 0$. A similar argument treats the case of a geodesic $\gamma_{x,p}$ with $(x,p)\in W^+(2\sqrt{3}M)$.
\end{proof}

\subsection{Parabolic trapping at infinity}

For all $\ell \geq 0$, the radial potential $V_{\ell}$ has a local maximum at infinity. The maximum at infinity takes the value one, which corresponds to the rest mass of the particles in the systems we consider. The radial geodesic equation \eqref{r_second_derivative_equation} shows the existence of \emph{orbits at infinity} contained in the spheres $$\mathcal{S}^{\infty}(\ell):=\Big\{(x,p)\in \P: \ell(x,p)=\ell,\quad r=\infty,\quad p^r=0\Big\}.$$ These orbits have particle energy $E=1$, by the mass-shell relation. We call $\mathcal{S}^{\infty}(\ell)$ \emph{spheres of trapped orbits at infinity}, and the orbits they contain \emph{trapped orbits at infinity}.

We will show that the fixed point $(\sqrt{2M}r^{-\frac{1}{2}}=0,\, p^r=0)$ of the radial flow is \emph{parabolic}. Furthermore, we will identify suitable stable and unstable manifolds $W^{\pm}_1(\ell)$ associated to the sphere of trapped orbits at infinity $\mathcal{S}^{\infty}(\ell)$. The submanifolds $W^{\pm}_1(\ell)$ are contained on the energy level $\{E=1\}$. We parametrise the momentum coordinate $p^r$ of future-trapped and past-trapped orbits at infinity, by $$p_{\ell,1}^{r,-}(r):=\dfrac{\sqrt{2M}}{r^{\frac{3}{2}}}\Big(r^2-\dfrac{\ell^2}{2M}r+\ell^2\Big)^{\frac{1}{2}} \qquad \text {and}\qquad p^{r,+}_{\ell,1}(r):=-\dfrac{\sqrt{2M}}{r^{\frac{3}{2}}}\Big(r^2-\dfrac{\ell^2}{2M}r+\ell^2\Big)^{\frac{1}{2}},$$ respectively.

We summarise this discussion with the following proposition.

\begin{proposition}\label{prop_stablemtrapp_parabolictrapp}
Let $\ell\in [0,4M]$. The stable manifolds of the sphere $\mathcal{S}^{\infty}(\ell)$ at infinity are analytic codimension two submanifolds of $\mathcal{P}$, given by
\begin{align*}
    W_{1}^+(\ell)&=\Big\{(x,p)\in\P: \ell(x,p)=\ell,\quad p^r=p^{r,+}_{\ell,1}  \Big\},\\
     W_{1}^-(\ell)&=\Big\{ (x,p)\in\P:  \ell(x,p)=\ell,\quad p^r=p_{\ell,1}^{r,-} \Big\}.
\end{align*}
Moreover, the intersection $W_1^+(\ell)\, \cap \, W_1^-(\ell)$ is equal to the sphere $\mathcal{S}^{\infty} (\ell)$ when $\ell>4M$, and the intersection $W_1^+(4M)\, \cap \, W_1^-(4M)$ is equal to $\mathcal{S}^{\infty} (4M)\cup \mathcal{S} (4M)$.	
\end{proposition}

If $\ell\in [0,4M)$, the stable manifolds $W_1^{\pm}(\ell)$ are contained in the unbounded region $\{r>2M\}$. By definition of $p^{r,\pm}_{\ell,1}$, we have $r^2-\frac{\ell^2}{2M}r+\ell^2\geq 0$ for all $(x,p)\in W_1^{\pm}(\ell)$. Furthermore, if $\ell\in [0,4M)$, we have $$\forall r>2M,\qquad r^2-\frac{\ell^2}{2M}r+\ell^2=\Big(r-\frac{\ell^2}{4M}\Big)^2+\frac{\ell^2}{16M^2}(16M^2-\ell^2)\geq 0.$$ Thus, $W_1^{\pm}(\ell)$ are indeed contained in the unbounded region $\{r>2M\}$. In contrast, if $\ell=4M$, then the stable manifolds $W_1^{\pm}(4M)$ are only contained in $\{r\geq 4M\}$. This property holds since $r^2-\frac{\ell^2}{2M}r+\ell^2=(r-4M)^2$ if $\ell=4M$. 

\begin{proposition}\label{prop_stablemtrapp_parabolictrapp_high}
Let $\ell> 4M$. The stable manifolds of the sphere $\mathcal{S}^{\infty}(\ell)$ at infinity, are analytic codimension two submanifolds of $\mathcal{P}$, given by
\begin{align*}
    W_{1}^{\pm}(\ell)=\Big\{(x,p)\in\P: \ell(x,p)=\ell,\quad p^r=p^{r,+}_{\ell,1}  \Big\}\cup \Big\{(x,p)\in\P: \ell(x,p)=\ell,\quad p^r=p^{r,-}_{\ell,1}  \Big\},
    \end{align*}
In particular, the stable manifolds $W_1^+(\ell)$ and $W_1^-(\ell)$ are equal. 
\end{proposition}

If $\ell> 4M$, the stable manifolds $W_{1}^{\pm}(\ell)$ are only contained in $\{r\geq \bar{r}_+(\ell)\}$, where $\bar{r}_+(\ell)$ is the larger root of the polynomial $r^2-\frac{\ell^2}{2M}r+\ell^2$. 

\begin{remark}
For $\ell> 4M$, all the orbits in the stable manifolds are of \emph{homoclinic type}. In other words, these orbits satisfy that $\lim_{s\to \pm\infty} (r(s),p^r(s))=(\infty,0)$. Thus, the sets of future-trapped and past-trapped orbits at infinity are both equal, and characterised as $\{E(x,p)=1\}$. 
\end{remark}

\subsubsection{Parabolic behaviour at infinity}

Let us show the parabolic behaviour of the fixed point $(\sqrt{2M}r^{-\frac{1}{2}}=0,\, p^r=0)$ at infinity. For this purpose, we will consider the \emph{McGehee type coordinates} $$x=\sqrt{\frac{2M}{r}},\qquad y=p^r.$$ This coordinate system for the radial flow takes infinity to the origin. Furthermore, it allows us to analyse the behaviour of the radial flow in a neigbourhood of the spheres $\mathcal{S}^{\infty}(\ell)$ at infinity. Similar coordinates were first introduced by McGehee \cite{MG73} for the study of stable manifolds for parabolic fixed points at infinity with applications to celestial mechanics.

\begin{proposition}
Consider the radial flow with angular momentum $\ell$. The geodesic flow for the McGehee type coordinates is given by
\begin{align*}
\frac{\d x}{\d s}=-\frac{x^3}{4M}\partial_y K_l=-\frac{x^3y}{4M},\qquad \quad\frac{\d y}{\d s}=\frac{x^3}{4M}\partial_x K_l=\frac{x^3}{4M}\Big(-x+\frac{\ell^2}{2M^2}x^3-\frac{3\ell^2}{4M^2}x^5\Big),
\end{align*}
where $$K_l(x,y):=\frac{y^2}{2}-\frac{x^2}{2}+\frac{\ell^2}{8M^2}(x^4-x^6)=\frac{E^2-1}{2}.$$ 
\end{proposition}

\begin{proof}
First, the derivative of $x$ along the geodesic flow is $$\frac{\d x}{\d s}=\sqrt{2M}\frac{\d }{\d s}(r^{-\frac{1}{2}})=-\frac{1}{2r^{\frac{3}{2}}}p^r=-\frac{x^3y}{4M}.$$ The derivative of $y$ along the geodesic flow is exactly the radial geodesic equation, so $$\frac{\d y}{\d s}=-\dfrac{M}{r^4}\Big(r^2-\dfrac{\ell^2}{M}r+3\ell^2\Big)=-\frac{M}{r^{\frac{3}{2}}}\Big(\frac{1}{r^{\frac{1}{2}}}-\frac{\ell^2}{M}\frac{1}{r^{\frac{3}{2}}}+3\ell^2\frac{1}{r^{\frac{5}{2}}}\Big)=\frac{x^3}{4M}\Big(-x+\frac{\ell^2}{2M^2}x^3-\frac{3\ell^2}{4M^2}x^5\Big).$$ By the mass-shell relation, the conserved quantity $\frac{E^2-1}{2}$ is equal to $$\frac{E^2-1}{2}=\frac{1}{2}(p^r)^2-\frac{M}{r^{\frac{1}{2}}}\Big(\frac{1}{r^{\frac{1}{2}}}-\frac{\ell^2}{2M}\frac{1}{r^{\frac{3}{2}}}+\frac{\ell^2}{r^{\frac{5}{2}}}\Big)=\frac{y^2}{2}-\frac{x^2}{2}+\frac{\ell^2}{8M^2}(x^4-x^6).$$ By a direct computation of the partial derivatives of $K_l(x,y)$, we obtain the final identities.
\end{proof}

\subsubsection{Expansion and contraction of the radial flow with $E=1$}
Let us consider the conserved quantity along the geodesic flow given by
\begin{equation*}
E^2-1.
\end{equation*}
By Proposition \ref{prop_stablemtrapp_parabolictrapp}, the union $W^{+}_1(\ell)\, \cup \, W^{-}_1(\ell)$ of the stable manifolds associated to the sphere at infinity $\mathcal{S}^{\infty}(\ell)$, is characterised as $\{E^2-1=0\}$.

\begin{proposition}
Let us consider the radial flow with $\ell\geq 0$. The quantity $E^2-1$ satisfies 
\begin{equation}\label{defin_funct_parabolic}
E^2-1=(p^r)^2-\dfrac{2M}{r^3}\Big(r^2-\dfrac{\ell^2}{2M}r+\ell^2\Big)
\end{equation} 
Set the functions $\psi_{\pm}^{\ell}\colon \P\to\R$ given by $$\psi_{\pm}^{\ell}(x,p):=p^r\mp\dfrac{(2M)^{\frac{1}{2}}}{r^{\frac{3}{2}}}\Big(r^2-\dfrac{\ell^2}{2M}r+\ell^2\Big)^{\frac{1}{2}}.$$ Then, the derivative of $\psi_{\pm}^{\ell}$ along the geodesic flow is
\begin{equation}\label{expcontrparabolic}
\dfrac{\d \psi_{\pm}^{\ell}}{\d s}=\mp\frac{\sqrt{2M}(r^2-\frac{\ell^2}{M}r+3\ell^2)}{2r^{\frac{5}{2}}(r^2-\frac{\ell^2}{2M}r+\ell^2)^{\frac{1}{2}}}  \psi_{\pm}^{\ell}.
\end{equation}
\end{proposition}

\begin{proof}
The identity \eqref{defin_funct_parabolic} follows directly by the mass-shell relation. We note that the functions $\psi_{\pm}^{\ell}$ are well-defined whenever $r^2-\frac{\ell^2}{2M}r+\ell^2\geq 0$. Furthermore, as long as $r$ is sufficiently large, we have $r^2-\frac{\ell^2}{2M}r+\ell^2\geq 0$ for all $\ell\geq 0$. We finally show \eqref{expcontrparabolic}. We first compute the derivative
\begin{align}
\dfrac{\d}{\d s}\Big(\frac{\sqrt{2M}}{r^{\frac{3}{2}}}\Big(r^2-\dfrac{\ell^2}{2M}r+\ell^2\Big)^{\frac{1}{2}}\Big)&=\frac{\sqrt{2M}}{r^3}\Big(\dfrac{\d}{\d s}\Big(r^2-\frac{\ell^2}{2M}r+\ell^2\Big)^{\frac{1}{2}} r^{\frac{3}{2}}-\frac{\d}{\d s}r^{\frac{3}{2}}\Big(r^2-\frac{\ell^2}{2M}r+\ell^2\Big)^{\frac{1}{2}}\Big)\nonumber\\
&=-\frac{p^r}{2r^{\frac{5}{2}}(r^2-\frac{\ell^2}{2M}r+\ell^2)^{\frac{1}{2}}}\Big(r^2-\frac{\ell^2}{M}r+3\ell^2\Big).\label{ident_pf_parabolicexpa}
\end{align}
We obtain \eqref{expcontrparabolic} by the radial geodesic equation and \eqref{ident_pf_parabolicexpa}.
\end{proof}

We observe that the term in the RHS of \eqref{expcontrparabolic} satisfies 
\begin{equation}\label{rmk_sign_parabolic_fix_pt}
\frac{\sqrt{2M}(r^2-\frac{\ell^2}{M}r+3\ell^2)}{2r^{\frac{5}{2}}(r^2-\frac{\ell^2}{2M}r+\ell^2)^{\frac{1}{2}}}>0,
\end{equation}
for $r$ sufficiently large in terms of $\ell$. This shows the local expansion and contraction properties of the geodesic flow near the parabolic unstable and stable manifolds, respectively. However, the expansion and contraction properties \emph{degenerate at infinity}, since the function \eqref{rmk_sign_parabolic_fix_pt} vanishes at $r=\infty$. This degeneracy is due to the parabolic behaviour of the fixed point $(\sqrt{2M}r^{-\frac{1}{2}}=0,\,p^r=0)$.

\begin{remark}
The zero sets of the functions $\psi_{+}^{\ell}$ and $\psi_{-}^{\ell}$ define subsets of the parabolic stable and unstable manifolds associated to the sphere $\mathcal{S}^{\infty}(\ell)$, respectively. For all $\ell\in [0,4M]$, we have
\begin{align*}
W_1^{\pm}(\ell)&=\Big\{(x,p)\in \mathcal{P}:  \ell(x,p)=\ell,\quad \psi_{\mp}^{\ell}(x,p)=0\Big\}.
\end{align*}
On the other hand, for all $\ell\in (4M,\infty)$, we have
\begin{align*}
W_{1, \textrm{in/out}}^{\pm}(\ell)&:=\Big\{(x,p)\in \mathcal{P}:  \ell(x,p)=\ell,\quad \psi_{\mp}^{\ell}(x,p)=0\Big\}\subset W_{1}^{\pm}(\ell).
\end{align*}
We note that the geodesics in the subsets $W^{\pm}_{1,\mathrm{in/out}}(\ell)$ of the stable manifolds are always either ingoing or outgoing. In particular, the sets $W^{\pm}_{\mathrm{in/out}}(\ell)$ only contain the ingoing or the outgoing part of the homoclinic orbits in the energy level $\{E=1\}$ when $\ell\in (4M,\infty)$.
\end{remark}

For the parabolic fixed point $(\sqrt{2M}r^{-\frac{1}{2}}=0,\, p^r=0)$, the rate of contraction and expansion for the geodesic flow on the stable manifolds $W^{\pm}_1(\ell)$ is given by $t^{\frac{1}{3}}$. We proceed to show this behaviour using the McGehee type coordinates $(\sqrt{2M}r^{-\frac{1}{2}},p^r)$.

\begin{proposition}
Let $R>2M$. For every geodesic $\gamma_{x,p} \colon [0,a]\to \{r\geq R\}$ with $(x,p)\in W_1^-(\ell)$, we have
$$\sqrt{\frac{2M}{r(t(s))}}\lesssim \frac{1}{t^{\frac{1}{3}}(s)},\qquad\quad |p^r(t(s))|\lesssim \frac{1}{t^{\frac{1}{3}}(s)},\qquad \forall s\in [0,a].$$ Moreover, for every geodesic $\gamma_{x,p} \colon [0,a]\to \{r \geq R\}$ with $(x,p)\in W_1^+(\ell)$, we have
$$\sqrt{\frac{2M}{r(t(-s))}}\lesssim \frac{1}{|t^{\frac{1}{3}}(-s)|},\qquad\quad |p^r(t(-s))|\lesssim \frac{1}{|t^{\frac{1}{3}}(-s)|},\qquad \forall s\in [-a,0].$$
\end{proposition}

\begin{proof}
We suppose first that $(x,p)\in W_1^-(\ell)$. By the invariance of the stable manifold $W_1^-(\ell)$, we have that $(x(s),p(s))\in W_1^-(\ell)$ for all $s\geq 0$. We first consider the case of radial geodesics, i.e. orbits with $\ell=0$. If $\ell=0$ and $E=1$, then the mass-shell relation \eqref{identity_particle_energy_angular_momentum} is given by $(p^r)^2=\frac{2M}{r}.$ We recall that $\gamma$ is outgoing if $(x,p)\in W_1^-(0)$. So, we can integrate the mass-shell relation by $$\frac{2}{3}r^{\frac{3}{2}}(s)-\frac{2}{3} r^{\frac{3}{2}}(0)=\int_{r(0)}^{r(s)}\sqrt{r}\d r=\sqrt{2M}s.$$ Furthermore, we have the lower bound $$s=\sqrt{2M}\int_{0}^{t(s)} \frac{\d t}{p^t}=\sqrt{2M}\int_{0}^{t(s)} \Big(1-\frac{2M}{r}\Big)\d t\gtrsim t(s),$$ since $E=1$. As a result, we have $t(s)\lesssim r^{\frac{3}{2}}(s)$, so 
\begin{equation}\label{pfineqstablemfldsparabolicradial}
\sqrt{\frac{2M}{r(s)}}\lesssim \frac{1}{t^{\frac{1}{3}}(s)}.
\end{equation}
The decay of $p^r$ follows then, by using the identity $(p^r)^2=\frac{2M}{r}.$ The argument works similarly if $(x,p)\in W_1^+(0)$ instead.

Next, we consider the general case of geodesics with $\ell>0.$ If $E=1$, then the mass-shell relation \eqref{identity_particle_energy_angular_momentum} can be written as $(p^r)^2=\frac{2M}{r^3}(r^2-\frac{\ell^2}{2M}r+\ell^2).$ We recall that $\gamma$ is outgoing if $(x,p)\in W_1^-(\ell)$. So, we can integrate the mass-shell relation by $$\frac{2}{3} r^{\frac{3}{2}}(s)-\frac{2}{3}r^{\frac{3}{2}}(0)\gtrsim \int_{r(0)}^{r(s)}\frac{r^{\frac{3}{2}}}{(r^2-\frac{\ell^2}{2M}r+\ell^2)^{\frac{1}{2}}}\d r=\sqrt{2M}s.$$ Furthermore, we have the lower bound $$s=\sqrt{2M}\int_{0}^{t(s)} \frac{\d t}{p^t}=\sqrt{2M}\int_{0}^{t(s)} \Big(1-\frac{2M}{r}\Big)\d t\gtrsim t(s),$$ since $E=1$. As a result, we have $t(s)\lesssim r^{\frac{3}{2}}(s)$, so 
\begin{equation}\label{pfineqstablemfldsparabolicgenera}
\sqrt{\frac{2M}{r(s)}}\lesssim \frac{1}{t^{\frac{1}{3}}(s)}.
\end{equation} 
The decay of $p^r$ follows then, by using $(p^r)^2=\frac{2M}{r^3}(r^2-\frac{\ell^2}{2M}r+\ell^2).$ The argument works similarly if $(x,p)\in W_1^+(\ell)$ instead.
\end{proof}

\section{Concentration estimates on the stable manifolds}\label{section_radial_flow}

In this section, we prove concentration estimates on the stable manifolds associated to the three forms of trapping on $\mathcal{D}$: unstable trapping, degenerate trapping at ISCO, and parabolic trapping at infinity. These estimates will be obtained by integrating the radial flow in the different regions of the mass-shell, depending on the form of the radial potential $V_{\ell}$. We recall that the mass-shell relation can be written as
\begin{equation}\label{radial_geodesic_eqn_rewritten_section_radial_coordinates}
    E^2=(p^r)^2+V_{\ell}(r),
\end{equation}
which can be integrated along the geodesic flow. We decompose the dispersive region $\D$ into the three invariant regions
\begin{align*}
  \D_{\mathrm{low}}:=&~\Big\{(x,p)\in \D: \ell(x,p)\in[0,2\sqrt{3}M] \Big\},\\
  \D_{\mathrm{bd}}:=&~\Big\{(x,p)\in \D: \ell(x,p)\in[2\sqrt{3}M,4M] \Big\},\\
   \D_{\mathrm{high}}:=&~\Big\{(x,p)\in \D: \ell(x,p)\in [4M,\infty) \Big\}.  
\end{align*}
We note that $  \D=\D_{\mathrm{low}}\,\cup \,\D_{\mathrm{bd}}\,\cup\, \D_{\mathrm{high}}$. In the following three subsections, we show concentration estimates on the stable manifolds in the region $\D_{\mathrm{low}}$, the region $\D_{\mathrm{bd}}$, and the region $\D_{\mathrm{high}}$, respectively.

\subsection{The region $\D_{\mathrm{low}}$}

In this region, the radial potential $V_{\ell}$ satisfies $\frac{\d}{\d r}V_{\ell}\geq 0$ for all $r\geq 2M$. Furthermore, the derivative $\frac{\d}{\d r}V_{\ell}$ only vanishes at $r=6M$ when $\ell=2\sqrt{3}M$. Moreover, the potential $V_{\ell}(r)$ has a maximum at infinity where $\lim_{r\to \infty}V_{\ell}(r)=1$. One can easily show that there are only two types of trapping in $\D_{\mathrm{low}}$: parabolic trapping at infinity, and degenerate trapping at ISCO.

\subsubsection{The subregion $E\sim1$}

Let us study concentration estimates on the parabolic manifolds at the energy level $\{E=1\}$. Parabolic trapping at infinity holds for all $\ell\in [0,2\sqrt{3}M]$ in $\D_{\mathrm{low}}$. We will estimate the radial flow in a uniform neighbourhood of the parabolic manifolds, where geodesics spend arbitrarily long periods of time in the far-away region, before crossing $\mathcal{H}^+$. We begin considering the case of radial geodesics, in other words, orbits with vanishing angular momentum.

\begin{proposition}\label{proposition_slow_decay_particle_energy_one_include_radial_geodesics_radial}
Let $R>r_0>2M$. For every geodesic $\gamma_{x,p}\colon [0,a]\to \{r>r_0\}$ with angular momentum $\ell=0$, particle energy $E\in (\frac{95}{100},1)$, and initial data $(x,p)\in \{r<R\}$, we have
\begin{equation}\label{estimate_advanced_time_coordinate_timelike_geodesics_low_angular_momentum_particle_energy_one_radial}
 |1-E^2|\lesssim \dfrac{1}{v^{\frac{2}{3}}(s)},\qquad \forall s\in [0,a].
\end{equation}
Moreover, for all $s\in [0,a],$ we have
\begin{align*}
\Big|\sqrt{\dfrac{2M}{r(s)}}-p^r(s)\Big|\lesssim \dfrac{1}{v^{\frac{1}{3}}(s)},\qquad \text{if}\quad p^r\geq 0,\\
\Big|\sqrt{\dfrac{2M}{r(s)}}+p^r(s)\Big|\lesssim \dfrac{1}{v^{\frac{1}{3}}(s)},\qquad \text{if}\quad p^r\leq 0.
\end{align*}
The same decay estimates hold with respect to the time coordinate $u$.
\end{proposition}

\begin{proof}
If $\ell=0$, then the mass-shell relation \eqref{radial_geodesic_eqn_rewritten_section_radial_coordinates} can be written as
\begin{equation}\label{geoeqnradiallemzero}
1-E^2=\frac{2M}{r}-(p^r)^2.
\end{equation}
So, for radial geodesics with $E<1$, the radial momentum coordinate $p^r$ vanishes at $r=\frac{2M}{1-E^2}$. We note that the turning point at $r=\frac{2M}{1-E^2}$ moves to infinity as $E\to 1$.

Let us suppose first that $\gamma$ is an outgoing geodesic. Integrating the mass-shell relation \eqref{geoeqnradiallemzero}, we have
\begin{align}
\label{estimate_integration_radial_geodesic_eqn_particle_energy_one}
s&=\int_{r(0)}^{\frac{2M}{1-E^2}}\dfrac{\d r}{(\frac{2M}{r}-(1-E^2))^{\frac{1}{2}}}=\dfrac{1}{\sqrt{1-E^2}}\int_{r(0)}^{\frac{2M}{1-E^2}}\dfrac{\d r}{(\frac{2M}{1-E^2}\frac{1}{r}-1)^{\frac{1}{2}}}.
\end{align}
By using the change of variables $r\mapsto \tilde{r}_0(r):=r (1-E^2)$, we have 
\begin{equation}\label{estim_good_3half_radial}
\dfrac{1}{\sqrt{1-E^2}}\int_{r(0)}^{\frac{2M}{1-E^2}}\dfrac{\d r}{(\frac{2M}{1-E^2}\frac{1}{r}-1)^{\frac{1}{2}}}=\dfrac{1}{(1-E^2)^{\frac{3}{2}}}\int_{r(0) (1-E^2)}^{2M}\dfrac{\d \tilde{r}_0}{(\frac{2M}{\tilde{r}_0}-1)^{\frac{1}{2}}}\lesssim \dfrac{1}{(1-E^2)^{\frac{3}{2}}},
\end{equation}
where the estimate above follows by an explicit computation of the integral term. Moreover, the coordinate $p^v$ satisfies $\Omega^2 p^v\leq E<1$, so 
\begin{equation}\label{estimate_particle_energy_one_radial_geodesics}
s=\int_{v_0}^{v(s)}\frac{\d v}{p^v}\geq \int_{v_0}^{v(s)}\Omega^2(r) \d v \gtrsim v(s),
\end{equation}
since $\gamma(s)\in\{r>r_0>2M\}$ for all $s\in[0,a]$. As a result, we obtain the second estimate in \eqref{estimate_advanced_time_coordinate_timelike_geodesics_low_angular_momentum_particle_energy_one_radial} by $$\Big|\dfrac{2M}{r(s)}-p^r(s)^2\Big|=|1-E^2|\lesssim v^{-\frac{2}{3}}(s),$$ by putting together \eqref{estim_good_3half_radial} and \eqref{estimate_particle_energy_one_radial_geodesics}. Finally, we obtain the first estimate in \eqref{estimate_advanced_time_coordinate_timelike_geodesics_low_angular_momentum_particle_energy_one_radial} by
$$  \Big|\sqrt{\dfrac{2M}{r(s)}}-p^r(s)\Big|\leq  |1-E^2|^{\frac{1}{2}}\lesssim v^{-\frac{1}{3}}(s),$$ where we used the $\frac{1}{2}$-Hölder continuity of the square root.

Analogous estimates hold when $\gamma$ is ingoing instead. For geodesics with a turning point at $r=\frac{2M}{1-E^2}$, we put together the estimates in the case where $\gamma$ is outgoing and ingoing. The decay estimates in \eqref{estimate_advanced_time_coordinate_timelike_geodesics_low_angular_momentum_particle_energy_one_radial} also hold with respect to $u$. For this, we use that $s\gtrsim u(s)$ since $\gamma(s)\in\{r>r_0>2M\}$ for all $s\in[0,a]$.
\end{proof}

We extend next Proposition \ref{proposition_slow_decay_particle_energy_one_include_radial_geodesics_radial} to the case of geodesics with $\ell\leq 2\sqrt{3}M$. In this regime, we will consider geodesics $\gamma_{x,p}$ with particle energy $E\in (\frac{95}{100},1)$, so that we localise the behaviour of the geodesic flow near the parabolic stable manifolds $W^{\pm}_1(\ell)$. 

\begin{proposition}\label{proposition_slow_decay_particle_energy_one_include_radial_geodesics}
Let $R>r_0>12M$. For every geodesic $\gamma_{x,p}\colon [0,a]\to \{r>r_0\}$ with angular momentum $\ell\leq 2\sqrt{3}M$, particle energy $E\in (\frac{95}{100},1)$, and initial data $(x,p)\in \{r<R\}$, we have
\begin{equation}\label{estimate_advanced_time_coordinate_timelike_geodesics_low_angular_momentum_particle_energy_one}
 |1-E^2|\lesssim \dfrac{1}{v^{\frac{2}{3}}(s)},\qquad \forall s\in [0,a].
\end{equation}
Moreover, for all $s\in [0,a],$ we have
\begin{align}
\Big|\dfrac{\sqrt{2M}}{r^{\frac{3}{2}}}\Big(r^2-\dfrac{\ell^2}{2M}r+\ell^2\Big)^{\frac{1}{2}}-p^r\Big|\lesssim \dfrac{1}{v^{\frac{1}{3}}(s)},\qquad \text{if}\quad p^r\geq 0,\label{ref_parab_Dlow_concentr}\\
\Big|\dfrac{\sqrt{2M}}{r^{\frac{3}{2}}}\Big(r^2-\dfrac{\ell^2}{2M}r+\ell^2\Big)^{\frac{1}{2}}+p^r\Big|\lesssim \dfrac{1}{v^{\frac{1}{3}}(s)},\qquad \text{if}\quad p^r\leq 0.
\end{align}
The same decay estimates hold with respect to the time coordinate $u$.
\end{proposition}

\begin{proof}
We first write the mass-shell relation \eqref{radial_geodesic_eqn_rewritten_section_radial_coordinates} as 
\begin{equation}\label{geoeqnnotradialproplow}
1-E^2=\dfrac{2M}{r^3}\Big(r^2-\dfrac{\ell^2}{2M}r+\ell^2\Big)-(p^r)^2.
\end{equation}
Let $E<1$. By the form of the radial potential $V_{\ell}$, we define the unique root of $E^2-V_{\ell}(r)$ as $r_E$. The radial momentum coordinate $p^r$ vanishes at $r=r_E$.  A direct computation shows that $$E^2-V_{\ell}\Big(\frac{2M}{1-E^2}\Big)<0<E^2-V_{\ell}\Big(\frac{M}{1-E^2}\Big),$$ so $r_E\in (\frac{M}{1-E^2},\frac{2M}{1-E^2})$. As a result, the turning point at $r=r_E$ moves to infinity as $E\to 1$.

We suppose first that $\gamma$ is an outgoing geodesic. Integrating the mass-shell relation \eqref{geoeqnnotradialproplow}, we have
\begin{align}
    s&= \int_{r(0)}^{r_E} \dfrac{\d r}{(\frac{2M}{r^3}(r^2-\frac{\ell^2}{2M}r+\ell^2)-(1-E^2))^{\frac{1}{2}}}=\frac{1}{\sqrt{1-E^2}}\int_{r(0)}^{r_E} \frac{\d r}{(\frac{2M}{1-E^2}\frac{1}{r^3}(r^2-\frac{\ell^2}{2M}r+\ell^2)-1)^{\frac{1}{2}}}.\nonumber
\end{align}
Let us set the function $\tilde{r}_{\ell}\colon [2M,\infty)\to \R$ given by $$\tilde{r}_{\ell}(r):=\frac{r^3}{r^2-\frac{\ell^2}{2M}r+\ell^2}(1-E^2).$$ The function $\tilde{r}_{\ell}(r)$ is well-defined on $[2M,\infty)$, since $r^2-\frac{\ell^2}{2M}r+\ell^2$ is positive when $\ell\leq 2\sqrt{3}M$. In the rest of the proof, we denote $\tilde{r}_{\ell}(r)$ simply as $\tilde{r}(r)$. We now compute the derivative $$\frac{\d \tilde{r}}{\d r}=r^2\frac{r^2-\frac{\ell^2}{M}r+3\ell^2}{(r^2-\frac{\ell^2}{2M}r+\ell^2)^2}(1-E^2).$$ We note that $(1-E^2)\frac{\d \tilde{r}}{\d r}(r)>0$ for $r>r_0>12M$. By using the change of variables $r\mapsto \tilde{r}(r)$, we have $$\frac{1}{\sqrt{1-E^2}}\int_{r(0)}^{r_E} \frac{\d r}{(\frac{2M}{1-E^2}\frac{1}{r^3}(r^2-\frac{\ell^2}{2M}r+\ell^2)-1)^{\frac{1}{2}}}=\dfrac{1}{(1-E^2)^{\frac{3}{2}}}\int_{\tilde{r}(0) }^{\tilde{r}(r_E)}\dfrac{\d \tilde{r}}{(\frac{2M}{\tilde{r}}-1)^{\frac{1}{2}}}\lesssim \dfrac{1}{(1-E^2)^{\frac{3}{2}}},$$ where the estimate above follows by an explicit computation of the integral term. As a result, we obtain the second estimate in \eqref{estimate_advanced_time_coordinate_timelike_geodesics_low_angular_momentum_particle_energy_one} by $$\Big|\dfrac{2M}{r(s)}-\dfrac{\ell^2}{r^2(s)}+\dfrac{2M\ell^2}{r^3(s)}-p^r(s)^2\Big|=|1-E^2|\lesssim v^{-\frac{2}{3}}(s),$$ where we have used that $s\gtrsim v(s)$ since $\gamma([0,a]) \subset \{r>r_0\}$. See the estimate \eqref{estimate_particle_energy_one_radial_geodesics}. Finally, the we obtain the estimate \eqref{ref_parab_Dlow_concentr} by
$$ \Big|\sqrt{\dfrac{2M}{r^3}\Big(r^2-\dfrac{\ell^2}{2M}r+\ell^2\Big)}-p^r\Big|\leq  |1-E^2 |^{\frac{1}{2}}\lesssim v^{-\frac{1}{3}}(s),$$ where we used the $\frac{1}{2}$-Hölder continuity of the square root.

Analogous estimates hold when $\gamma$ is ingoing instead. For geodesics with a turning point at $\tilde{r}_{\ell}(r_E)$, we put together the estimates in the case where $\gamma$ is outgoing and ingoing. The decay estimates in \eqref{estimate_advanced_time_coordinate_timelike_geodesics_low_angular_momentum_particle_energy_one} also hold with respect to $u$. For this, we use that $s\gtrsim u(s)$ since $\gamma([0,a])\in\{r>r_0>12M\}$. 
\end{proof}

\begin{remark}
The decay rate $t^{-{\frac{2}{3}}}$ in Proposition \ref{proposition_slow_decay_particle_energy_one_include_radial_geodesics} holds because the radial potential $V_{\ell}(r)$ satisfies $|V_{\ell}(r)\,-\,1|\sim 2Mr^{-1}$ when $r$ goes to infinity. More generally, for a radial potential $\Phi(r)$ satisfying that $|\Phi(r)\,-\,1|\sim r^{-n}$ when $r$ goes to infinity, one could show concentration estimates with the decay rate $t^{-\frac{2n}{n+2}}$. This observation is relevant when considering decay estimates for massive Vlasov fields on higher dimensional Schwarzschild black holes.
\end{remark}

\subsubsection{The subregion $\ell\sim 2\sqrt{3}M$ and $E\sim\frac{2\sqrt{2}}{3}$ }

Let us study concentration estimates on the stable manifolds associated to the degenerate trapping at ISCO in the region $\D_{\mathrm{low}}$. We recall that this form of trapping holds for $\ell=2\sqrt{3}M$ and $E=\frac{2\sqrt{2}}{3}$. We will control the radial flow in a uniform neighbourhood of the sphere $\mathcal{S}_-(2\sqrt{3}M)$. We will first consider the case of geodesics contained in $\{\ell=2\sqrt{3}M\}$. Note that the set $\{\ell=2\sqrt{3}M\}$ is located at the boundary of $\D_{\mathrm{low}}$. 

\begin{proposition}\label{prop_isco_nearatstablemfl}
For every geodesic $\gamma_{x,p}\colon [0,a]\to \{5M<r<7M\}$ with angular momentum $\ell= 2\sqrt{3}M$ and particle energy $E\in (\frac{2}{3},\frac{21\sqrt{2}}{30})$, we have
$$\Big|\frac{2M}{1-E^2}-18M\Big|\lesssim \frac{1}{v^6(s)},\qquad \forall s\in [0,a].$$ Moreover, for all $s\in [0,a]$, we have 
\begin{align}
|p^r(s)|&\lesssim \frac{1}{v^{3}(s)},\qquad \text{if} \quad r(s)\geq 6M,\\
\Big|\frac{p^r}{\sqrt{1-E^2}}-\Big(\frac{6M}{r}-1\Big)^{\frac{3}{2}}\Big|&\lesssim \frac{1}{v^{3}(s)},\qquad \text{if} \quad r(s)\leq 6M.
\end{align}
The same decay estimates hold with respect to the time coordinate $u$.
\end{proposition}

\begin{proof}
By Proposition \ref{propexpcontracisoc}, the mass-shell relation \eqref{radial_geodesic_eqn_rewritten_section_radial_coordinates} can be written as 
\begin{equation}\label{int_isco_good_geod_eqn}
H_{2\sqrt{3}M}(E):=\frac{2M}{1-E^2}-18M=\omega_{2\sqrt{3}M}\Big(\frac{(p^r)^2}{1-E^2}-\Big(\frac{6M}{r}-1\Big)^3\Big),  
\end{equation}
where $\omega_{2\sqrt{3}M}(r):=r^{3}(r^2-6Mr+12M^2)^{-1}$. For simplicity, we write below $H$ and $\omega$, for $H_{2\sqrt{3}M}$ and $\omega_{2\sqrt{3}M}$, respectively. By the form of the radial potential $V_{\ell}$, we define $r_{H}$ as the unique real value such that $$H=\frac{(r_H-6M)^3}{r_H^2-6Mr_H+12M^2}.$$ The radial momentum coordinate $p^r$ vanishes at $r=r_H $ by \eqref{int_isco_good_geod_eqn}. Moreover, one can show that $r_H \in \{5M<r<7M\}$.

We suppose first that $\gamma$ is an outgoing geodesic. We will consider three different cases: geodesics in the region $\{r\geq 6M\}$, in the region $\{r\leq 6M, \, H>0\}$, or in the region $\{r\leq 6M,\, H<0\}$.

\emph{The region $\{r\geq 6M\}$.} We first note that orbits in $\{r\geq 6M\}$ satisfy $H>0$ by definition of $H$. Without loss of generality, we consider the case when $r(0)=6M$. Integrating the geodesic equation \eqref{int_isco_good_geod_eqn}, we have
\begin{align}
\bar{t}&=\int_{6M}^{r_H}\frac{\d r}{(\frac{H}{\omega}-(1-\frac{6M}{r})^3)^{\frac{1}{2}}}=\frac{1}{H^{\frac{1}{2}}}\int_{6M}^{r_H}\frac{\omega^{\frac{1}{2}} \d r}{(1-(1-\frac{6M}{r})^3\frac{\omega}{H})^{\frac{1}{2}}},\label{ISCO_change_variables}
\end{align} 
since the geodesic $\gamma$ is outgoing.
Let us set the function $\tilde{r}_{2\sqrt{3}M}\colon [2M,\infty)\to\R$ given by 
\begin{equation}\label{change_coord_isco_estimat}
\tilde{r}_{2\sqrt{3}M}(r)=6M+\frac{1}{H^{\frac{1}{3}}}\frac{r-6M}{(r^2-6Mr+12M^2)^{\frac{1}{3}}}.
\end{equation}
The function $\tilde{r}_{2\sqrt{3}M}(r)$ is well-defined on $[2M,\infty)$ since $r^2-6Mr+12M^2\geq 3M^2.$ In the rest of the proof, we denote $\tilde{r}_{2\sqrt{3}M}$ simply as $\tilde{r}$. We compute the radial derivative $$\frac{\d \tilde{r}}{\d r}=\frac{1}{H^{\frac{1}{3}}}\Big(\frac{1}{(r^2-6Mr+12M^2)^{\frac{1}{3}}}-\frac{2(r-3M)(r-6M)}{3(r^2-6Mr+12M^2)^{\frac{4}{3}}}\Big).$$ We note that $$H^{\frac{1}{3}} \frac{\d \tilde{r}}{\d r}\Big|_{r=6M}=\frac{1}{\sqrt[3]{12M^2}}>0.$$ Furthermore, the derivative $\frac{\d \tilde{r}}{\d r}$ is positive in a neighbourhood of $r=6M$. 

By using the change of variables $r\mapsto \tilde{r}(r)$, we have
\begin{align}
\frac{1}{H^{\frac{1}{2}}}\int_{6M}^{r_H}\frac{\omega^{\frac{1}{2}} \d r}{(1-(1-\frac{6M}{r})^3\frac{\omega}{H})^{\frac{1}{2}}}&=\frac{1}{H^{\frac{1}{6}}}\int_{6M}^{ \tilde{r}(r_H)}\frac{\omega^{\frac{1}{2}}}{(1-(\tilde{r}-6M)^3)^{\frac{1}{2}}}\frac{\d r}{\d \tilde{r}} \d \tilde{r} \nonumber\\
&\lesssim  \frac{1}{H^{\frac{1}{6}}}\int_{6M}^{ \tilde{r}(r_H)}\frac{\d \tilde{r}}{(1-(\tilde{r}-6M)^3)^{\frac{1}{2}}}\lesssim \frac{1}{H^{\frac{1}{6}}},\label{estimate_isco_integral_bound}
\end{align} 
where the last estimate follows by an explicit computation of the integral term. As a result, we obtain $\sqrt{H}\lesssim v^{-3}(s)$, where we have used that $\bar{t}\gtrsim v(s)$ since $\gamma([0,a])\subset \{5M<r<7M\}$. We conclude this case since $$|p^r(s)|\lesssim \sqrt{H}\lesssim v^{-3}(s),$$ by the positivity of the second term in the definition of $H$ in \eqref{int_isco_good_geod_eqn}. 

\emph{The region $\{r\leq 6M,\, H>0\}$.} Without loss of generality, we consider the case when $r(a)=6M$. Integrating the geodesic equation \eqref{int_isco_good_geod_eqn}, we have \begin{align*}
\bar{t}&=\int_{r(0)}^{6M}\frac{\d r}{((\frac{6M}{r}-1)^3+\frac{H}{\omega})^{\frac{1}{2}}}=\frac{1}{H^{\frac{1}{2}}}\int_{r(0)}^{6M}\frac{\omega^{\frac{1}{2}}\d r}{((\frac{6M}{r}-1)^3\frac{\omega}{H}+1)^{\frac{1}{2}}}.
\end{align*} 
By using the change of variables $r\mapsto \tilde{r}(r)$, we have
$$\frac{1}{H^{\frac{1}{2}}}\int_{r(0)}^{6M}\frac{\omega^{\frac{1}{2}}\d r}{((\frac{6M}{r}-1)^3\frac{\omega}{H}+1)^{\frac{1}{2}}}=\frac{1}{H^{\frac{1}{6}}}\int_{\tilde{r}(0)}^{\tilde{r}(6M)}\frac{\omega^{\frac{1}{2}} }{((6M-\tilde{r})^3+1)^{\frac{1}{2}}}\frac{\d r}{\d \tilde{r}}\d \tilde{r} \lesssim H^{-\frac{1}{6}},$$ where the estimate follows in the same way as in \eqref{estimate_isco_integral_bound}. As a result, we obtain $\sqrt{H}\lesssim v^{-3}(s)$, where we have used that $\bar{t}\gtrsim v(s)$ since $\gamma([0,a])\subset \{5M<r<7M\}$. We conclude this case since $$\Big|\frac{p^r}{\sqrt{1-E^2}}-\Big(\frac{6M}{r}-1\Big)^{\frac{3}{2}}\Big|\lesssim \sqrt{H}\lesssim v^{-3}(s),$$ where we have used a lower bound for $\omega_{2\sqrt{3}M}$, and the $\frac{1}{2}$-Hölder continuity of the square root. 

\emph{The region $\{r\leq 6M,\, H<0\}$.} Without loss of generality, we consider the case when $r(0)=5M$. Integrating the geodesic equation \eqref{int_isco_good_geod_eqn}, we have $$\bar{t}=\int_{5M}^{r_H}\frac{\d r}{((\frac{6M}{r}-1)^3-\frac{(-H)}{\omega})^{\frac{1}{2}}}=\frac{1}{(-H)^{\frac{1}{2}}}\int_{5M}^{r_H}\frac{w^{\frac{1}{2}}\d r}{(\frac{\omega}{-H}(\frac{6M}{r}-1)^3-1)^{\frac{1}{2}}},$$ for all $\bar{t}$. By using the change variables $$r\mapsto \tilde{r}(r)=6M+\frac{1}{(-H)^{\frac{1}{3}}}\frac{r-6M}{(r^2-6Mr+12M^2)^{\frac{1}{3}}},$$ we have
$$\frac{1}{(-H)^{\frac{1}{2}}}\int_{5M}^{r_H}\frac{\d r}{(\frac{\omega}{-H}(\frac{6M}{r}-1)^3-1)^{\frac{1}{2}}}=\frac{1}{(-H)^{\frac{1}{6}}}\int_{\tilde{r}(5M)}^{\tilde{r}(r_H)}\frac{\omega^{\frac{1}{2}}}{((6M-\tilde{r})^3-1)^{\frac{1}{2}}} \frac{\d r}{\d \tilde{r}}\d \tilde{r}\lesssim \frac{1}{(-H)^{\frac{1}{6}}}.$$ As a result, we obtain $\sqrt{-H}\lesssim v^{-3}$, where we have used that $\bar{t}\gtrsim v(s)$ since $\gamma([0,a])\subset \{5M<r<7M\}$. We conclude this case since $$\Big|\frac{p^r}{\sqrt{1-E^2}}-\Big(\frac{6M}{r}-1\Big)^{\frac{3}{2}}\Big|\lesssim (-H)^{\frac{1}{2}}\lesssim v^{-3}(s),$$ where we have used a lower bound for $\omega_{2\sqrt{3}M}$, and the $\frac{1}{2}$-Hölder continuity of the square root.

Analogous estimates hold in the regime where $\gamma$ is ingoing. For geodesics with a turning point at $r_H$, we put together the estimates when $\gamma$ is outgoing and ingoing in the regions considered above. The same decay estimates hold with respect to the time coordinate $u$. For this, we use that $s\gtrsim u(s)$ since $\gamma([0,a])\subset\{r>5M\}$.  
\end{proof}

Next, we study concentration estimates on the stable manifolds for geodesics with angular momentum $\ell\sim 2\sqrt{3}M$, when $\ell$ is strictly less than $2\sqrt{3}M$. Let us establish uniform estimates for the geodesic flow as $\ell\to  2\sqrt{3}M$. For this purpose, we will perform estimates for all $\ell<2\sqrt{3}M$ near a suitably chosen energy level $\{E=E_-(\ell)\}$. 

Given $\ell\in [0,2\sqrt{3}M]$, we denote by $r_E$ the radial value where orbits with angular momentum $\ell$ and particle energy $E\in(0,1)$ turn, i.e. $p^r=0$. In other words, the real value $r_E\in(2M,\infty)$ satisfies that $E^2=V_{\ell}(r_E)$ for every $E\in (0,1)$. The map $E\mapsto r_E$ is well-defined since $\frac{\d}{\d r}V_{\ell}\geq 0$ when $\ell\in [0,2\sqrt{3}M]$. 

Let us define the function $\Psi_E\colon [2M,\infty)\to \R$ given by $$\Psi_E(r):=r^3\Big(V_{\ell}(r_E)-V_{\ell}(r)\Big).$$ By definition, we have $\Psi_E(r_E)=0$ for every $E\in (0,1)$. Furthermore, we have that $\Psi_E(r)=r^3(p^r)^2$ by the mass-shell relation. 

\begin{proposition}\label{prop_unique_energy _level}
Let $\ell\in(4\sqrt{\frac{2}{3}}M,2\sqrt{3}M]$. The unique continuous function $E_-\colon (4\sqrt{\frac{2}{3}}M,2\sqrt{3}M] \to [0,\infty)$ such that $r_{E_-(\ell)}$ is a double root of $\Psi_{E_-(\ell)}(r)$ and $E_-(2\sqrt{3}M)=\frac{2\sqrt{2}}{3}$, is given by $$E_-(\ell)=\Big(1-\frac{2M}{r_-^3(\ell)}\Big(r_-^2(\ell)-\frac{\ell^2}{2M}r_-(\ell)+\ell^2\Big)\Big)^{\frac{1}{2}},\qquad 
r_-(\ell):=\frac{3\ell^2}{8M}\Big(1 + \sqrt{1-\frac{32M^2}{3\ell^2}}\Big).$$ Moreover, the conserved quantity $H_{\ell}(E):=\frac{2M}{1-E^2}-\frac{2M}{1-E^2_-(\ell)}$ satisfies 
\begin{equation}\label{id_cons_quant_convex_below_isco}
H_{\ell}(E)=\frac{r^3}{r^2-\frac{\ell^2}{2M}r+\ell^2}\Big[\frac{(p^r)^2}{1-E^2}+\frac{r-r_-}{r^3}\Big((r-r_-)^2+\frac{r_-^3}{r_-^2-\frac{\ell^2}{2M}r_-+\ell^2}\Big(\frac{\ell^2}{2M}-r_-\Big)\Big)\Big]. 
\end{equation}
\end{proposition}

\begin{proof}
By the mass-shell relation, we have 
\begin{equation*}\label{mass_shell_forturn_unique_energy}
\frac{(p^r)^2}{E^2-1}-1=\frac{2M}{E^2-1}\cdot\frac{1}{r}-\frac{\ell^2}{E^2-1}\cdot\frac{r-2M}{r^3}=\frac{2M}{E^2-1}\cdot\frac{r^2-\frac{\ell^2}{2M}r+\ell^2}{r^3}.
\end{equation*}
So, the radial value $r_{E}$, where orbits with particle energy $E\in(0,1)$ turn, satisfies 
\begin{equation}\label{def_turning_pt_lema_uniq}
E^2=1-\frac{2M}{r_{E}^3}\Big(r_{E}^2-\frac{\ell^2}{2M}r_{E}+\ell^2\Big).
\end{equation}
Let us consider the cubic polynomial $h_E(r)$ given by $$h_E(r):=-\frac{1}{2M}\frac{r^3_E}{r^2_E-\frac{\ell^2}{2M}r_E+\ell^2}\Psi_E(r)=r^3-\frac{r^3_E}{r^2_E-\frac{\ell^2}{2M}r_E+\ell^2}\Big(r^2-\frac{\ell^2}{2M}r+\ell^2\Big).$$ We write $h_E(r)$ as a polynomial in $r-r_{E}$ by
\begin{align}
h_E(r)&=\frac{1}{r_{E}^2-\frac{\ell^2}{2M}r_{E}+\ell^2}\Big(\Big(r_{E}^2-\frac{\ell^2}{2M}r_{E}+\ell^2\Big)r^3-r_{E}^3\Big(r^2-\frac{\ell^2}{2M}r+\ell^2\Big)\Big)\nonumber\\
&=\frac{1}{r_{E}^2-\frac{\ell^2}{2M}r_{E}+\ell^2}\Big(\Big(r_{E}^2-\frac{\ell^2}{2M}r_{E}+\ell^2\Big)(r^3-r_{E}^3)+r_{E}^3(r_{E}-r)\Big(r-\frac{\ell^2}{2M}+r_{E}\Big)\Big)\nonumber\\
&=(r-r_{E})\Big((r-r_{E})^2+3rr_{E}-\frac{r_{E}^3}{r_{E}^2-\frac{\ell^2}{2M}r_{E}+\ell^2}\Big(r-\frac{\ell^2}{2M}+r_{E}\Big)\Big).\label{linear_term_cancel}
\end{align}
The condition $\frac{\mathrm{d}^2}{\mathrm{d} r^2}\Psi_{E_-(\ell)}(r_{E_-(\ell)})=0$ holds if and only if $h''_{E_-(\ell)}(r_{E_-}(\ell))=0$. The latter property is satisfied when the linear term in the second factor of \eqref{linear_term_cancel} vanishes. In other words, we require that 
\begin{equation}\label{eqn_solrlforsmallell}
r_{E_-(\ell)}^2-\frac{3\ell^2}{4M}r_{E_-(\ell)}+\frac{3\ell^2}{2}=0.
\end{equation}
for every $\ell\in (4\sqrt{\frac{2}{3}}M,2\sqrt{3}M]$. By the condition $E_-(2\sqrt{3}M)=\frac{2\sqrt{2}}{3}$, we note that $r_{E_-(2\sqrt{3}M)}=6M$. Thus, we choose the regular solution of \eqref{eqn_solrlforsmallell} given by $$r_-(\ell):=\frac{3\ell^2}{8M}\Big(1 + \sqrt{1-\frac{32M^2}{3\ell^2}}\Big).$$ The unique continuous function $E_-(\ell)$ satisfying the conditions in the statement is obtained by plugging $r_-(\ell)$ in \eqref{def_turning_pt_lema_uniq}. Finally, the mass-shell relation \eqref{radial_geodesic_eqn_rewritten_section_radial_coordinates} can be written in terms of $H_{\ell}(E)$ as
\begin{equation*}
H_{\ell}(E)=\frac{r^3}{r^2-\frac{\ell^2}{2M}r+\ell^2}\Big[\frac{(p^r)^2}{1-E^2}+\frac{r-r_-(\ell)}{r^3}\Big((r-r_-(\ell))^2+\frac{r_-^3(\ell)}{r_-^2(\ell)-\frac{\ell^2}{2M}r_-(\ell)+\ell^2}\Big(\frac{\ell^2}{2M}-r_-(\ell)\Big)\Big)\Big].
\end{equation*}
\end{proof}
For our purposes, we will show that the coefficient $b(\ell)\in\R$ in the last parenthesis of \eqref{id_cons_quant_convex_below_isco}, given by $$b(\ell):=\frac{r_-^3(\ell)}{r_-^2(\ell)-\frac{\ell^2}{2M}r_-(\ell)+\ell^2}\Big(\frac{\ell^2}{2M}-r_-(\ell)\Big),$$ is positive for every $\ell\sim 2\sqrt{3}M$ with $\ell \neq 2\sqrt{3}M$. We begin proving an elementary lemma about the derivatives of the function $r_-(\ell)$.

\begin{lemma}\label{lemma_derivative_rlminus}
The first and second order derivatives of $r_-: (4\sqrt{\frac{2}{3}}M,2\sqrt{3}M]\to (2M,\infty)$ are given by
\begin{align*}
\frac{\d r_-}{\d \ell}&=\frac{4M}{3\ell^3}\frac{1}{\sqrt{1-\frac{32M^2}{3\ell^2}}}r^2_-(\ell),\\
\frac{\mathrm{d}^2r_-}{\mathrm{d}\ell^2}&=-\frac{4M}{\ell^4}\frac{1}{\sqrt{1-\frac{32M^2}{3\ell^2}}}r^2_-(\ell)-\frac{128M^3}{9\ell^6}\frac{1}{(1-\frac{32M^2}{3\ell^2})^{\frac{3}{2}}}r^2_-(\ell)+\frac{32M^2}{9\ell^6}\frac{1}{1-\frac{32M^2}{3\ell^2}}r^3_-(\ell).
\end{align*}
\end{lemma}

\begin{proof}
We first compute the derivative
\begin{align*}
\frac{\d r_-}{\d \ell}&=\frac{3\ell}{4M}\frac{1}{\sqrt{1-\frac{32M^2}{3\ell^2}}}\Big(\sqrt{1-\frac{32M^2}{3\ell^2}}+\Big(1-\frac{32M^2}{3\ell^2}\Big)+\frac{16M^2}{3\ell^2}\Big)\\
&=\frac{3\ell}{8M}\frac{1}{\sqrt{1-\frac{32M^2}{3\ell^2}}}\Big(1+\sqrt{1-\frac{32M^2}{3\ell^2}}\Big)^2=\frac{4M}{3\ell^3}\frac{1}{\sqrt{1-\frac{32M^2}{3\ell^2}}}r^2_-(\ell),
\end{align*}
where we used in the last equality the definition of $r_-(\ell)$. Next, we compute the second order derivative
\begin{align*}
\frac{\mathrm{d}^2r_-}{\mathrm{d}\ell^2}
&=-\frac{4M}{\ell^4}\frac{1}{\sqrt{1-\frac{32M^2}{3\ell^2}}}r^2_-(\ell)-\frac{128M^3}{9\ell^6}\frac{1}{(1-\frac{32M^2}{3\ell^2})^{\frac{3}{2}}}r^2_-(\ell)+\frac{4M}{3\ell^3}\frac{1}{\sqrt{1-\frac{32M^2}{3\ell^2}}}2r_-(\ell)\frac{\mathrm{d}r_-}{\mathrm{d}\ell}\\
&=-\frac{4M}{\ell^4}\frac{1}{\sqrt{1-\frac{32M^2}{3\ell^2}}}r^2_-(\ell)-\frac{128M^3}{9\ell^6}\frac{1}{(1-\frac{32M^2}{3\ell^2})^{\frac{3}{2}}}r^2_-(\ell)+\frac{32M^2}{9\ell^6}\frac{1}{1-\frac{32M^2}{3\ell^2}}r^3_-(\ell).
\end{align*}
\end{proof}

Let us set the function $c \colon  (2\sqrt{2}M,2\sqrt{3}M] \to\R$ by $$c(\ell):=\frac{\ell^2}{2M}-r_-(\ell).$$ By definition of $b(\ell)$, we have that $$b(\ell)= \frac{c(\ell) r_-^3(\ell)}{r_-^2(\ell)-\frac{\ell^2}{2M}r_-(\ell)+\ell^2}.$$ As a result, the coefficient $b(\ell)$ is positive for $\ell\sim 2\sqrt{3}M$ with $\ell \neq 2\sqrt{3}M$, if and only if $c(\ell)$ is positive for $\ell\sim 2\sqrt{3}M$ with $\ell \neq 2\sqrt{3}M$. We observe that $c(2\sqrt{3}M)=0$.

\begin{lemma}
The function $c(\ell)$ has a local minimum at $\ell=2\sqrt{3}M$. In particular, the functions $c(\ell)$ and $b(\ell)$ are non-negative in an open interval whose right end point is open at $\ell=2\sqrt{3}M$.
\end{lemma}

\begin{proof}
By definition of $r_-(\ell)$, we have $c(2\sqrt{3}M)=0$. By computing the first and second order derivatives of $c(\ell)$, we have $\frac{\d c}{\d \ell}(2\sqrt{3}M)=0$ and $\frac{\d^2 c}{\d \ell^2}(2\sqrt{3}M)=\frac{8}{M}>0$. Thus, the function $c(\ell)$ has a local minimum at $\ell=2\sqrt{3}M$. In particular, we have $c(\ell)\geq 0$ in an open interval whose right end point is open at $\ell=2\sqrt{3}M$. Furthermore, we note that $$\frac{r_-^3(\ell)}{r_-^2(\ell)-\frac{\ell^2}{2M}r_-(\ell)+\ell^2}\Big|_{\ell=2\sqrt{3}M}=18M>0.$$ As a result, the function $b(\ell)$ is non-negative in an open interval whose right end point is open at $\ell=2\sqrt{3}M$, since $$b(\ell)=\frac{r_-^3(\ell)}{r_-^2(\ell)-\frac{\ell^2}{2M}r_-(\ell)+\ell^2}c(\ell)>0.$$ 
\end{proof}

We can finally show uniform concentration estimates as $\ell\to 2\sqrt{3}M$.

\begin{proposition}\label{prop_conc_est_local_isco_below}
For every geodesic $\gamma_{x,p}\colon [0,a]\to \{5M<r<7M\}$ with angular momentum $\ell\in (2\sqrt{2}M,2\sqrt{3}M)$ and particle energy $E\in (\frac{2}{3},\frac{21\sqrt{2}}{30})$, we have
$$\Big|\frac{2M}{1-E^2}-\frac{2M}{1-E^2_-(\ell)}\Big|\lesssim \frac{1}{v^6(s)},\qquad  \forall s\in [0,a].$$ Moreover, for all $s\in [0,a]$, we have 
\begin{align}
|p^r(s)|&\lesssim \frac{1}{v^{3}(s)},\quad \text{if} \quad r(s)\geq r_-(\ell),\\
\Big|\frac{p^r}{\sqrt{1-E^2}}-\frac{(r_--r)^{\frac{1}{2}}}{r^{\frac{3}{2}}}\Big((r-r_-)^2+\frac{r_-^3}{r_-^2-\frac{\ell^2}{2M}r_-+\ell^2}\Big(\frac{\ell^2}{2M}-r_-\Big)\Big)^{\frac{1}{2}}\Big|&\lesssim \frac{1}{v^{3}(s)},\quad \text{if} \quad r(s)\leq r_-(\ell).
\end{align}
The same decay estimates hold with respect to the time coordinate $u$.
\end{proposition}

\begin{proof}
By Proposition \ref{prop_unique_energy _level}, the mass-shell relation \eqref{radial_geodesic_eqn_rewritten_section_radial_coordinates} can be written as 
\begin{equation}\label{int_lowlow_good_geod_eqn}
H_{\ell}=\omega_{\ell}(r)\Big[\frac{(p^r)^2}{1-E^2}+\frac{r-r_-}{r^3}\Big((r-r_-)^2+\frac{r_-^3}{r_-^2-\frac{\ell^2}{2M}r_-+\ell^2}\Big(\frac{\ell^2}{2M}-r_-\Big)\Big)\Big],
\end{equation}
where $\omega_{\ell}(r)=r^3(r^2-\frac{\ell^2}{2M}r+\ell^2)^{-1}$. For simplicity, we write below $H$ and $\omega$, for $H_{\ell}$ and $\omega_{\ell}$, respectively. By the form of the radial potential $V_{\ell}$, we define $r_H$ as the unique real value such that $$H=\frac{r_H-r_-}{r^2_H-\frac{\ell^2}{2M}r_H+\ell^2}\Big((r_H-r_-)^2+\frac{r_-^3}{r_-^2-\frac{\ell^2}{2M}r_-+\ell^2}\Big(\frac{\ell^2}{2M}-r_-\Big)\Big).$$ The radial momentum coordinate $p^r$ vanishes at $r=r_H$ by \eqref{int_lowlow_good_geod_eqn}. Moreover, one can show that $r_H \in \{5M<r<7M\}$.  

We suppose first that $\gamma$ is an outgoing geodesic. We will consider three different cases: geodesics in the region $\{r \geq r_-\}$, in the region $\{r\leq r_-, \, H>0\}$, or in the region $\{r\leq r_-,\, H<0\}$.

\emph{The region $\{ r \geq r_-(\ell)\}$.} We first note that orbits in $\{r\geq r_-\}$ satisfy $H>0$ by definition of $H$. Without loss of generality, we consider the case when $r(0)=r_-(\ell)$. Integrating the geodesic equation \eqref{int_lowlow_good_geod_eqn}, we have
\begin{align*}
\bar{t}&=\int_{r_-(\ell)}^{r_H}\frac{\mathrm{d} r}{(\frac{H}{\omega}-\frac{1}{r^3}(r-r_-)^3-\frac{r-r_-}{r^3}(\frac{\ell^2}{2M}-r_-)\frac{r_-^3}{r_-^2-\frac{\ell^2}{2M}r_-+\ell^2})^{\frac{1}{2}}},
\end{align*}
since the geodesic $\gamma$ is outgoing. Let us set the function $\tilde{r}_{\ell}\colon [2M,\infty)\to \R$ given by 
\begin{equation}\label{int_radial_normaliz_convex}
\tilde{r}_{\ell}(r):=r_-(\ell)+\dfrac{1}{b^{\frac{1}{2}}(\ell)}(r-r_-(\ell)),\quad b(\ell)=\Big(\frac{\ell^2}{2M}-r_-\Big)\frac{r_-^3}{r_-^2-\frac{\ell^2}{2M}r_-+\ell^2}.
\end{equation} The function $\tilde{r}_{\ell}(r)$ is well-defined on $[2M,\infty)$ since $\ell\in (2\sqrt{2}M,2\sqrt{3}M)$. In the rest of the proof, we denote $\tilde{r}_{\ell}$ simply as $\tilde{r}$. We note that $\frac{\d \tilde{r}_{\ell}}{\d r}=b^{-\frac{1}{2}}(\ell)>0$. By using the change of variables $r\mapsto \tilde{r}_{\ell}(r)$, we have 
\begin{align}
\bar{t}&=\int_{r_-(\ell)}^{r_H}\frac{\mathrm{d} r}{(\frac{H}{\omega}-\frac{r-r_-}{r^3}((r-r_-)^2+b(\ell)))^{\frac{1}{2}}}\nonumber\\
&=b^{\frac{1}{2}}(\ell)\int_{r_-(\ell)}^{\tilde{r}(r_H)}\Big(\frac{H}{\omega(\tilde{r})}-\frac{b^{\frac{3}{2}}(\ell) (\tilde{r}-r_-)}{(r_-+b^{\frac{1}{2}}(\ell)(\tilde{r}-r_-))^3}\Big((\tilde{r}-r_-)^2+1\Big)\Big)^{-\frac{1}{2}}\mathrm{d} \tilde{r}\nonumber\\
&=b^{-\frac{1}{4}}(\ell)\int_{r_-(\ell)}^{\tilde{r}(r_H)}\Big(\frac{H}{b^{\frac{3}{2}}(\ell)}-\frac{ (\tilde{r}-r_-)}{(r_-+b^{\frac{1}{2}}(\ell)(\tilde{r}-r_-))^3}((\tilde{r}-r_-)^2+1)\omega\Big)^{-\frac{1}{2}}\omega^{\frac{1}{2}}\mathrm{d} \tilde{r}\nonumber\\
&\lesssim b^{-\frac{1}{4}}(\ell),\label{key_int_est_convex}
\end{align}
where the last inequality holds for $|H|\lesssim b^{\frac{3}{2}}(\ell)$. Thus, we have $\bar{t}\lesssim b^{-\frac{1}{4}}(\ell)\lesssim H^{-\frac{1}{6}}$, so we obtain $\sqrt{H}\lesssim v^{-3}$. We have used that $\bar{t}\gtrsim v(s)$ since $\gamma([0,a])\subset \{5M<r<7M\}$. We conclude this case by $$|p^r(s)|\lesssim \sqrt{H}\lesssim v^{-3}(s),$$ by the positivity of the second term in the definition of $H$ in \eqref{int_lowlow_good_geod_eqn}.

\emph{The region $\{r\leq r_-(\ell),\, H>0\}$.} Without loss of generality, we consider the case when $r(0)=5M$. Integrating the geodesic equation \eqref{int_lowlow_good_geod_eqn}, we have
\begin{align*}
\bar{t}&=\int_{5M}^{r_H}\frac{\mathrm{d} r}{(\frac{H}{\omega}-\frac{1}{r^3}(r-r_-)^3-\frac{r-r_-}{r^3}(\frac{\ell^2}{2M}-r_-)\frac{r_-^3}{r_-^2-\frac{\ell^2}{2M}r_-+\ell^2})^{\frac{1}{2}}}.
\end{align*}
Furthermore, we have $$\bar{t}\leq \int_{5M}^{r_H}\frac{\mathrm{d} r}{(\frac{H}{\omega}+\frac{1}{r^3}(r_--r)^3)^{\frac{1}{2}}},$$ since $r_-(\ell)\geq r$. By the estimate \eqref{estimate_isco_integral_bound} in Proposition \ref{prop_isco_nearatstablemfl}, we obtain $\sqrt{H}\lesssim v^{-3}$, where we have used that $\bar{t}\gtrsim v(s)$ since $\gamma([0,a])\subset \{5M<r<7M\}$. We conclude this case since $$\Big| \frac{p^r}{\sqrt{1-E^2}}-\frac{\sqrt{r_--r}}{r^{\frac{3}{2}}}\Big((r-r_-)^2+\frac{r_-^3}{r_-^2-\frac{\ell^2}{2M}r_-+\ell^2}\Big(\frac{\ell^2}{2M}-r_-\Big)\Big)^{\frac{1}{2}} \Big|\lesssim \sqrt{H}\lesssim v^{-3}(s),$$ where we have used a lower bound for $\omega_{\ell}$, and the $\frac{1}{2}$-Hölder continuity of the square root.

\emph{The region $\{r\leq r_-(\ell),\, H<0\}$.} Without loss of generality, we consider the case when $r(0)=5M$. Integrating the geodesic equation, we have
\begin{align*}
\bar{t}&=\int_{5M}^{r_H}\frac{\mathrm{d} r}{(-\frac{(-H)}{\omega}-\frac{1}{r^3}(r-r_-)^3-\frac{r-r_-}{r^3}(\frac{\ell^2}{2M}-r_-)\frac{r_-^3}{r_-^2-\frac{\ell^2}{2M}r_-+\ell^2})^{\frac{1}{2}}}.
\end{align*}
Furthermore, we have $$\bar{t}\leq \int_{5M}^{r_H}\frac{\mathrm{d} r}{(-\frac{(-H)}{\omega}+\frac{1}{r^3}(r_--r)^3)^{\frac{1}{2}}},$$ since $r_-(\ell)\geq r$. By the estimate \eqref{estimate_isco_integral_bound} in Proposition \ref{prop_isco_nearatstablemfl}, we obtain $\sqrt{-H}\lesssim v^{-3}$, where we have used that $\bar{t}\gtrsim v(s)$ since $\gamma([0,a])\subset \{5M<r<7M\}$. We conclude this case since $$\Big| \frac{p^r}{\sqrt{1-E^2}}-\frac{\sqrt{r_--r}}{r^{\frac{3}{2}}}\Big((r-r_-)^2+\frac{r_-^3}{r_-^2-\frac{\ell^2}{2M}r_-+\ell^2}\Big(\frac{\ell^2}{2M}-r_-\Big)\Big)^{\frac{1}{2}} \Big|\lesssim \sqrt{-H}\lesssim v^{-3}(s),$$ where we have used a lower bound for $\omega$, and the $\frac{1}{2}$-Hölder continuity of the square root.

Analogous estimates hold in the regime where $\gamma$ is ingoing. For geodesics with a turning point at $r_H$, we put together the estimates when $\gamma$ is outgoing and ingoing in the regions considered above. The same decay estimates hold with respect to the time coordinate $u$. For this, we use that $s\gtrsim u(s)$ since $\gamma([0,a])\subset \{r>5M\}$. 
\end{proof}

\begin{remark}
The decay rate $t^{-3}$ in Proposition \ref{proposition_decay_estimate_isco} holds, because of the cubic potential $(1-\frac{6M}{r})^3$ in the geodesic equation \eqref{idHamiltisco}. More generally, for a radial potential of the form $(1-\frac{r_c}{r})^n$ with $r_c>0$, one can show concentration estimates with the decay rate $t^{-\frac{n}{n-2}}$. 
\end{remark}

\subsection{The region $\D_{\mathrm{bd}}$}

In this region, the radial potential $V_{\ell}(r)$ has critical points at $r=r_{\pm}(\ell)$, where $V_{\ell}(r)$ has a local maximum at $r_{-}(\ell)$, and a local minimum at $r_{+}(\ell)$. Moreover, the potential $V_{\ell}(r)$ has a maximum at infinity where $\lim_{r\to \infty}V_{\ell}(r)=1$. We remark that $V_{\ell}(r_-(\ell))<1$, so the fixed point $(r=r_-(\ell),\, p^r=0)$ of the radial flow is a homoclinic point. One can easily show that there are three types of trapping in $\D_{\mathrm{bd}}$: parabolic trapping at infinity, degenerate trapping at ISCO, and unstable trapping.

\subsubsection{The subregion $E\sim 1$}

Let us study concentration estimates on the parabolic manifolds at the energy level $\{E=1\}$. Parabolic trapping at infinity holds for all $\ell\in [2\sqrt{3}M,4M]$ in $\D_{\mathrm{bd}}$. We will estimate the radial flow in a uniform neighbourhood of the parabolic manifolds.

\begin{proposition}\label{proposition_slow_decay_geodesics_medium_angular_momentum_particle_one}
Let $R>r_0>16M$. For every geodesic $\gamma_{x,p}\colon [0,a]\to \{r>r_0\}$ with angular momentum $\ell\in (2\sqrt{3}M,4M)$, particle energy $E\in [E_-(\ell),1)$, and initial data $(x,p)\in\{16M<r<R\}$, we have $$|1-E^2|\lesssim \dfrac{1}{v^{\frac{2}{3}}(s)},\qquad \forall s\in[0,a].$$ 
Moreover, for all $s\in [0,a],$ we have
\begin{align*}
\Big|\dfrac{\sqrt{2M}}{r^{\frac{3}{2}}}\Big(r^2-\dfrac{\ell^2}{2M}r+\ell^2\Big)^{\frac{1}{2}}-p^r\Big|\lesssim \dfrac{1}{v^{\frac{1}{3}}(s)},\qquad \text{if}\quad p^r\geq 0,\\
\Big|\dfrac{\sqrt{2M}}{r^{\frac{3}{2}}}\Big(r^2-\dfrac{\ell^2}{2M}r+\ell^2\Big)^{\frac{1}{2}}+p^r\Big|\lesssim \dfrac{1}{v^{\frac{1}{3}}(s)},\qquad \text{if}\quad p^r\leq 0.
\end{align*}
The same decay estimates hold with respect to the time coordinate $u$.
\end{proposition}

\begin{proof}
The proof of this proposition is similar to the proof of Proposition \ref{proposition_slow_decay_particle_energy_one_include_radial_geodesics}. Here, we consider $R>r_0>16M$ instead, so that we obtain uniform bounds for $\frac{\mathrm{d}\tilde{r}_{\ell}}{\mathrm{d}r}$ after performing the change of variables $r\mapsto \tilde{r}_{\ell}(r)$. Moreover, we consider geodesics with initial data $(x,p)\in\{16M<r<R\}$ to avoid making any reference to the flow near $\{r=r_-(\ell)\}$ for $\ell\sim 4M$. The rest of the proof follows without further changes.
\end{proof}

\subsubsection{The subregion $E\sim \frac{2\sqrt{2}}{3}$ with $\ell\sim 2\sqrt{3}M$}

Let us study concentration estimates on the stable manifolds associated to the degenerate trapping at ISCO in the region $\D_{\mathrm{bd}}$. We will estimate the radial flow in a uniform neighbourhood of the sphere $\mathcal{S}_-(2\sqrt{3}M)$. The set $\{\ell =2\sqrt{3}M\}$ is located at the boundary of $\D_{\mathrm{bd}}$. 

\begin{proposition}\label{proposition_decay_estimate_isco}
For every geodesic $\gamma_{x,p}\colon [0,a]\to \{5M<r<7M\}$ with angular momentum $\ell\in (2\sqrt{3}M, 4M)$ and particle energy $E\in (\frac{2\sqrt{2}}{3},\frac{21\sqrt{2}}{30})$, we have
$$\Big|\frac{2M}{1-E^2}-\frac{2M}{1-E^2_-(\ell)}\Big|\lesssim \frac{1}{v^6(s)},\qquad \forall s\in [0,a].$$ Moreover, for all $s\in [0,a]$, we have
\begin{align*}
|p^r(s)|&\lesssim v^{-3}(s),\qquad \text{if}\quad r\in[-a(\ell),\infty),\\
\Big|\frac{p^r}{\sqrt{1-E^2}}-\Big(-1-\frac{a(\ell)}{r}\Big)^{\frac{1}{2}}\Big(1-\frac{r_-(\ell)}{r}\Big)\Big|&\lesssim  v^{-3}(s) ,\qquad \text{if} \quad p^r\geq 0 \quad \text{ and } \quad r\in [r_-(\ell) ,-a(\ell)],\\
\Big|\frac{p^r}{\sqrt{1-E^2}}+\Big(-1-\frac{a(\ell)}{r}\Big)^{\frac{1}{2}}\Big(1-\frac{r_-(\ell)}{r}\Big)\Big|&\lesssim v^{-3}(s), \qquad \text{if} \quad p^r\leq 0 \quad \text{ and } \quad r\in (5M,-a(\ell)].
\end{align*}
The same decay estimates hold with respect to the time coordinate $u$.
\end{proposition}

\begin{proof}
By Proposition \ref{prop_hamilt_exp_contr_main_hyp}, the mass-shell relation \eqref{radial_geodesic_eqn_rewritten_section_radial_coordinates} can be written as 
\begin{equation}\label{int_unst_good_geod_eqn_est}
H_{\ell}(E):=\frac{2M}{1-E^2}-\frac{2M}{1-E^2_-(\ell)}=\omega_{\ell}(r)\Big(\frac{(p^r)^2}{1-E^2}+\Big(1+\frac{a(\ell)}{r}\Big)\Big(1-\frac{r_-(\ell)}{r}\Big)^2\Big),
\end{equation}
where $\omega_{\ell}(r):=r^3(r^2-\frac{\ell^2}{2M}r+\ell^2)^{-1}$. For simplicity, we write below $H$ and $\omega$, for $H_{\ell}$ and $\omega_{\ell}$, respectively. By the form of the radial potential $V_{\ell}$, we define $r_H$ as the unique real value such that $$H=\frac{(r_H+a(\ell))(r_H-r_-(\ell))}{r^2_H-\frac{\ell^2}{2M}r_H+\ell^2}.$$ The radial momentum coordinate $p^r$ vanishes at $r=r_H$ by \eqref{int_unst_good_geod_eqn_est}. Moreover, one can show that $r_H\in \{5M<r<7M\}$. 

We suppose first that $\gamma$ is an outgoing geodesic. We will consider three different cases: geodesics in the region $\{r\geq r_-(\ell)\}$, in the region $\{r\leq r_-(\ell),\, H>0\}$, or in the region $\{r\leq r_-(\ell),\, H<0\}$.

\emph{The region $\{r\geq r_-(\ell)\}$.} We first note that orbits in $\{r\geq r_-(\ell)\}\cap \D_{\mathrm{bd}}$ satisfy $H\geq 0$. In general, orbits in $\{r\geq r_-(\ell)\}$ can also have $H<0$, however, this only occurs in the region $\P\setminus \D$ of the mass-shell. Without loss of generality, we consider the case where $r(0)=r_-(\ell)$. Integrating the geodesic equation \eqref{int_unst_good_geod_eqn_est}, we have
\begin{align}
\bar{t}&=\int_{r_-(\ell)}^{r_H}\frac{\d r}{(\frac{H}{\omega}-\frac{1}{r^3}(r+a(\ell))(r-r_-)^2)^{\frac{1}{2}}},
\end{align} 
since the geodesic $\gamma$ is outgoing.  We note that 
\begin{equation}\label{ident_making_isco_relevant}
(r+a(\ell))(r-r_-(\ell))^2=(r-r_-(\ell))^3-(-a(\ell)-r_-(\ell))(r-r_-(\ell))^2,
\end{equation}
where $-a(\ell)-r_-(\ell)>0$ for $\ell\in (2\sqrt{3}M, 2\sqrt{3}M+\frac{1}{2}M]$. Thus, we can bound 
\begin{align*}
\bar{t}&=\int_{r_-(\ell)}^{r_H}\frac{\mathrm{d} r}{(\frac{H}{\omega}-\frac{1}{r^3}(r-r_-(\ell))^3+\frac{1}{r^3}(-a(\ell)-r_-(\ell))(r-r_-(\ell))^2)^{\frac{1}{2}}}\\
&\leq \int_{r_-(\ell)}^{r_H}\frac{\mathrm{d} r}{(\frac{H}{\omega}-\frac{1}{r^3}(r-r_-(\ell))^3)^{\frac{1}{2}}}.
\end{align*}
By using the estimate \eqref{estimate_isco_integral_bound} performed in Proposition \ref{prop_isco_nearatstablemfl}, we obtain that $\sqrt{H}\lesssim v^{-3}$. We have used that $\bar{t}\gtrsim v(s)$ since $\gamma([0,a])\subset \{5M<r<7M\}$. 
Thus, for every $r\geq -a(\ell)$, we have $$|p^r(s)|\lesssim \sqrt{H}\lesssim v^{-3}(s),$$ by the positivity of the second term in the definition of $H$ in \eqref{int_unst_good_geod_eqn_est}. Otherwise, if $r\leq -a(\ell)$, then 
\begin{align*}
\Big|\frac{p^r}{\sqrt{1-E^2}}-\Big(-1-\frac{a(\ell)}{r}\Big)^{\frac{1}{2}}\Big(1-\frac{r_-(\ell)}{r}\Big)\Big|&\lesssim \sqrt{H}\lesssim v^{-3}(s) ,\qquad \text{if} \quad p^r>0,\\
\Big|\frac{p^r}{\sqrt{1-E^2}}+\Big(-1-\frac{a(\ell)}{r}\Big)^{\frac{1}{2}}\Big(1-\frac{r_-(\ell)}{r}\Big)\Big|&\lesssim \sqrt{H}\lesssim v^{-3}(s), \qquad \text{if} \quad p^r<0,
\end{align*}
where we have used a lower bound for $\omega$, and the $\frac{1}{2}$-Hölder continuity of the square root.

\emph{The region $\{r\leq r_-(\ell),\, H>0\}$.}  
Without loss of generality, we consider the case when $r(a)=r_-(\ell)$. Integrating the geodesic equation \eqref{int_unst_good_geod_eqn_est} and using \eqref{ident_making_isco_relevant}, we have 
\begin{align*}
\bar{t}&\leq \int_{5M}^{r_-(\ell)}\frac{\mathrm{d} r}{(\frac{H}{\omega}-\frac{1}{r^3}(r-r_-(\ell))^3)^{\frac{1}{2}}}.
\end{align*} 
By using the estimate \eqref{estimate_isco_integral_bound} performed in Proposition \ref{prop_isco_nearatstablemfl}, we obtain that $\sqrt{H}\lesssim v^{-3}$. We have used that $\bar{t}\gtrsim v(s)$ since $\gamma([0,a])\subset \{5M<r<7M\}$. We conclude this case since $$\Big|\frac{p^r}{\sqrt{1-E^2}}+\Big(-1-\frac{a(\ell)}{r}\Big)^{\frac{1}{2}}\Big(1-\frac{r_-(\ell)}{r}\Big)\Big|\lesssim \sqrt{H}\lesssim v^{-3}(s),$$ where we have used a lower bound for $\omega_{\ell}$, and the $\frac{1}{2}$-Hölder continuity of the square root. 

\emph{The region $\{r\leq r_-(\ell),\, H<0\}$.} 
Without loss of generality, we consider the case when $r(a)=r_-(\ell)$. Integrating the geodesic equation \eqref{int_unst_good_geod_eqn_est} and using \eqref{ident_making_isco_relevant}, we have 
\begin{align*}
\bar{t}&\leq \int_{5M}^{r_-(\ell)}\frac{\mathrm{d} r}{(-\frac{(-H)}{\omega}-\frac{1}{r^3}(r-r_-(\ell))^3)^{\frac{1}{2}}}.
\end{align*} 
By using the estimate \eqref{estimate_isco_integral_bound} performed in Proposition \ref{prop_isco_nearatstablemfl}, we obtain that $\sqrt{-H}\lesssim v^{-3}$. We have used that $\bar{t}\gtrsim v(s)$ since $\gamma([0,a])\subset \{5M<r<7M\}$. We conclude this case since $$\Big|\frac{p^r}{\sqrt{1-E^2}}+\Big(-1-\frac{a(\ell)}{r}\Big)^{\frac{1}{2}}\Big(1-\frac{r_-(\ell)}{r}\Big)\Big|\lesssim \sqrt{-H}\lesssim v^{-3}(s),$$ where we have used a lower bound for $\omega$, and the $\frac{1}{2}$-Hölder continuity of the square root. 

Analogous estimates hold in the regime where $\gamma$ is ingoing. For geodesics with a turning point at $r=r_H$, we put together the estimates when $\gamma$ is outgoing and ingoing. The same decay estimates hold with respect to  the time coordinate $u$. For this, we use that $s\gtrsim u(s)$ since $\gamma([0,a])\subset \{r>5M\}$. 
\end{proof}

\subsubsection{The subregion $E\sim E_-(\ell)$}

Let us study concentration estimates on the stable manifolds associated to the unstable trapping at the energy level $\{E(x,p)=E_-(\ell)\}$ for $\ell\in (2\sqrt{3}M,4M)$. We will estimate the radial flow in a neighbourhood of the sphere $\mathcal{S}^-(\ell)$. 

\begin{proposition}\label{prop_decay_contr_exp_hyp_bound}
Let $\delta>0$. There exists $\epsilon\in (0,\frac{1}{2}M)$ such that for every geodesic $\gamma_{x,p}\colon [0,a]\to \{|r-r_-(\ell)|<\epsilon\}$ with angular momentum $\ell\in (2\sqrt{3}M+\epsilon,4M)$ and particle energy $E\in (E_-(\ell)-\epsilon,E_-(\ell)+\epsilon)$, we have
\begin{equation}
    |p^r(s)\,-\,p^{r,+}_{\ell}(r(s))|\lesssim \dfrac{1}{\exp((\lambda^{\frac{1}{2}}(\ell)-\delta)t(s))}, \qquad \Big|\frac{2M}{1-E^2}\,-\,\frac{2M}{1-E^2_-(\ell)}\Big|\lesssim   \dfrac{1}{\exp((\lambda^{\frac{1}{2}}(\ell)-\delta)t(s))}
\end{equation} 
for all $s\in [0,a]$.
\end{proposition}

\begin{proof}
We first note that the Lyapunov exponents of $(r_-(\ell),0)$ satisfy $$\lambda^{\frac{1}{2}}(\ell)\geq \lambda^{\frac{1}{2}}(2\sqrt{3}M+\epsilon)>0,$$ by the monotonicity property of $\lambda^{\frac{1}{2}}(\ell)$ in Proposition \ref{prop_monot_lyap_exp}. We set $\epsilon>0$ small enough so that $$ \Big|\lambda^{\frac{1}{2}}(\ell)-\frac{(1-E^2)^{\frac{1}{2}}}{E}\frac{(r_+(\ell)-r)(r-2M)}{2r^{\frac{1}{2}}(-a(\ell)-r)^{\frac{1}{2}}(r^2-\frac{\ell^2}{2M}r+\ell^2)}\Big| \leq \delta,$$ for all $(x,p)\in \{|r-r_-(\ell)|\leq  \epsilon\}\cap\{|E-E_-(\ell)|\leq  \epsilon  \}$. Next, we integrate the derivative along the geodesic flow of $\varphi_{-}^{\ell}$ by 
\begin{equation}\label{ident_int_varphimin_local_est}
\varphi_{-}^{\ell}\Big(x(t(s)),p(t(s))\Big)=\varphi_{-}^{\ell}(x,p)\exp \Big(-\frac{(1-E^2)^{\frac{1}{2}}}{E}\int_{0}^{t(s)} \frac{(r_+(\ell)-r)(r-2M)}{2r^{\frac{1}{2}}(-a(\ell)-r)^{\frac{1}{2}}(r^2-\frac{\ell^2}{2M}r+\ell^2)}\d t \Big),
\end{equation}
where $(x(0)=x,p(0)=p)$. As a result, we have the bound $$|\varphi_{-}^{\ell}|\lesssim \exp \Big(-\frac{(1-E^2)^{\frac{1}{2}}}{E}\int_{0}^{t(s)} \frac{(r_+(\ell)-r)(r-2M)}{2r^{\frac{1}{2}}(-a(\ell)-r)^{\frac{1}{2}}(r^2-\frac{\ell^2}{2M}r+\ell^2)}\d t \Big)\lesssim \exp \Big(-\Big(\lambda^{\frac{1}{2}}(\ell)-\delta\Big)t(s)\Big),$$ so the momentum coordinate $p^r$ satisfies $$|p^r(s)\,-\,p^{r,+}_{\ell}(r(s))|\lesssim \Big|\varphi_{-}^{\ell}\Big(x(t(s)),p(t(s))\Big)\Big|\lesssim \exp \Big(-\Big(\lambda^{\frac{1}{2}}(\ell)-\delta\Big)t(s)\Big).$$ Finally, we note that $\varphi_{+}^{\ell}$ is bounded in the region $\{|r-r_-(\ell)|<\epsilon\}\cap \{|E-E_-(\ell)|<\e\}$. Thus, $$\Big|\frac{2M}{1-E^2}\,-\,\frac{2M}{1-E^2_-(\ell)}\Big|=\Big|\varphi_{+}^{\ell}\varphi_{-}^{\ell}\Big(x(t(s)),p(t(s))\Big)\Big| \lesssim \exp \Big(-\Big(\lambda^{\frac{1}{2}}(\ell)-\delta\Big)t(s)\Big).$$ 
\end{proof}

Similarly, we show concentration estimates on the stable manifolds associated to the unstable trapping at the energy level $\{E=1\}$ for $\ell=4M$. In this case, we consider the defining functions $\psi_{\pm}^{4M}$ for the stable manifolds. We treat this case apart because the defining functions $\varphi_{\pm}^{\ell}$ degenerate at $E=1$.

\begin{proposition}\label{prop_hyp_unst_trapp_Eequalone}
Let $\delta>0$. There exists $\epsilon>0$ such that for every geodesic $\gamma_{x,p}\colon [0,a]\to \{|r-4M|<\epsilon\}$ with angular momentum $\ell=4M$ and particle energy $E\in (1-\epsilon,1+\epsilon)$, we have
\begin{equation}
   |p^r(s)\,-\,p^{r,+}_{\ell}(r(s))|\lesssim \dfrac{1}{\exp((\frac{1}{8\sqrt{2}M}-\delta)t(s))}, \qquad |E^2\,-\,1|\lesssim   \dfrac{1}{\exp((\frac{1}{8\sqrt{2}M}-\delta)t(s))}
\end{equation} 
for all $s\in [0,a]$.
\end{proposition}

\begin{proof}
We set $\epsilon>0$ small enough so that $$ \Big|\frac{1}{8\sqrt{2}M}-\frac{\sqrt{2M}}{2E}\frac{1}{r^{\frac{3}{2}}} \Big(1-\frac{12M}{r}\Big)\Big(1-\frac{2M}{r}\Big)\Big| \leq \delta,$$ for all $ (x,p)\in \{|r-4M|\leq  \epsilon\} \cap \{|E-E_-(\ell)|<\e\}.$ We integrate the derivative along the geodesic flow of $\psi_{+}^{4M}$ by $$\psi_{+}^{4M}\Big(x(t(s)),p(t(s))\Big)=\psi_{+}^{4M}(x,p)\exp \Big(-\frac{1}{E}\int_{0}^{t(s)} \frac{\sqrt{2M}}{2}\frac{1}{r^{\frac{3}{2}}} \Big(1-\frac{12M}{r}\Big)\Big(1-\frac{2M}{r}\Big)\d t \Big),$$ where $(x(0)=x,p(0)=p)$. As a result, we have $$|\psi_{+}^{4M}|\lesssim \exp \Big(-\frac{1}{E}\int_{0}^{t(s)} \frac{\sqrt{2M}}{2}\frac{1}{r^{\frac{3}{2}}} \Big(1-\frac{12M}{r}\Big)\Big(1-\frac{2M}{r}\Big) \d t \Big)\lesssim \exp \Big(-\Big(\frac{1}{8\sqrt{2}M}-\delta\Big)t(s)\Big),$$ so the momentum coordinate $p^r$ satisfies $$|p^r(s)\,-\,p^{r,+}_{\ell}(r(s))|\lesssim \Big|\psi_{+}^{4M}\Big(x(t(s)),p(t(s))\Big)\Big|\lesssim \exp \Big(-\Big(\frac{1}{8\sqrt{2}M}-\delta\Big)t(s)\Big).$$ Finally, we note that $\psi_{-}^{4M}$ is bounded in the region $\{|r-r_-(\ell)|<\epsilon\}\cap\{|E-E_-(\ell)|<\e\}$. Thus, $$\Big|\frac{2M}{1-E^2}\,-\,\frac{2M}{1-E^2_-(\ell)}\Big|=\Big|\psi_{+}^{4M}\psi_{-}^{4M}\Big(x(t(s)),p(t(s))\Big)\Big| \lesssim \exp \Big(-\Big(\frac{1}{8\sqrt{2}M}-\delta\Big)t(s)\Big).$$
\end{proof}

\begin{remark}
We observe that Propositions \ref{prop_decay_contr_exp_hyp_bound} and \ref{prop_hyp_unst_trapp_Eequalone}, could have been obtained by using the quantitative estimates for the geodesic flow that follow from the proof of the stable manifold theorem. See \cite[Theorem 17.4.3]{KH95} for more details about the standard proof of the stable manifold theorem, based on the \emph{graph transform} over the stable and unstable subspaces defined by the linearised flow. Instead, we have performed a proof that makes use of the explicit defining functions of the stable manifolds associated to the trapped set $\Gamma$ of Schwarzschild. This strategy gives more information about the concentration of the geodesic flow on the unstable manifold in the specific setup of Schwarzschild spacetime.
\end{remark}

For our purposes, we will require concentration estimates near the whole homoclinic orbits at the energy level $\{E=E_-(\ell)\}$. The following estimates do not only consider geodesics in a neighbourhood of the sphere $\mathcal{S}^-(\ell)$ as in Propositions \ref{prop_decay_contr_exp_hyp_bound} and \ref{prop_hyp_unst_trapp_Eequalone}. 

\begin{proposition}\label{prop_conc_estim_homocl_bounded_angu_mom}
Let $R>2M$. For every geodesic $\gamma_{x,p}\colon [0,a]\to \{5M<r<R\}$ with angular momentum $\ell\in (2\sqrt{3}M, 4M)$ and particle energy $E\in (\frac{2}{3},\frac{21\sqrt{2}}{30})$, we have
$$\Big|\frac{2M}{1-E^2}\,-\,\frac{2M}{1-E^2_-(\ell)}\Big|\lesssim \frac{1}{v^6(s)},\qquad \forall s\in [0,a].$$ Moreover, for all $s\in [0,a]$, we have
\begin{align*}
|p^r(s)|&\lesssim v^{-3}(s),\qquad \text{if}\quad r\in[-a(\ell),\infty),\\
\Big|\frac{p^r}{\sqrt{1-E^2}}\,-\,\Big(-1-\frac{a(\ell)}{r}\Big)^{\frac{1}{2}}\Big(1-\frac{r_-(\ell)}{r}\Big)\Big|&\lesssim  v^{-3}(s) ,\qquad \text{if} \quad p^r\geq 0 \quad \text{ and } \quad r\in [r_-(\ell) ,-a(\ell)],\\
\Big|\frac{p^r}{\sqrt{1-E^2}}\,+\,\Big(-1-\frac{a(\ell)}{r}\Big)^{\frac{1}{2}}\Big(1-\frac{r_-(\ell)}{r}\Big)\Big|&\lesssim v^{-3}(s), \qquad \text{if} \quad p^r\leq 0 \quad \text{ and } \quad r\in (5M,-a(\ell)].
\end{align*}
The same decay estimates hold with respect to the time coordinate $u$.
\end{proposition}

\begin{proof}
The proof of this proposition is similar to the proof of Proposition \ref{proposition_decay_estimate_isco}. We consider here geodesics with a larger range of possible angular momentum values. By the assumption that the particle energy $E$ is in $(\frac{2}{3},\frac{21\sqrt{2}}{30})$, we only consider unstable trapping in a range $(2\sqrt{3}M, L_1)$ with $L_1<4M$. Because of this property, the function $-a(\ell)\geq 2M$ satisfies the uniform bound $-a(\ell)\leq A$ for a constant $A>2M$. As a result, the corresponding homoclinic orbits are contained in the bounded region $\{-a(\ell)\leq A\}$. By these considerations, the proof can be adapted from the arguments in Proposition \ref{proposition_decay_estimate_isco}.
\end{proof}

\begin{remark}\label{remark_control_homoclinic_critical_angular_mom_bounde}
We remark that Proposition \ref{prop_conc_estim_homocl_bounded_angu_mom} also applies to geodesics with angular momentum $\ell\in (2\sqrt{3}M, 4M)$ and particle energy $E\in (\frac{21\sqrt{2}}{30},1)$ in the bounded region. We also note that this result does not apply in the far-away region, because the homoclinic orbits become larger and larger as $\ell\to 4M$, so the constant in the RHS of the estimates degenerates.
\end{remark}

\subsection{The region $\D_{\mathrm{high}}$}

In this region, the radial potential $V_{\ell}(r)$ has critical points at $r=r_{\pm}(\ell)$, where $V_{\ell}(r)$ has a local maximum at $r_-(\ell)$, and a local minimum at $r_+(\ell)$. Moreover, the potential $V_{\ell}(r)$ has a maximum at infinity where $\lim_{r\to\infty}V_{\ell}(r)=1$. We remark that $V_{\ell}(r_-(\ell))>1$, so the fixed point $(r=r_-(\ell),\, p^r=0)$ of the radial flow is \emph{not a homoclinic point}. One can easily show that there is only one type of trapping in $\D_{\mathrm{high}}$: unstable trapping.

\subsubsection{The subregion $E\sim E_-(\ell)$}

Let us study concentration estimates on the stable manifolds associated to the unstable trapping effect at the energy level $\{E(x,p)=E_-(\ell)\}$ for $\ell\in (4M,\infty)$. We will estimate the radial flow in a neighbourhood of the spheres $\mathcal{S}^-(\ell)$ for all $\ell\geq \ell_0>4M$ for a fixed $\ell_0\in\R$.

\begin{proposition}\label{lemma_application_stable_manifol_thm_large_angular_momentum}
Let $\ell_0>4M$. There exists $\e>0$ such that for every geodesic $\gamma_{x,p}\colon [0,a]\to \{ |r-r_-(\ell)|<\epsilon \}$ with angular momentum $\ell\geq \ell_0$ and particle energy $E\in (E_-(\ell)-\epsilon,E_-(\ell)+\epsilon)$, we have
\begin{equation}\label{estim_concentr_high_ang}
    |p^r(s)\,-\,p^{r,+}_{\ell}(r(s))|\lesssim \dfrac{1}{\exp( \frac{1}{8\sqrt{2}M}t(s))}, \qquad \Big|\frac{2M}{1-E^2}\,-\,\frac{2M}{1-E^2_-(\ell)}\Big|\lesssim   \dfrac{1}{\exp( \frac{1}{8\sqrt{2}M}t(s))},
\end{equation} 
for all $s\in [0,a]$.
\end{proposition}

\begin{proof}
We first note that the Lyapunov exponents of the fixed points $(r_-(\ell),0)$ satisfy 
\begin{equation}\label{est_lyap_exp_monot_high}
\lambda^{\frac{1}{2}}(\ell)\geq \lambda^{\frac{1}{2}}(\ell_0)\geq \lambda^{\frac{1}{2}}(4M)=\frac{1}{8\sqrt{2}M},
\end{equation}
by the monotonicity property of $\lambda^{\frac{1}{2}}(\ell)$ in Proposition \ref{prop_monot_lyap_exp}. We set $\epsilon>0$ small enough so that $$ \Big|\lambda^{\frac{1}{2}}(\ell)-\frac{(1-E^2)^{\frac{1}{2}}}{E}\frac{(r_+(\ell)-r)(r-2M)}{2r^{\frac{1}{2}}(-a(\ell)-r)^{\frac{1}{2}}(r^2-\frac{\ell^2}{2M}r+\ell^2)} \Big| \leq \frac{1}{2}\Big(\lambda^{\frac{1}{2}}(\ell_0)-\frac{1}{8\sqrt{2}M}\Big)=:d(\ell_0),$$ for all $(x,p)\in\{|r-r_-(\ell)|\leq  \epsilon\}\cap \{|E(x,p)-E_-(\ell)|<\e\}$. Next, we integrate as in \eqref{ident_int_varphimin_local_est} the derivative along the geodesic flow of $\varphi_{-}^{\ell}$. As a result, the momentum coordinate $p^r$ satisfies $$|p^r(s)\,-\,p^{r,+}_{\ell}(r(s))|\lesssim \Big|\varphi_{-}^{\ell}\Big(x(t(s)),p(t(s))\Big)\Big|\lesssim \exp \Big(-\Big(\lambda^{\frac{1}{2}}(\ell)-d(\ell_0)\Big)t(s)\Big)\lesssim \exp\Big(-\frac{t(s)}{8\sqrt{2}M}\Big),$$ where we have used \eqref{est_lyap_exp_monot_high} in the last estimate. Finally, we note that $\varphi_{+}^{\ell}$ is bounded above in the domain $\{|r-r_-(\ell)|<\epsilon\}\cap \{|E-E_-(\ell)|<\e\}$. Thus, $$\Big|\frac{2M}{1-E^2}\,-\,\frac{2M}{1-E^2_-(\ell)}\Big|=\Big|\varphi_{+}^{\ell}\varphi_{-}^{\ell}\Big(x(t(s)),p(t(s))\Big)\Big|\lesssim \exp \Big(-\Big(\lambda^{\frac{1}{2}}(\ell)-d(\ell_0)\Big)t(s)\Big)\lesssim \exp\Big(-\frac{t(s)}{8\sqrt{2}M}\Big),$$ where we have used again \eqref{est_lyap_exp_monot_high}.
\end{proof}

\section{Decay in time for dispersive Vlasov fields on Schwarzschild}\label{section_proofs_main_results_massive}

In this section, we obtain the decay estimates for the energy-momentum tensor $\T_{\mu\nu}$ stated in Theorem \ref{theorem_decay_fast}, Theorem \ref{theorem_decay_bded}, and Theorem \ref{theorem_decay_slow}. The proofs of the main results are based on the concentration estimates previously obtained on the stable manifolds associated unstable trapping, degenerate trapping at ISCO, and parabolic trapping at infinity. The estimates obtained in the previous section control the flow in a neighbourhood of the stable manifolds. In this section, we also estimate the flow in the rest of the mass-shell. We recall that the initial distribution function $f_0$ is assumed to be compactly supported. 

\subsection{A priori estimates in the near-horizon region}

In this subsection, we show that we can assume without loss of generality that the initial distribution function $f_0$ is supported away of $\mathcal{H}^+$.

\begin{lemma}\label{lem_horizon_control}
There exists $r_0>2M$ and $V_0>0$ such that for every geodesic $\gamma_{x,p}\colon [0,a]\to \{r<r_0\}$ with $(x,p)\in \supp (f_0)$, we have $v(s)\leq V_0$ for all $s\in [0,a]$.
\end{lemma}

\begin{proof} 
By the compact support assumption of $f_0$ in the momentum variables, there exist $r_0>2M$ and an advanced time $V_0>0$, such that every geodesic $\gamma_{x,p}$ crosses $\mathcal{H}^+$ in $\{v\leq V_0\}$ for $(x,p)\in\supp(f_0)\,\cap\, \{r\leq r_0\}$. This property holds because the momentum $p$ is strictly away of the null generator of $\mathcal{H}^+$ in $T_x\mathcal{E}$ for all $(x,p)\in \supp (f_0)\,\cap \,\{r=2M\}$. We are using here the compact support of $f_0$. 
\end{proof}

\begin{remark}
Lemma \ref{lem_horizon_control} only holds for massive Vlasov fields $f$ for which the initial data $f_0$ are compactly supported in the subset $\Sigma$ of the mass-shell over the initial hypersurface.
\end{remark}

\begin{lemma}\label{lemma_lowerbound_partenergy_horizon}
There exist constants $L_1\geq 0$, and $E_1\geq E_0>0$ such that for all $(x,p)\in \supp (f_0)\,\cap \, \{r\geq r_0\}$, we have $\ell(x,p)\in [0,L_1]$, and $E(x,p)\in[E_0,E_1]$.
\end{lemma}

\begin{proof} 
By the compact support of the initial distribution function $f_0$, there exist $E_1>0$ and $L_1>0$ such that $E(x,p)\in [0,E_1]$ and $\ell(x,p)\in [0,L_1]$ for all $(x,p)\in \supp (f_0)$. By the mass-shell relation, the particle energy $E$ satisfies $$\forall (x,p)\in\supp(f_0),\qquad E^2(x,p)\geq \Big(1-\dfrac{2M}{r}\Big)\Big(1+\dfrac{\ell^2}{r^2}\Big) \geq 1-\frac{2M}{r}\geq 1-\frac{2M}{r_0}=:E_0^2>0.$$
\end{proof}

Next, we upgrade Lemma \ref{lem_horizon_control} by removing the assumption that $(x,p)\in \supp (f_0)$.

\begin{corollary}\label{cor_horizon_control_general}
There exist $r_0>2M$ and $V_1>0$ such that for every geodesic $\gamma_{x,p}\colon [0,a]\to \{r<r_0\}\,\cap\, \supp (f)$ with particle energy $E\geq E_0$, we have $v(s)\leq V_1$ for all $s\in [0,a]$.
\end{corollary}

\begin{proof}
By Lemma \ref{lem_horizon_control}, we have $|\ell(x,p)|\leq L_1$ for all $(x,p)\in \supp (f)$. As a result, we have 
\begin{equation}\label{estimate_horizon_control_general}
4\Omega^2p^up^v=\frac{\ell^2}{r^2}+1\leq \frac{L_1^2}{4M^2}+1.
\end{equation}
Furthermore, the normalised momentum coordinate $\Omega^2p^u$ satisfies $$\Omega^2p^u=E-\Omega^2p^v\geq E_0- \Omega^2p^v\geq E_0-\Omega,$$ where we have used Lemma \ref{lemma_lowerbound_partenergy_horizon} in the first estimate, and \eqref{estimate_good_dpv} in the second one. By setting $r_0-2M$ small enough, we obtain that $\Omega^2 p^u$ is strictly away from zero, so $p^v$ can be bounded above by using \eqref{estimate_horizon_control_general}. Thus, for all $(x,p)\in \supp (f)\,\cap \, \{r=2M\}$, the momentum $p$ is strictly away from the null generator of $\mathcal{H}^+$. We are using here the compact support assumption on $f_0$. From this, the result follows.
\end{proof}

In the rest of the section, we proceed to show the main results of the article

\subsection{Decay for Vlasov fields compactly supported on $\D_0$}

In this subsection, we prove time decay for the energy-momentum tensor $\T_{\mu\nu}$ of a Vlasov field $f$ with initial data compactly supported on $\D_{0}$. We note that there exist constants $E_1>E_0>1$ and $L_1\geq 0$, such that $E(x,p)\in [E_0,E_1]$ and $\ell(x,p)\in [0,L_1]$ for all $(x,p)\in \supp (f_0)$. This properties hold by the compact support assumption on the initial distribution function $f_0$.

The proof of the decay estimates of the energy-momentum tensor $\T_{\mu\nu}$ follow by proving decay of the momentum support of the distribution $f$. In this subsection, we will show the decay of the momentum support by proving concentration estimates on the stable manifolds associated to the energy level $\{E=E_-(\ell)\}$ for $\ell\geq 4M$. We obtain concentration estimates in terms of the coordinate $p^v$ to apply these estimates later to the components of $\mathrm{T}_{\mu\nu}$.

\subsubsection{Estimates in the bounded region}

First, we recall that the coordinate $p^v$ can be written as 
\begin{equation}\label{ident_pv_basic_particle_radial_moment}
p^v=\frac{1}{2\Omega^2}(E\,+\, p^r).
\end{equation}
In the following, we denote the coordinate $p^v$ of past-trapped and future-trapped geodesics with angular momentum $\ell>2\sqrt{3}M$ and particle energy $E=E_-(\ell)$, by $p^{v,+}_{\ell}$ and $p^{v,-}_{\ell}$, respectively. By \eqref{ident_pv_basic_particle_radial_moment}, we deduce $$p^{v,\pm}_{\ell}(r):=\dfrac{1}{2\Omega^2}\Big(E_-(\ell)\,+\, p^{r,\pm}_{\ell}(r)\Big),\qquad \quad p^{r,\pm}_{\ell}(r)=\sgn\Big(\pm(r-r_-(\ell))\Big)\sqrt{E^2_-(\ell)\,-\,V_{\ell}(r)},$$ where we have used the mass-shell relation. We will now show a concentration estimate in terms of $p^v$ using the bounds previously derived for $p^r$ in Proposition \ref{lemma_application_stable_manifol_thm_large_angular_momentum}. 

\begin{lemma}\label{lemma_decay_pv_coordinate_large_angular_momentum}
There exists $R>2M$ such that for every geodesic $\gamma_{x,p}\colon [0,a]\to \{r<R\}\,\cap \, \supp(f)$ with angular momentum $\ell\geq 4M$ and particle energy $E\geq E_0>1$, we have 
\begin{equation}\label{estimate_first_lemma_theorem_fast}
    |p^v(s)\,-\, p^{v,+}_{\ell}(r(s))|\lesssim \exp\Big(-\frac{1}{4\sqrt{2}M}v(s)\Big).
\end{equation}
\end{lemma}

\begin{proof}
By the monotonicity of the map $\ell \mapsto E_-(\ell)$, there exists $\ell_0>4M$ such that $E_0=E_-(\ell_0)$. As a result, unstable trapping only holds for $\ell\geq \ell_0$ in $\supp(f)$. We assume without loss of generality that $f$ is supported on $\{\ell\geq \ell_0, \,|E-E_-(\ell)|\leq \frac{1}{2}(E_0-1) \}$, because of the compact support assumption on the initial data. The geodesics in the complement $\supp(f)\setminus (\{\ell\geq \ell_0,\, |E-E_-(\ell)|\leq \frac{1}{2}(E_0-1) \})$ either cross $\mathcal{H}^+$ or leave the region $\{r<R\}$, before the advanced time $v$ is sufficiently large. 

\emph{The subregion $r\sim r_-(\ell)$.} We first prove the estimate  \eqref{estimate_first_lemma_theorem_fast} in the region $\{ |r-r_-(\ell)|<\epsilon \}$. By Proposition \ref{lemma_application_stable_manifol_thm_large_angular_momentum}, there exists $\epsilon>0$ such that for every geodesic $\gamma \colon [s_1,s_2]\to \{r<R\}\,\cap \, \supp(f)$ with angular momentum $\ell\geq 4M$ and particle energy $E\in (E_-(\ell)-\epsilon,E_-(\ell)+\epsilon)$, we have 
\begin{align*}
|p^r(s)\,-\,p^{r,+}_{\ell}(r(s))|&\lesssim \exp\Big( -\frac{1}{4\sqrt{2}M}\Big(v(s_2)-v(s_1)\Big)\Big),\\  
\Big|\frac{2M}{1-E^2}\,-\,\frac{2M}{1-E^2_-(\ell)}\Big|&\lesssim   \exp\Big( -\frac{1}{4\sqrt{2}M}\Big(v(s_2)-v(s_1)\Big)\Big),
\end{align*} 
for all $(x,p)\in \supp(f)\,\cap\,\{|r-r_-(\ell)|<\epsilon\}$. Thus, the difference $E \,-\, E_-(\ell)$ satisfies that $$|E\,-\,E_-(\ell)|\lesssim |E^2\,-\,E^2_-(\ell)|=|(p^r)^2\,-\,(p^{r,+}_{\ell})^2|\lesssim \exp\Big( -\frac{1}{4\sqrt{2}M}\Big(v(s_2)-v(s_1)\Big)\Big).$$ where we have used that $E\,+\,E_-(\ell)\geq 1$ for all $\ell\geq 4M$. Moreover, the difference $p^v\,-\,p^{v,+}_{\ell}$ satisfies
\begin{align*}
    |p^v(s)\,-\,p^{v,+}_{\ell}(r(s))|&\leq\dfrac{1}{2\Omega^2}|E\,-\,E_-(\ell)|+ \dfrac{1}{2\Omega^2}|p^r\,-\,p^{r,+}_{\ell}
    |\lesssim \exp\Big( -\frac{1}{4\sqrt{2}M}\Big(v(s_2)-v(s_1)\Big)\Big) .
    \end{align*}

\emph{Away from $r\sim r_-(\ell)$ and $r\sim 2M$.} We consider next the complement of the domain $ \{|r-r_-(\ell)|<\epsilon\}$ as a subset of $\{r_0<r<R\}$. Let us bound the advanced time that geodesics spend in $\{r_0\leq r\leq r_-(\ell)-\epsilon\}\, \cup\, \{r_-(\ell)+\epsilon<r<R\}$. There exists $C_0>0$ such that 
\begin{equation}\label{estim_lowe_bound_pr_away_trapp}
(p^r)^2=E^2-V_{\ell}(r)>C_0^2>0,
\end{equation}
in the region $\{r>r_0\}$, since the turning points for these orbits occur for $r\sim r_-(\ell)$. As a result, for every geodesic $\gamma$ contained in $\{r_0\leq r\leq r_-(\ell)-\epsilon\}\, \cup\, \{r_-(\ell)+\epsilon<r<R\}$, we have 
\begin{equation}\label{est_adv_time_bounded_region}
\forall s_2\geq s_1,\qquad\quad v(s_2)-v(s_1)\lesssim \int_{v(s_1)}^{v(s_2)}\Omega^2(r) \d v\leq \int_{v(s_1)}^{v(s_2)}\frac{\d v}{p^v}= s_2-s_1=\int_{r(s_1)}^{r(s_2)}\frac{\d r}{p^r} \lesssim R,
\end{equation}
where we have used that $\Omega^2 p^v\leq E\leq E_1$, and \eqref{estim_lowe_bound_pr_away_trapp}. We obtain that in finite advanced time every geodesic leaves $ \{|r-r_-(\ell)|>\epsilon\}\,\cap \,\{r<R\}$. 

By the contraction property of $\varphi_-^{\ell}$ in \eqref{expcontrac_hyperbol_2}, we have the monotonicity property 
\begin{equation}
\forall \,s_2\geq s_1, \qquad\quad \Big|\varphi_-^{\ell}\Big(x(s_2),p(s_2)\Big)\Big|\leq \Big|\varphi_-^{\ell}\Big(x(s_1),p(s_1)\Big)\Big|.
\end{equation}
Let us consider a geodesic $\gamma \colon [s_1,s_2]\to \{r_0\leq r\leq r_-(\ell)-\epsilon\}\, \cup\, \{r_-(\ell)+\epsilon<r<R\}\,\cap \, \supp(f)$. By the monotonicity property of $\varphi_-$, and the uniform bound on the advanced time away from the trapped set, we have that $$\Big|\varphi_-^{\ell}\Big(x(s_2),p(s_2)\Big)\Big|\leq \Big|\varphi_-^{\ell}\Big(x(s_1),p(s_1)\Big)\Big|\lesssim  \exp\Big( -\frac{1}{4\sqrt{2}M}v(s_1)\Big)\lesssim  \exp\Big( -\frac{1}{4\sqrt{2}M}v(s_2)\Big),$$ since $(x(s_1),p(s_1))$ can be assumed to be in $\{r-r_-(\ell)\leq \epsilon\}$. Finally, we follow the same arguments used in the subregion $r\sim r_-(\ell)$ to show that \eqref{estimate_first_lemma_theorem_fast} holds. First, one propagates the estimates \eqref{estim_concentr_high_ang} in Proposition \ref{lemma_application_stable_manifol_thm_large_angular_momentum}. As a result, we control the difference $E^2\,-\,E^2_-(\ell)$, and then the difference $p^v(s)\,-\,p^{v,+}_{\ell}(r(s))$. 

\emph{The subregion $r\sim 2M$.} We use first Proposition \ref{lemma_redshift} to obtain that
\begin{equation}
\dfrac{r^2}{\ell^2+r^2}(p^v\,-\, p^{v,+}_{\ell}(r))(s_2)=\dfrac{r^2}{\ell^2+r^2}(p^v\,-\,p^{v,+}_{\ell}(r))(s_1)\exp\Big(-\int_{v(s_1)}^{v(s_2)}\frac{2M}{(r^2+\ell^2)r^2}\Big(r^2-\frac{\ell^2}{M}r+3\ell^2\Big)\d   v'\Big),\nonumber
\end{equation}
for any piece of geodesic. We remark that the exponent in the RHS of the identity above is negative for $r-2M$ sufficiently small. Thus, we have the estimate $$|p^v(s_2)\,-\,p^{v,+}_{\ell}(r(s_2))|\lesssim |p^v(s_1)\,-\,p^{v,+}_{\ell}(r(s_1))|\lesssim  \exp\Big( -\frac{1}{4\sqrt{2}M}v(s_1)\Big)\lesssim  \exp\Big( -\frac{1}{4\sqrt{2}M}v(s_2)\Big)$$
since the difference $v(s_2)-v(s_1)$ is uniformly bounded.

The estimate \eqref{estimate_first_lemma_theorem_fast} follows by putting together the estimates performed in the three regions considered. We note that the constant in the RHS of \eqref{estimate_first_lemma_theorem_fast} is uniform among the geodesics under study.
\end{proof}

We show now the estimates in Theorem \ref{theorem_decay_fast} in the bounded region of spacetime.\\

\emph{Proof of Theorem \ref{theorem_decay_fast} in the bounded region}. By Lemma \ref{lemma_decay_pv_coordinate_large_angular_momentum}, the support of the distribution function in the coordinate $p^v$ decays exponentially. Then, for all $x\in \{r<R\}$, the component $\T_{uv}[f]$ of the energy-momentum tensor satisfies
\begin{align*}
    \frac{1}{\Omega^2}\T_{uv}[f]&=\frac{1}{\Omega^2}\int_{p^v}\int_{p^A}\int_{p^B} f(x,p)p_{u} p_{v}\dfrac{r^2\sqrt{\det \gamma}}{p^v}\d p^v\d p^A\d p^B\\
    &=\dfrac{1}{4\Omega^2}\int_{p^v}\int_{p^A}\int_{p^B} f(x,p)(\Omega^2p^{u}) (\Omega^2p^{v})\dfrac{r^2\sqrt{\det \gamma}}{p^v}\d p^v\d p^A\d p^B\\
    &\lesssim  \dfrac{\|f_0\|_{L^{\infty}_{x,p}}}{\exp(\frac{1}{4\sqrt{2}M}v)} \int_{r^4(g_{\S^2})_{AB}p^Ap^B\leq L_1} r^2 \d p^A\d p^B\lesssim  \dfrac{\|f_0\|_{L^{\infty}_{x,p}}}{\exp(\frac{1}{4\sqrt{2}M}v)},
\end{align*}
where we have estimated $\Omega^2p^u$ using Lemma \ref{lemma_asymptotic_tangent_at_horizon}. The same argument shows the decay estimates for the other  components of $\mathrm{T}_{\mu\nu}$.

\subsubsection{Estimates in the far-away region} 

We begin proving a comparison lemma between the retarded time $u(s)$ and the radial coordinate $r(s)$, along outgoing geodesics in the far-away region.

\begin{lemma}\label{lem_comparison_r_u_high_angular_mom}
There exists $R>2M$ such that for every outgoing geodesic $\gamma_{x,p}\colon [0,a]\to \{r>R\}\,\cap\,\supp(f) $ with angular momentum $\ell\geq 4M$, and particle energy $E\geq E_0>1$, we have $u(s)\lesssim  r(s)$ for all $s\in[0,a]$.
\end{lemma}

\begin{proof}
There exists $R>2M$ such that every geodesic $\gamma_{x,p}\colon [0,a]\to \{r>R\}\, \cap \,\supp(f)$ is outgoing, since $E>1$ on the support of the distribution $f$. We can assume that $R>2M$ is large enough so that $$\forall (x,p)\in \{r>R\},\qquad p^r= \Big(E^2-1+\dfrac{2M}{r}-\dfrac{\ell^2}{r^2}+\dfrac{2M\ell^2}{r^3}\Big)^{\frac{1}{2}}\geq \dfrac{\sqrt{E^2-1}}{2},$$ for every outgoing geodesic. We are using here that $E\geq E_0>1$. As a result, the parameter $s$ along a geodesic $\gamma$ in $\supp(f)$ satisfies 
\begin{equation}\label{estimate_upper_bound_r}
s=\int_{r(0)}^{r(s)}\dfrac{\d r}{p^r}=\int_{r(0)}^{r(s)}\dfrac{\d r}{\sqrt{E^2-V_{\ell}(r)}}\lesssim r(s).
\end{equation}
On the other hand, the coordinate $p^u$ satisfies $\Omega^2p^u\leq E\leq E_1$, so we have
\begin{equation}\label{estimate_lower_bound_u}
s=\int_{u(0)}^{u(s)}\frac{\d u}{p^u}\geq \int_{u(0)}^{u(s)}\Omega^2(r) \d u \gtrsim u(s),
\end{equation}
for every geodesic $\gamma$ in $\{r>R\}\,\cap \, \supp (f)$. Thus, we have $r(s)\gtrsim u(s)$ for all $s\in[0,a]$.
\end{proof}

We conclude the proof of Theorem \ref{theorem_decay_fast} with the estimates in the far-away region.\\

\emph{Proof of Theorem \ref{theorem_decay_fast} in the far-away region}. By Proposition \ref{proposition_asymptotic_tangent_at_infinity}, the support of the distribution in the coordinate $p^v$ decays like $r^{-1}$. We are using here that $E>1$. Then, for all $x\in \{r>R\}$, the component $\T_{uv}[f]$ of the energy-momentum tensor satisfies
\begin{align*}
    \T_{uv}[f]&=\int_{p^v}\int_{p^A}\int_{p^B} f(x,p)p_{u} p_{v}\dfrac{r^2\sqrt{\det \gamma}}{p^v}\d p^v\d p^A\d p^B\\
    &=\dfrac{1}{4}\int_{p^v}\int_{p^A}\int_{p^B} f(x,p)(\Omega^2p^{u}) (\Omega^2p^{v})\dfrac{r^2\sqrt{\det \gamma}}{p^v}\d p^v\d p^A\d p^B\\
    &\lesssim \|f_0\|_{L^{\infty}_{x,p}} \int_{r^4(g_{\S^2})_{AB}p^Ap^B\leq L_1}r \d p^A\d p^B\\
    &\lesssim \dfrac{\|f_0\|_{L^{\infty}_{x,p}}}{r^3}\lesssim \dfrac{\|f_0\|_{L^{\infty}_{x,p}}}{u^3}.
\end{align*}
where we used \eqref{estimate_upper_bound_r}--\eqref{estimate_lower_bound_u} in the last estimate. The same argument shows the decay estimates for the other  components of $\mathrm{T}_{\mu\nu}$.

\subsection{Decay for Vlasov fields supported up to the boundary of $\D_1$} 

The proof of Theorem \ref{theorem_decay_bded} follows by the same strategy followed in the proof of Theorem \ref{theorem_decay_fast}. In the bounded region, the estimates can be directly extended to $f\chi_{\mathcal{D}_1}$ because the normal hyperbolicity of the sphere of trapped orbits $\mathcal{S}^{-}(\ell)$ holds for all $\ell \leq 4M$. Let us prove now the decay estimates in the far-away region. We note that there exist constants $E_1\geq 1$ and $L_1\geq 0$, such that $E(x,p)\in [1,E_1]$ and $\ell(x,p)\in [0,L_0]$ for all $(x,p)\in \supp (f \chi_{\mathcal{D}_1})$. These a priori bounds hold because of the compact support assumption of $f_0$.

\subsubsection{Estimates in the far-away region} 

We first prove a comparison lemma between the retarded time $u(s)$ and the radial coordinate $r(s)$ along outgoing orbits. 

\begin{lemma}
There exists $R>2M$ such that for every outgoing geodesic $\gamma_{x,p}\colon [0,a]\to \{r>R\}\,\cap \,\supp(f )$ with angular momentum $\ell\geq 4M$, and particle energy $E\geq 1$, we have $u(s)\lesssim  r^{\frac{3}{2}}(s)$ for all $s\in[0,a]$.
\end{lemma}

\begin{proof}
There exists $R>2M$ such that every geodesic $\gamma_{x,p}\colon [0,a]\to \{r>R\}\, \cap \,\supp(f)$ is outgoing, since $E\geq 1$ on the support of the distribution $f$. By the estimate \eqref{estimate_radial_momentum_particle_enegy_larger_one} in the proof of Proposition \ref{proposition_asymptotic_tangent_at_infinity}, the radial coordinate $r(s)$ along outgoing geodesics with $E\geq 1$ satisfies $p^r(s) \gtrsim r^{-\frac{1}{2}}(s)$. Thus, we have  
\begin{equation}\label{sharp_comparison_parabolic}
r^{\frac{3}{2}}(s)\geq \int_{r(0)}^{r(s)}r^{\frac{1}{2}}\d r\gtrsim \int_{r(0)}^{r(s)}\dfrac{\d r}{p^r}=\int_{u(0)}^{u(s)}\dfrac{\d u}{p^u}\geq \int_{u(0)}^{u(s)}\Omega^2(r) \d u \gtrsim u(s),
\end{equation}
where we used $\Omega^2p^u\leq E\leq E_1$. Thus, we have $r^{\frac{3}{2}}(s)\gtrsim u(s)$ for all $s\in[0,a]$.
\end{proof}

We perform now the estimates in Theorem \ref{theorem_decay_bded} in the far-away region.\\

\emph{Proof of Theorem \ref{theorem_decay_bded} in the far-away region}. By Proposition \ref{proposition_asymptotic_tangent_at_infinity}, the support of the distribution in the coordinate $p^v$ decays like $r^{-\frac{1}{2}}$.We are using here that $E\geq 1$. Thus, for all $x\in \{r>R\}$, the component $\T_{uv}[ f \chi_{\mathcal{D}_1}]$ of the energy-momentum tensor satisfies 
\begin{align}
    \T_{uv}[ f \chi_{\mathcal{D}_1}]&=\int_{p^v}\int_{p^A}\int_{p^B} (f\chi_{\mathcal{D}_1})(x,p)p_{u} p_{v}\dfrac{r^2\sqrt{\det \gamma}}{p^v}\d p^v\d p^A\d p^B\nonumber\\
    &=\dfrac{1}{4}\int_{p^v}\int_{p^A}\int_{p^B} (f\chi_{\mathcal{D}_1})(x,p)(\Omega^2p^{u}) (\Omega^2p^{v})\dfrac{r^2\sqrt{\det \gamma}}{p^v}\d p^v\d p^A\d p^B\nonumber\\
    &\lesssim \|f_0\|_{L^{\infty}_{x,p}} r^{\frac{3}{2}} \int_{r^4(g_{\S^2})_{AB}p^Ap^B\leq L_1} \d p^A\d p^B\nonumber\\
    &\lesssim  \dfrac{\|f_0\|_{L^{\infty}_{x,p}}}{r^{\frac{5}{2}}}\lesssim  \dfrac{\|f_0\|_{L^{\infty}_{x,p}}}{u^{\frac{5}{3}}},\label{decay_estimatefgeq1faraway}
\end{align}
where we used \eqref{sharp_comparison_parabolic} in the last estimate. The same argument shows the decay estimates for the other components of $\mathrm{T}_{\mu\nu}[f\chi_{\mathcal{D}_1}]$.

\subsection{Decay for Vlasov fields supported up to the boundary of $\D$} 

Let us prove decay in time for the energy-momentum tensor $\T_{\mu\nu}[f \chi_{\mathcal{D}}]$ of a Vlasov field $f \chi_{\mathcal{D}}$ supported on $\D$. The proof of the decay estimates for the energy-momentum tensor follow by proving decay of the momentum support of the distribution. Note that there exist $E_1\geq 0$ and $L_1\geq 0$, such that $E(x,p)\in [0,E_1]$ and $\ell(x,p)\in [0,L_1]$ for all $(x,p)\in \supp (f\chi_{\mathcal{D}})$.

We have previously obtained time decay for the energy-momentum tensor $\T_{\mu\nu}$ of a Vlasov field $f$ with initial data compactly supported on $\D_1$. In order to use our previous decay estimates, we decompose the distribution $f$ as $$f=f_{\geq 1}+ f_{< 1}:=f\chi_{\{E\geq 1\}}+f\chi_{\{E< 1\}},$$ where $\chi_{\{E\geq 1\}}\colon \P\to [0,1]$ and $\chi_{\{E< 1\}}\colon \P\to [0,1]$ are standard characteristic functions. As a result, the components $\T_{\mu\nu}[f\chi_{\mathcal{D}}]$ of the energy-momentum tensor of the Vlasov field $f\chi_{\mathcal{D}}$ satisfy 
$$\T_{\mu\nu}[f\chi_{\mathcal{D}}]=\T_{\mu\nu}[f_{\geq 1}\chi_{\mathcal{D}}]+\T_{\mu\nu}[f_{< 1}\chi_{\mathcal{D}}],$$ for every $\mu$, $\nu$, where $\chi_{\mathcal{D}}\colon \mathcal{P}\to \R$ is the characteristic function of $\mathcal{D}$. We will show that $\T_{\mu\nu}[f_{\geq 1}\chi_{\mathcal{D}}]$ decays faster than $\T_{\mu\nu}[f_{< 1}\chi_{\mathcal{D}}]$, in both, the bounded and the far-away regions. 

Let us now focus on the region $\{E< 1\}$. We will analyse the momentum support of the distribution function in the following four subregions of phase space:
\begin{align*}
\D_{\textrm{a}}&=\Big\{(x,p)\in \P: \ell(x,p)\in \Big(2\sqrt{3}M, 4M\Big), \, E(x,p) \in \Big(\frac{2}{3},\frac{21\sqrt{2}}{30}\Big)\Big\},\\
\D_{\textrm{b}}&=\Big\{(x,p)\in \P: \ell(x,p)\in (2\sqrt{2}M, 2\sqrt{3}M), \, E(x,p) \in \Big(\frac{2}{3},\frac{21\sqrt{2}}{30}\Big)\Big\},\\
\D_{\textrm{c}}&=\Big\{(x,p)\in \P: \ell(x,p)\in (0,2\sqrt{3}M), \, E(x,p) \in \Big(\frac{21\sqrt{2}}{30},1\Big)\Big\},\\
\D_{\textrm{d}}&=\Big\{(x,p)\in \P: \ell(x,p)\in (2\sqrt{3}M, 4M),  E(x,p) \in \Big(\frac{21\sqrt{2}}{30},1\Big)\Big\}.
\end{align*}
In the complementary region of the domain $\D$, given by $$\D_{\mathrm{e}}:=\Big\{(x,p):E(x,p)<1\Big\}\setminus \Big(\D_a\,\cup\,\D_b\,\cup\,\D_c\,\cup\,\D_d\Big),$$ geodesics cross $\mathcal{H}^+$ before the advanced time $v$ is sufficiently large. 

In the bounded region, we assume without loss of generality that the Vlasov field $f \chi_{\mathcal{D}}$ is supported in the regions $\D_a$, $\D_b$, $\D_c$, and $\D_d$. In the far-away region, we assume instead that the Vlasov field $f \chi_{\mathcal{D}}$ is supported in the regions $\D_{\textrm{c}}$, and $\D_{\textrm{d}}$. The regions $\D_a$, $\D_b$, and $\D_e$, only concern the decay estimates in the bounded region.

\subsubsection{Estimates in the bounded region}

We first show concentration estimates in a neighbourhood of the sphere $\{r=6M\}$ of trapped orbits. We recall that the hyperbolicity of the radial flow on the spheres of trapped orbits $\mathcal{S}^-(\ell)$ degenerates when $\ell=2\sqrt{3}M$. We will show inverse polynomial decay estimates in this regime. We begin considering the case when $\ell\leq 2\sqrt{3}M$.

\begin{lemma}\label{lem_isco_smaller_angu_mom_ingred1}
For every geodesic $\gamma_{x,p}\colon [0,a]\to \{r<R\}\,\cap\,\supp(f\chi_{\mathcal{D}})$ with angular momentum $\ell\in (2\sqrt{2}M, 2\sqrt{3}M)$ and particle energy $E \in (\frac{2}{3},\frac{21\sqrt{2}}{30})$, we have 
\begin{align}
\Big|p^v(s)\,-\,\frac{E}{2\Omega^2(s)}\Big|&\lesssim \frac{1}{v^{3}(s)},\qquad \text{if}\quad r(s)\geq r_-(\ell),\label{est_isco_small_angul_ingred_non_stablemnfld}\\
|p^v(s)\,-\,p^{v,+}_{\ell}(r(s))|&\lesssim \frac{1}{v^{3}(s)},\qquad \text{if}\quad r(s)\leq r_-(\ell).\label{est_isco_small_angul_ingred_stablemnfld}
\end{align}
\end{lemma}

\begin{proof}
In this region of phase space, trapping only holds at $\mathcal{S}^-(2\sqrt{3}M)\subset \{r=6M\}$. We also recall that the stable manifolds $W^{\pm}(2\sqrt{3}M)$ associated to the sphere of trapped orbits $\mathcal{S}^-(2\sqrt{3}M)$, are contained in $\{r\leq 6M\}$. As a result, we can assume without loss of generality that the geodesics considered in this lemma are contained in $\{r\leq 7M\}$. We also assume without loss of generality that $\ell\geq \ell_0> \frac{4\sqrt{2}}{\sqrt{3}}M$ for some $\ell_0> \frac{4\sqrt{2}}{\sqrt{3}}M$. We make this assumption to avoid any discussion for $\ell\sim  \frac{4\sqrt{2}}{\sqrt{3}}M$ where $E_-(\ell)$ is not defined anymore. See Proposition \ref{prop_unique_energy _level} for more details.

Let $r_0>2M$ such that $r_0-2M$ is small enough. Let $\ell\in ( \frac{4\sqrt{2}}{\sqrt{3}}M, 2\sqrt{3}M)$. We divide the analysis in four different subregions: when $r_-(\ell)\leq r\leq 7M$, when $5M\leq r\leq r_-(\ell)$, when $2M+r_0\leq r \leq 5M$, and when $r\leq r_0$. We address these subregions in the same order.

\emph{The subregion $r_-(\ell)\leq r\leq 7M$}. By Proposition \ref{prop_conc_est_local_isco_below}, for every geodesic in $\{5M \leq r\leq 7M\}$ with angular momentum $\ell\in ( \frac{4\sqrt{2}}{\sqrt{3}}M, 2\sqrt{3}M)$ and particle energy $E\in (\frac{2}{3},\frac{21\sqrt{2}}{30})$, we have $$|p^r(s)|\lesssim  v^{-3}(s),\qquad \text{if} \quad r(s)\geq r_-(\ell).$$ Thus, the difference $p^v(s)\,-\,\frac{1}{2}\Omega^{-2}E$ satisfies that
\begin{align*}
        \Big|p^v(s)\,-\,\frac{E}{2\Omega^2(s)}\Big|&\leq \dfrac{1}{2\Omega^2}|p^r(s)|\lesssim v^{-3}(s),
\end{align*}
where we used the relation $p^v=\frac{1}{2\Omega^2}(E+ p^r).$

\emph{The subregion $5M\leq r\leq r_-(\ell)$}. By Proposition \ref{prop_conc_est_local_isco_below}, for every geodesic in $\{5M \leq r\leq 7M\}$ with angular momentum $\ell\in ( \frac{4\sqrt{2}}{\sqrt{3}}M, 2\sqrt{3}M)$ and particle energy $E\in (\frac{2}{3},\frac{21\sqrt{2}}{30})$, we have $$|p^r(s)\,-\,p^{r,+}_{\ell}(r(s))|=\Big|\frac{p^r}{\sqrt{1-E^2}}\,-\,\frac{(r_--r)^{\frac{1}{2}}}{r^{\frac{3}{2}}}\Big((r-r_-)^2+\frac{r_-^3}{r_-^2-\frac{\ell^2}{2M}r_-+\ell^2}\Big(\frac{\ell^2}{2M}-r_-\Big)\Big)^{\frac{1}{2}}\Big|\lesssim v^{-3}(s).$$ Thus, the difference $p^v\,-\,p^{v,+}_{\ell}$ satisfies that
\begin{align*}
        |p^v(s)\,-\,p^{v,+}_{\ell}(r(s))|&\leq \dfrac{1}{2\Omega^2}|E\,-\,E_{-}(\ell)|+ \dfrac{1}{2\Omega^2}|p^r(s)-p^{r,+}_{\ell}(r(s))|\lesssim v^{-3}(s),
\end{align*}
where we have used the estimate for the particle energy in Proposition \ref{prop_conc_est_local_isco_below}.

\emph{The subregion $r_0\leq r \leq 5M$.} The analysis in this subregion follows by the same arguments performed at the end of the proof of Lemma \ref{lemma_decay_pv_coordinate_large_angular_momentum}. We simply propagate to this region the concentration estimate obtained in the subregion $\{5M\leq r\leq 7M\}$.

\emph{The subregion $r\leq r_0$.} The analysis in this subregion follows by the same arguments performed at the end of the proof of Lemma \ref{lemma_decay_pv_coordinate_large_angular_momentum}. Specifically, we use the red-shift effect in the form presented in Proposition \ref{lemma_redshift}.

The estimates \eqref{est_isco_small_angul_ingred_non_stablemnfld}--\eqref{est_isco_small_angul_ingred_stablemnfld} follow by putting together the estimates performed in the three subregions considered above.
\end{proof}

We will now consider the case of orbits with angular momentum $\ell\geq 2\sqrt{3}M$. In the next lemma, we estimate the geodesic flow in a neighbourhood of the homoclinic orbits at the energy level $\{E=E_-(\ell)\}$ when $\ell$ is strictly smaller than $4M$. 

Let us recall that there exists a uniform constant $A>2M$ such that $-a(\ell)\leq A$ for every geodesic with $\ell \in (2\sqrt{3}M, 4M)$ and $E\in (\frac{2}{3},\frac{21\sqrt{2}}{30})$. This property was obtained in the proof of Proposition \ref{prop_conc_estim_homocl_bounded_angu_mom}.

\begin{lemma}\label{lem_isco_larger_angu_mom_ingred2}
For every geodesic $\gamma_{x,p}\colon [0,a]\to \{r<R\}\,\cap\,\supp(f\chi_{\mathcal{D}})$ with angular momentum $\ell \in (2\sqrt{3}M, 4M)$ and particle energy $E\in (\frac{2}{3},\frac{21\sqrt{2}}{30})$, we have
\begin{align}
\Big|p^v(s)\,-\,\frac{E}{2\Omega^2(r(s))}\Big|&\lesssim \frac{1}{v^{3}(s)},\qquad \text{if}\quad r\in [-a(\ell),R+A], \label{estim_homoclinic_low_ang_farfar}\\
|p^v(s)\,-\,p^{v,+}_{\ell}(r(s))|&\lesssim \frac{1}{v^{3}(s)},\qquad \text{if}\quad p^r\geq 0 \quad \text{and}\quad r\in [r_-(\ell),-a(\ell)],\\
|p^v(s)\,-\,p^{v,-}_{\ell}(r(s))|&\lesssim \frac{1}{v^{3}(s)},\qquad \text{if}\quad p^r\leq 0 \quad \text{and}\quad r\in [r_-(\ell),-a(\ell)],\label{estim_homoclinic_low_ang_bounded_neg_pr}\\
|p^v(s)\,-\,p^{v,+}_{\ell}(r(s))|&\lesssim \frac{1}{v^{3}(s)},\qquad \text{if}\quad  r\in [2M,r_-(\ell)].
\end{align}
\end{lemma}

\begin{proof}
In this region of phase space, unstable trapping only holds for $\ell\leq \ell_1$ for some $\ell_1<4M$. The homoclinic orbits associated to the sphere of trapped orbits $\mathcal{S}^-(\ell)$, are contained in the region $\{r\leq A\}$ for $A>2M$. Thus, we can assume without loss of generality that the geodesics considered in this lemma are contained in $\{r\leq A\}$.

Let $r_0>2M$ such that $r_0-2M$ is small enough. Let $\ell\in (2\sqrt{3}M,\ell_1)$. We divide the analysis in five different subregions: when $-a(\ell)\leq r\leq R+A$, when $r_-(\ell)\leq r\leq -a(\ell)$, when $5M\leq r \leq r_-(\ell)$, when $r_0\leq r \leq 5M$, and when $r\leq r_0$. We address these subregions in the same order. 

\emph{The subregion $-a(\ell)\leq r\leq R+A$}. By Proposition \ref{proposition_decay_estimate_isco}, for every geodesic in $-a(\ell)\leq r\leq R+A$ with angular momentum $\ell\in ( 2\sqrt{3}M,4M)$ and particle energy $E\in (\frac{2}{3},\frac{21\sqrt{2}}{30})$, we have $$|p^r(s)|\lesssim \frac{1}{v^{3}(s)},\qquad  \text{if} \quad r(s)\geq -a(\ell).$$ Thus, the difference $p^v(s)\,-\,\frac{1}{2}\Omega^{-2}E$ satisfies that
\begin{align*}
        \Big|p^v(s)\,-\,\frac{E}{2\Omega^2(s)}\Big|&\leq \dfrac{1}{2\Omega^2}|p^r(s)|\lesssim \frac{1}{v^{3}(s)},
\end{align*}
where we used that $p^v=\frac{1}{2\Omega^2}(E+ p^r).$

\emph{The subregion $r_-(\ell)\leq r\leq -a(\ell)$}. By Proposition \ref{proposition_decay_estimate_isco}, for every geodesic in $r_-(\ell)\leq r\leq -a(\ell)$ with angular momentum $\ell\in (  2\sqrt{3}M,4M)$ and particle energy $E\in (\frac{2}{3},\frac{21\sqrt{2}}{30})$, we have 
\begin{align*}
|p^r(s)\,-\,p^{r,+}_{\ell}(r(s))|&\lesssim \dfrac{1}{v^{3}(s)},\qquad  \text{if} \quad p^r\geq 0,\\
|p^r(s)\,-\,p^{r,-}_{\ell}(r(s))|&\lesssim \dfrac{1}{v^{3}(s)},\qquad  \text{if} \quad p^r\leq 0.
\end{align*}
Thus, the differences $p^v\,-\,p^{v,+}_{\ell}$ and $p^v\,-\,p^{v,-}_{\ell}$ satisfy that
\begin{align*}
        |p^v(s)\,-\,p^{v,+}_{\ell}(r(s))|&\leq \dfrac{1}{2\Omega^2}|E\,-\,E_{-}(\ell)|+ \dfrac{1}{2\Omega^2}|p^r(s)\,-\,p^{r,+}_{\ell}(r(s))|\lesssim v^{-3}(s),\qquad \text{if} \quad p^r\geq 0,\\
        |p^v(s)\,-\,p^{v,-}_{\ell}(r(s))|&\leq \dfrac{1}{2\Omega^2}|E\,-\,E_{-}(\ell)|+ \dfrac{1}{2\Omega^2}|p^r(s)\,-\,p^{r,-}_{\ell}(r(s))|\lesssim v^{-3}(s),\qquad \text{if} \quad p^r\leq 0,
\end{align*}
where we have used the estimate for the particle energy in Proposition \ref{proposition_decay_estimate_isco}.

\emph{The subregion $5M\leq r\leq r_-(\ell)$}. By Proposition \ref{proposition_decay_estimate_isco}, for every geodesic in $5M\leq r\leq r_-(\ell)$ with angular momentum $\ell\in (  2\sqrt{3}M,4M)$ and particle energy $E\in (\frac{2}{3},\frac{21\sqrt{2}}{30})$, we have 
\begin{align*}
|p^r(s)\,-\,p^{r,+}_{\ell}(r(s))|&\lesssim v^{-3}(s).
\end{align*}
Thus, the difference $p^v\,-\,p^{v,+}_{\ell}$ satisfies that
\begin{align*}
        |p^v(s)\,-\,p^{v,+}_{\ell}(r(s))|&\leq \dfrac{1}{2\Omega^2}|E\,-\,E_{-}(\ell)|+ \dfrac{1}{2\Omega^2}|p^r(s)\,-\,p^{r,+}_{\ell}(r(s))|\lesssim v^{-3}(s),
\end{align*}
where we have used the estimate for the particle energy in Proposition \ref{proposition_decay_estimate_isco}.

The regions $r_0\leq r \leq 5M$ and $r\leq r_0$ can be taken care in the same fashion as in the proof of Proposition \ref{lem_isco_smaller_angu_mom_ingred1}. The estimates \eqref{estim_homoclinic_low_ang_farfar}--\eqref{estim_homoclinic_low_ang_bounded_neg_pr} follow by putting together the estimates performed in the five subregions considered above.
\end{proof}

Next, we will show concentration estimates in a neighbourhood of the parabolic manifolds at infinity. The leading contribution in the energy-momentum tensor in the bounded region comes from this region. Let us recall again that the coordinate $p^v$ can be written as 
\begin{equation}\label{ident_pv_basic_particle_radial_moment2}
p^v=\frac{1}{2\Omega^2}(E\,+\, p^r).
\end{equation}
We denote the coordinate $p^v$ of past-trapped and future-trapped geodesics with angular momentum $\ell$ and particle energy $E=1$, by $p^{v,+}_{\ell,1}$ and $p^{v,-}_{\ell,1}$, respectively. By \eqref{ident_pv_basic_particle_radial_moment2}, we deduce $$p^{v,\pm}_{\ell,1}(r)=\frac{1}{2\Omega^2}\Big(E\,+\,p^{r,\pm}_{\ell,1}(r)\Big),\quad \qquad \quad p^{r,\pm}_{\ell,1}(r)=\mp \sqrt{1\,-\, V_{\ell}(r)},$$ where we have used the mass-shell relation. We will now show a concentration estimate for $p^v$ using the bounds for $p^r$ in Propositions \ref{proposition_slow_decay_particle_energy_one_include_radial_geodesics} and \ref{proposition_slow_decay_geodesics_medium_angular_momentum_particle_one}. We begin considering the case when $\ell\leq 2\sqrt{3}M$.

\begin{lemma}\label{lemma_decay_pv1_coordinate_low_angular_momentum}
For every geodesic $\gamma_{x,p}\colon [0,a]\to \{r<R\}\,\cap\,\supp(f\chi_{\mathcal{D}})$ with angular momentum $\ell \in (0, 2\sqrt{3}M)$ and particle energy $E\in (\frac{21\sqrt{2}}{30},1)$, we have
\begin{align}
|p^v(s)\,-\,p^{v,-}_{\ell ,1}(r(s))|\lesssim \dfrac{1}{v^{\frac{1}{3}}(s)},\qquad \text{if}\quad p^r(s)\geq 0,\\
|p^v(s)\,-\,p^{v,+}_{\ell ,1}(r(s))|\lesssim \dfrac{1}{v^{\frac{1}{3}}(s)},\qquad \text{if}\quad p^r(s)\leq 0.\label{estim_concen_near_parabolic_incoming_small_ang_mom_bded}
\end{align}
\end{lemma}

\begin{proof}
We can assume with out loss of generality that every geodesic $\gamma$ is incoming in the bounded region $\{r<R\}$, by the compact support assumption on the initial data. Let $r_0>2M$ such that $r_0-2M$ is small enough. We divide the analysis in three different subregions: when $r\geq 16M$, when $r_0\leq r \leq 16M$, and when $r\leq r_0$. We address these subregions in the same order. 

\emph{The subregion $r\geq 16M$}. By Proposition \ref{proposition_slow_decay_particle_energy_one_include_radial_geodesics}, for every geodesic in $\{r\geq 16M\}$ with angular momentum $\ell\in (0, 2\sqrt{3}M)$ and particle energy $E\in (\frac{21\sqrt{2}}{30},1)$, we have $$|p^r(s)\,-\,p^{r,+}_{\ell,1}(r(s))|=\Big|p^r\,+\,\dfrac{\sqrt{2M}}{r^{\frac{3}{2}}}\Big(r^2-\dfrac{\ell^2}{2M}r+\ell^2\Big)^{\frac{1}{2}}\Big|\lesssim v^{-\frac{1}{3}}(s).$$ Thus, the difference $p^v\,-\,p^{v,+}_{\ell, 1}$ satisfies that
\begin{align*}
        |p^v(s)\,-\,p^{v,+}_{\ell, 1}(r(s))|&\leq \dfrac{1}{2\Omega^2}|E\,-\,1|+ \dfrac{1}{2\Omega^2}|p^r(s)\,-\,p^{r,+}_{\ell, 1}(r(s))|\lesssim v^{-\frac{1}{3}}(s),
\end{align*}
where we have used the $\frac{1}{2}$-Hölder continuity of the square root, and the estimate for the particle energy in Proposition \ref{proposition_slow_decay_particle_energy_one_include_radial_geodesics}. 

\emph{The subregion $r_0\leq r \leq 16M$.} The analysis in this subregion follows by the same arguments performed at the end of the proof of Lemma \ref{lemma_decay_pv_coordinate_large_angular_momentum}. We simply propagate to this region the concentration estimate obtained in the subregion $\{r>16M\}$.

\emph{The subregion $r\leq r_0$.} The analysis in this subregion follows by the same arguments performed at the end of the proof of Lemma \ref{lemma_decay_pv_coordinate_large_angular_momentum}. Specifically, we use the red-shift effect in the form presented in Proposition \ref{lemma_redshift}.

The estimate \eqref{estim_concen_near_parabolic_incoming_small_ang_mom_bded} follows by putting together the estimates performed in the three subregions considered above. In the region $r_0\leq r \leq 16M$, we also need the estimates in the region $r\geq 16M$. In the region $r\leq r_0$, we also need the estimates in the region $r_0\leq r \leq 16M$ and $r\geq 16M$.
\end{proof}

We will now consider the case when $\ell\sim 4M$. In this regime, the homoclinic orbits move towards infinity as $\ell \to 4M$. The homoclinic orbits become the stable manifolds of the spheres $\mathcal{S}^-(4M)$ and $\mathcal{S}^{\infty}(4M)$, in the limit when $\ell\to 4M$. For this reason, we need to control the unstable trapping and the parabolic trapping at infinity at the same time. 

\begin{lemma}\label{lemma_decay_pv1_coordinate_critic_angular_momentum}
For every geodesic $\gamma_{x,p}\colon [0,a]\to \{r<R\}\,\cap\,\supp(f\chi_{\mathcal{D}})$ with angular momentum $\ell \in (2\sqrt{3}M, 4M)$ and particle energy $E\in (\frac{21\sqrt{2}}{30},1)$, we have
\begin{align}
|p^v(s)\,-\,p^{v,+}_{\ell}(r(s))|&\lesssim \frac{1}{ v^{3}(s)},\qquad \text{if}\quad p^r\geq 0 \quad \text{and}\quad r\in [r_-(\ell),R],\label{estim_first_last_ingred_bounded_main}\\
|p^v(s)\,-\,p^{v,+}_{\ell ,1}(r(s))|+|p^v(s)\,-\,p^{v,-}_{\ell}(r(s))|&\lesssim \frac{1}{v^{\frac{1}{3}}(s)},\qquad \text{if}\quad p^r\leq 0 \quad \text{and}\quad  r\in [r_-(\ell),R],\\
|p^v(s)\,-\,p^{v,+}_{\ell ,1}(r(s))|+|p^v(s)\,-\,p^{v,+}_{\ell}(r(s))|&\lesssim \frac{1}{v^{\frac{1}{3}}(s)},\qquad \text{if}\quad   r\in [2M,r_-(\ell)].\label{estim_second_last_ingred_bounded_main}
\end{align}
\end{lemma}

\begin{proof}
By the assumption that $\ell \in (2\sqrt{3}M, 4M)$ and $E\in (\frac{21\sqrt{2}}{30},1)$, unstable trapping can only hold for $\ell\geq \ell_0$ for $\ell_0>2\sqrt{3}M$. In particular, there is no degenerate trapping at ISCO in this case. Also, we can assume without loss of generality that the geodesics we consider satisfy that $\ell\geq \ell_0>2\sqrt{3}M$. The geodesics in the complementary region either cross $\mathcal{H}^+$ or leave the region $\{r<R\}$, before the advanced time $v$ is sufficiently large.

In this region of phase space, there are homoclinic orbits for every $\ell\in (\ell_0,4M)$. For these homoclinic orbits we have that $-a(\ell)\geq A$, where $A>2M$ is the constant set in Proposition \ref{prop_conc_estim_homocl_bounded_angu_mom}. Under our assumptions, we can take $A>R$, so the turning points at $-a(\ell)$ take place in the far-away region. In other words, the geodesics we consider here are always ingoing or outgoing in the bounded region.

Let $r_0>2M$ such that $r_0-2M$ is small enough. Let $\ell\in (\ell_0,4M)$. We divide the analysis in four different subregions: when $r_-(\ell)\leq r \leq R$, when $5M\leq r \leq r_-(\ell)$, when $r_0\leq r \leq 5M$, and when $r\leq r_0$. We address these subregions in the same order.

\emph{The subregion $r_-(\ell)\leq r \leq R$}. By Proposition \ref{prop_conc_estim_homocl_bounded_angu_mom} and the Remark \ref{remark_control_homoclinic_critical_angular_mom_bounde}, for every geodesic in $r_-(\ell)\leq r\leq R$ with angular momentum $\ell\in (  2\sqrt{3}M,4M)$ and particle energy $E\in (\frac{2}{3},\frac{21\sqrt{2}}{30})$, we have 
\begin{align*}
|p^r(s)\,-\,p^{r,+}_{\ell}(r(s))|&\lesssim \dfrac{1}{v^{3}(s)},\qquad  \text{if} \quad p^r\geq 0,\\
|p^r(s)\,-\,p^{r,-}_{\ell}(r(s))|&\lesssim \dfrac{1}{v^{3}(s)},\qquad  \text{if} \quad p^r\leq 0.
\end{align*}
Thus, the differences $p^v\,-\,p^{v,+}_{\ell}$ and $p^v\,-\, p^{v,-}_{\ell}$ satisfy that
\begin{align*}
        |p^v(s)\,-\,p^{v,+}_{\ell}(r(s))|&\leq \dfrac{1}{2\Omega^2}|E\,-\,E_{-}(\ell)|+ \dfrac{1}{2\Omega^2}|p^r(s)\,-\,p^{r,+}_{\ell}(r(s))|\lesssim v^{-3}(s),\qquad \text{if} \quad p^r\geq 0,\\
        |p^v(s)\,-\,p^{v,-}_{\ell}(r(s))|&\leq \dfrac{1}{2\Omega^2}|E\,-\,E_{-}(\ell)|+ \dfrac{1}{2\Omega^2}|p^r(s)\,-\,p^{r,-}_{\ell}(r(s))|\lesssim v^{-3}(s),\qquad \text{if} \quad p^r\leq 0,
\end{align*}
where we have used the estimate for the particle energy in Proposition \ref{prop_conc_estim_homocl_bounded_angu_mom} and Remark \ref{remark_control_homoclinic_critical_angular_mom_bounde}.

In this subregion, we also need to control the geodesic flow for $E\sim 1$. This analysis is required because, as we will show in Lemma \ref{lem_far_away_parabolic_closing_high_angu}, there are incoming particles with $E\sim 1$ that enter this region. By Proposition \ref{prop_conc_estim_homocl_bounded_angu_mom} and the Remark \ref{remark_control_homoclinic_critical_angular_mom_bounde}, for every geodesic in $r_-(\ell)\leq r\leq R$ with angular momentum $\ell\in (  2\sqrt{3}M,4M)$ and particle energy $E\in (\frac{2}{3},\frac{21\sqrt{2}}{30})$, we have 
\begin{align*}
|p^r(s)\,-\,p^{r,+}_{\ell}(r(s))|&\lesssim v^{-3}(s).
\end{align*}
Thus, the difference $p^v\,-\,p^{v,+}_{\ell}$ satisfies that
\begin{align*}
        |p^v(s)\,-\,p^{v,+}_{\ell}(r(s))|&\leq \dfrac{1}{2\Omega^2}|E\,-\,E_{-}(\ell)|+ \dfrac{1}{2\Omega^2}|p^r(s)\,-\,p^{r,+}_{\ell}(r(s))|\lesssim v^{-3}(s)
\end{align*}
where we have used the estimate for the particle energy in Proposition \ref{prop_conc_estim_homocl_bounded_angu_mom} and Remark \ref{remark_control_homoclinic_critical_angular_mom_bounde}.

\emph{The subregion $5M\leq r \leq r_-(\ell)$}. The analysis in this subregion follows by the same arguments performed in the bounded region, but considering only the unstable manifold of the sphere of trapped orbits $\mathcal{S}^-(\ell).$

\emph{The subregion $r_0\leq r \leq 5M$.} The analysis in this subregion follows by the same arguments performed at the end of the proof of Lemma \ref{lemma_decay_pv_coordinate_large_angular_momentum}. We simply propagate to this region the concentration estimate obtained in the subregion $\{r>5M\}$.

\emph{The subregion $r\leq r_0$.} The analysis in this subregion follows by the same arguments performed at the end of the proof of Lemma \ref{lemma_decay_pv_coordinate_large_angular_momentum}. Specifically, we use the red-shift effect in the form presented in Proposition \ref{lemma_redshift}.

The estimates \eqref{estim_first_last_ingred_bounded_main}--\eqref{estim_second_last_ingred_bounded_main} follow by putting together the estimates performed in the three subregions considered above. In the region $5M\leq r \leq r_-(\ell)$, we also need the estimates in the region $r\geq r_-(\ell)$. In the region $r_0\leq r\leq 5M$, we also need the estimates in the regions where $5M\leq r \leq r_-(\ell)$ and $r\geq r_-(\ell)$. In the region $r\leq r_0$, we need the estimates in all the previous regions.
\end{proof}

We can now prove the estimates of Theorem \ref{theorem_decay_slow} in the bounded region of spacetime.\\

\emph{Proof of Theorem \ref{theorem_decay_slow} in the bounded region}. By Lemma \ref{lem_isco_smaller_angu_mom_ingred1}, Lemma \ref{lem_isco_larger_angu_mom_ingred2}, Lemma \ref{lemma_decay_pv1_coordinate_low_angular_momentum}, and Lemma \ref{lemma_decay_pv1_coordinate_critic_angular_momentum}, the support of the distribution function in the coordinate $p^v$ decays like $v^{-\frac{1}{3}}$. Therefore, for all $x\in \{r<R\}$, the component $\T_{uv}[f_{< 1}\chi_{\mathcal{D}}]$ of the energy-momentum tensor satisfies
\begin{align*}
    \frac{1}{\Omega^2}\T_{uv}[f_{< 1}\chi_{\mathcal{D}}]&=\frac{1}{\Omega^2}\int_{p^v}\int_{p^A}\int_{p^B} (f_{<1}\chi_{\mathcal{D}})(x,p)p_{u} p_{v}\dfrac{r^2\sqrt{\det \gamma}}{p^v}\d p^v\d p^A\d p^B\\
    &=\dfrac{1}{4\Omega^2}\int_{p^v}\int_{p^A}\int_{p^B}(f_{<1}\chi_{\mathcal{D}})(x,p)(\Omega^2p^{u}) (\Omega^2p^{v})\dfrac{r^2\sqrt{\det \gamma}}{p^v}\d p^v\d p^A\d p^B\\
    &\lesssim \dfrac{\|f_0\|_{L^{\infty}_{x,p}}}{v^{\frac{1}{3}}}\int_{r^4(g_{\S^2})_{AB}p^Ap^B\leq L_1} r^2 \d p^A\d p^B\lesssim \dfrac{\|f_0\|_{L^{\infty}_{x,p}}}{v^{\frac{1}{3}}},
\end{align*}
where we have estimated $\Omega^2p^u$ using Lemma \ref{lemma_asymptotic_tangent_at_horizon}. Finally, we use the estimates obtained in the previous two sections for $f_{\geq 1}\chi_{\mathcal{D}}$, in order to show that $$\frac{1}{\Omega^2}\T_{uv}[f\chi_{\mathcal{D}}]=\frac{1}{\Omega^2}\T_{uv}[f_{\geq 1}\chi_{\mathcal{D}}]+\frac{1}{\Omega^2}\T_{uv}[f_{< 1}\chi_{\mathcal{D}}]\lesssim  \dfrac{\|f_0\|_{L^{\infty}_{x,p}}}{u^{\frac{1}{3}}r^2}.$$ The same argument shows the decay estimates for the other components of $\mathrm{T}_{\mu\nu}[f\chi_{\mathcal{D}}]$.

\subsubsection{Estimates in the far-away region}

Let us show concentration estimates in a neighbourhood of the parabolic manifolds. We start with the case when $\ell\leq 2\sqrt{3}M$.

\begin{lemma}\label{lem_far_away_parabolic_closing_low_angu}
For every geodesic $\gamma_{x,p}\colon [0,a]\to \{r>R\}\,\cap\,\supp(f\chi_{\mathcal{D}})$ with angular momentum $\ell \in (0, 2\sqrt{3}M)$ and particle energy $E\in (\frac{21\sqrt{2}}{30},1)$, we have
\begin{align}
|p^v(s)\,-\,p^{v,-}_{\ell ,1}(r(s))|&\lesssim \dfrac{1}{u^{\frac{1}{3}}(s)},\qquad \text{if}\quad p^r(s)\geq 0,\\
|p^v(s)\,-\,p^{v,+}_{\ell ,1}(r(s))|&\lesssim \dfrac{1}{u^{\frac{1}{3}}(s)},\qquad \text{if}\quad p^r(s)\leq 0.\label{estim_concen_near_parabolic_incoming_small_ang_mom}
\end{align}
\end{lemma}

\begin{proof}
We divide the analysis in two different subregions: when $p^r\geq 0$ and when $p^r\leq 0$. We address these subregions in the same order. 

\emph{The subregion $p^r\geq 0$.} By Proposition \ref{proposition_slow_decay_particle_energy_one_include_radial_geodesics}, for every geodesic in $\{r>R\}$ with angular momentum $\ell\leq 2\sqrt{3}M$ and particle energy $E\in (\frac{21\sqrt{2}}{30},1)$, we have $$|p^r(s)\,-\, p^{r,-}_{\ell,1}(r(s))|=\Big|p^r\,-\,\dfrac{\sqrt{2M}}{r^{\frac{3}{2}}}\Big(r^2-\dfrac{\ell^2}{2M}r+\ell^2\Big)^{\frac{1}{2}}\Big|\lesssim  u^{-\frac{1}{3}}(s).$$ Thus, the difference $p^v\,-\, p^{v,-}_{\ell, 1}$ satisfies that
\begin{align*}
        |p^v(s)\,-\,p^{v,-}_{\ell, 1}(r(s))|&\leq \dfrac{1}{2\Omega^2}|E\,-\,1|+ \dfrac{1}{2\Omega^2}|p^r(s)\,-\,p^{r,-}_{\ell, 1}(r(s))|\lesssim u^{-\frac{1}{3}}(s),
\end{align*}
where we have used the $\frac{1}{2}$-Hölder continuity of the square root, and the estimate for the particle energy in Proposition \ref{proposition_slow_decay_particle_energy_one_include_radial_geodesics}. 

\emph{The subregion $p^r\leq 0$.} The estimate \eqref{estim_concen_near_parabolic_incoming_small_ang_mom} for incoming geodesics is obtained similarly. In this case, we estimate the difference $p^v(s)\,-\,p^{v,+}_{\ell, 1}(r(s))$ instead of $p^v(s)\,-\,p^{v,-}_{\ell, 1}(r(s))$. 

The estimate \eqref{estim_concen_near_parabolic_incoming_small_ang_mom} follows by putting together the estimates performed in the two subregions considered. We note that the constant in the RHS of \eqref{estim_concen_near_parabolic_incoming_small_ang_mom} is uniform among the geodesics under study.
\end{proof}

We will now consider the case when $\ell\sim 4M$. In this regime, we deal with the parabolic trapping at infinity. In the presence of the homoclinic orbits, we can still perform the analysis in the far-away region as in Lemma \ref{lem_far_away_parabolic_closing_low_angu}.

\begin{lemma}\label{lem_far_away_parabolic_closing_high_angu}
For every geodesic $\gamma_{x,p}\colon [0,a]\to \{r>R\}\,\cap\,\supp(f\chi_{\mathcal{D}})$ with angular momentum $\ell \in (2\sqrt{3}M, 4M)$ and particle energy $E\in (\frac{21\sqrt{2}}{30},1)$, we have
\begin{align}
|p^v(s)\,-\,p^{v,-}_{\ell ,1}(r(s))|&\lesssim \dfrac{1}{u^{\frac{1}{3}}(s)},\qquad \text{if}\quad p^r\geq 0\quad \text{and}\quad r\in[R, \infty),\label{estim_last_ingred_far_away_main1}\\
|p^v(s)\,-\,p^{v,+}_{\ell ,1}(r(s))|&\lesssim \dfrac{1}{u^{\frac{1}{3}}(s)},\qquad \text{if}\quad p^r\leq 0\quad \text{and}\quad r\in [R, \infty).\label{estim_last_ingred_far_away_main2}
\end{align}
\end{lemma}

\begin{proof}
By the assumption that $\ell \in (2\sqrt{3}M, 4M)$ and $E\in (\frac{21\sqrt{2}}{30},1)$, unstable trapping can only hold for $\ell\geq \ell_0$ for $\ell_0>2\sqrt{3}M$. Also, we can assume without loss of generality that the geodesics we consider satisfy that $\ell\geq \ell_0>2\sqrt{3}M$. 

In this region of phase space, there are homoclinic orbits for every $\ell\in [\ell_0,4M)$. For these homoclinic orbits we have that $-a(\ell)\geq A$, where $A>2M$ is the constant set in Proposition \ref{prop_conc_estim_homocl_bounded_angu_mom}. We can assume that $A>R$, so the turning points at $-a(\ell)$ take place in the far-away region. 

Let $\ell\in [\ell_0,4M)$. We divide the analysis in two different subregions: when $p^r\geq 0$ and when $p^r\leq 0$. We address these subregions in the same order. 

\emph{The subregion $p^r\geq 0$.} By Proposition \ref{proposition_slow_decay_geodesics_medium_angular_momentum_particle_one}, for every geodesic in $\{r>R\}$ with angular momentum $\ell \in (2\sqrt{3}M, 4M)$ and particle energy $E\in (\frac{21\sqrt{2}}{30},1)$, we have $$|p^r(s)\,-\,p^{r,-}_{\ell,1}(r(s))|=\Big|p^r\,-\,\dfrac{\sqrt{2M}}{r^{\frac{3}{2}}}\Big(r^2-\dfrac{\ell^2}{2M}r+\ell^2\Big)^{\frac{1}{2}}\Big|\lesssim u^{-\frac{1}{3}}(s).$$ Thus, the difference $p^v\,-\,p^{v,-}_{\ell, 1}$ satisfies that
\begin{align*}
        |p^v(s)\,-\,p^{v,-}_{\ell, 1}(r(s))|&\leq \dfrac{1}{2\Omega^2}|E\,-\,1|+ \dfrac{1}{2\Omega^2}|p^r(s)\,-\,p^{r,-}_{\ell, 1}(r(s))|\lesssim u^{-\frac{1}{3}}(s),
\end{align*}
where we have used the $\frac{1}{2}$-Hölder continuity of the square root, and the estimate for the particle energy in Proposition \ref{proposition_slow_decay_geodesics_medium_angular_momentum_particle_one}. 

\emph{The subregion $p^r\leq 0$.} The estimate \eqref{estim_concen_near_parabolic_incoming_small_ang_mom} for incoming geodesics is obtained similarly. In this case, we estimate the difference $p^v(s)\,-\,p^{v,+}_{\ell, 1}(r(s))$ instead of $p^v(s)\,-\,p^{v,-}_{\ell, 1}(r(s))$.  
\end{proof}

We conclude this section with the proof of Theorem \ref{theorem_decay_slow} in the far-away region.\\

\emph{Proof of Theorem \ref{theorem_decay_slow} in the far-away region}. By Lemma \ref{lem_far_away_parabolic_closing_low_angu} and Lemma \ref{lem_far_away_parabolic_closing_high_angu}, the support of the distribution function in the coordinate $p^v$ decays like $u^{-\frac{1}{3}}$. Then, for all $x\in \{r>R\}$, the component $\T_{uv}[f_{< 1}\chi_{\mathcal{D}}]$ of the energy-momentum tensor satisfies
\begin{align*}
    \T_{uv}[f_{< 1}\chi_{\mathcal{D}}]&=\int_{p^v}\int_{p^A}\int_{p^B} (f_{< 1}\chi_{\mathcal{D}})(x,p)p_{u} p_{v}\dfrac{r^2\sqrt{\det \gamma}}{p^v}\d p^v\d p^A\d p^B\\
    &=\dfrac{1}{4}\int_{p^v}\int_{p^A}\int_{p^B} (f_{< 1}\chi_{\mathcal{D}})(x,p)(\Omega^2p^{u}) (\Omega^2p^{v})\dfrac{r^2\sqrt{\det \gamma}}{p^v}\d p^v\d p^A\d p^B\\
    &\lesssim \dfrac{\|f_0\|_{L^{\infty}_{x,p}}}{u^{\frac{1}{3}}} \int_{r^4(g_{\S^2})_{AB}p^Ap^B\leq L_1} r^2\d p^A\d p^B\lesssim \dfrac{\|f_0\|_{L^{\infty}_{x,p}}}{u^{\frac{1}{3}}r^2},
\end{align*}
where we have used the previous bound on the momentum coordinate $p^v$ along the geodesics with particle energy $E< 1$ in the far-away region. As a result, we can use the estimate \eqref{decay_estimatefgeq1faraway} for $f_{\geq 1}$, to show that $$\T_{uv}[f\chi_{\mathcal{D}}]=\T_{uv}[f_{\geq 1}\chi_{\mathcal{D}}]+\T_{uv}[f_{< 1}\chi_{\mathcal{D}}]\lesssim  \dfrac{\|f_0\|_{L^{\infty}_{x,p}}}{u^{\frac{1}{3}}r^2}.$$ The same argument shows the decay estimates for the other components of $\mathrm{T}_{\mu\nu}[f\chi_{\mathcal{D}}]$.


\bibliographystyle{alpha}
\bibliography{Bibliography.bib} 

\end{document}